\documentclass[12pt]{article}
\usepackage{amsmath}
\usepackage{graphicx}
\usepackage{natbib}

\usepackage{hyperref}
\usepackage{url} 
\usepackage{float}

\usepackage{multibib}
\newcites{sup}{References Supplement}

\newcommand{\blind}{0}

\usepackage{enumitem}
\usepackage{dsfont}
\usepackage{amsmath}
\usepackage{amsfonts}
\usepackage{amsthm}
\usepackage{nicefrac}
\usepackage[svgnames]{xcolor}

\usepackage[ruled]{algorithm2e}

\newcommand{\C}{\mathds{C}}

\newcommand{\R}{\mathds{R}}

\renewcommand{\S}{\mathds{S}}

\newcommand{\T}{\top}
\renewcommand{\vec}[1]{\big(#1\big)^\T}

\newcommand{\bB}{\mathbf{B}}
\newcommand{\bC}{\mathbf{C}}
\newcommand{\bD}{\mathbf{D}}
\newcommand{\bE}{\mathbf{E}}
\newcommand{\be}{\mathbf{e}}

\newcommand{\bG}{\mathbf{G}}

\newcommand{\bh}{\mathbf{h}}

\newcommand{\bM}{\mathbf{M}}
\newcommand{\bP}{\mathbf{P}}
\newcommand{\bp}{\mathbf{p}}
\newcommand{\bQ}{\mathbf{Q}}
\newcommand{\bU}{\mathbf{U}}
\newcommand{\bu}{\mathbf{u}}
\newcommand{\bR}{\mathbf{R}}

\newcommand{\bb}{\mathbf{b}}
\newcommand{\bZ}{\mathbf{Z}}

\newcommand{\bW}{\mathbf{W}}

\newcommand{\bV}{\mathbf{V}}
\newcommand{\bx}{\mathbf{x}}
\newcommand{\bX}{\mathbf{X}}
\newcommand{\by}{\mathbf{y}}
\newcommand{\bI}{\mathbf{I}}
\newcommand{\bmu}{\boldsymbol{\mu}}

\newcommand{\bTheta}{\boldsymbol{\Theta}}
\newcommand{\bXi}{\boldsymbol{\Xi}}
\newcommand{\bbeta}{\boldsymbol{\beta}}

\newcommand{\bPsi}{\boldsymbol{\Psi}}
\newcommand{\bpsi}{\boldsymbol{\psi}}

\newcommand{\cA}{\mathcal{A}}
\newcommand{\cB}{\mathcal{B}}
\newcommand{\cF}{\mathcal{F}}
\newcommand{\cM}{\mathcal{M}}
\newcommand{\cT}{\mathcal{T}}

\newcommand{\cX}{\mathcal{X}}
\newcommand{\cY}{\mathcal{Y}}
\newcommand{\cL}{\mathcal{L}}
\newcommand{\cU}{\mathcal{U}}
\newcommand{\cV}{\mathcal{V}}

\newcommand{\E}{\mathds{E}}

\newcommand{\argmin}[1][ ]{\underset{#1}{\operatorname{argmin}\;}}

\newcommand{\tr}{\operatorname{tr}}

\newcommand{\Exp}{\operatorname{Exp}}
\newcommand{\Log}{\operatorname{Log}}

\newcommand{\Transp}{\operatorname{Transp}}

\newcommand{\conj}[1]{{#1}^{\dagger}}
\renewcommand{\i}{\sqrt{\text{-}1}}
\newcommand{\ind}{\mathds{1}}

\DeclareFontFamily{OT1}{mathc}{}
\DeclareFontShape{OT1}{mathc}{m}{n}{ <-> mathc10 }{}
\DeclareRobustCommand\mycal[1]{\text{\usefont{OT1}{mathc}{m}{n}#1}}

\usepackage[normalem]{ulem}
\newcommand{\old}[1]{{\color{lightgray} #1}}
\newcommand{\new}[1]{{\color{blue} #1}}

\renewcommand{\old}[1]{} 
\renewcommand{\new}[1]{{#1}}

\usepackage{xr} 
\makeatletter
\newcommand*{\addFileDependency}[1]{
	\typeout{(#1)}
	\@addtofilelist{#1}
	\IfFileExists{#1}{}{\typeout{No file #1.}}
}
\makeatother

\newcommand*{\myexternaldocument}[1]{%
	\externaldocument{#1}%
	\addFileDependency{#1.tex}%
	\addFileDependency{#1.aux}%
}
\myexternaldocument{shapeboost-supplement} 

\begin{document}

\def\spacingset#1{\renewcommand{\baselinestretch}%
{#1}\small\normalsize} \spacingset{1}


\if0\blind
{
  \title{\bf Functional additive models on manifolds of planar shapes and forms}
  \author{Almond St\"ocker\new{, Lisa Steyer,} 
    and 
    Sonja Greven \\
    School of Business and Economics, Humboldt-Universität zu Berlin}
  \maketitle
} \fi

\if1\blind
{
  \bigskip
  \bigskip
  \bigskip
  \begin{center}
    {\LARGE\bf Functional additive models on manifolds of planar shapes and forms}
\end{center}
  \medskip
} \fi

\bigskip
\begin{abstract}
\new{The ``shape'' of a planar curve and/or landmark configuration is considered its equivalence class under translation, rotation and scaling, its ``form'' its equivalence class under translation and rotation while scale is preserved.}
\old{Defining shape and form as equivalence classes under translation, rotation and -- for shapes -- also scale, w}\new{W}e extend generalized additive regression to models for \new{such shapes/forms as responses respecting }\old{the shape/form of planar curves and/or landmark configurations. The model respects }the resulting quotient geometry \new{by}\old{of the response,} employing the squared geodesic distance as loss function and a geodesic response function to map the additive predictor to the shape/form space. For fitting the model, we propose a Riemannian $L_2$-Boosting algorithm well suited for a potentially large number of possibly  parameter-intensive model terms, which also yields automated model selection. We provide novel intuitively interpretable visualizations for (even  non-linear) covariate effects in the shape/form space via suitable tensor-product
factorization. 
The usefulness of the proposed framework is illustrated in an analysis of 1) astragalus shapes of wild and domesticated sheep and 2) cell forms generated in a biophysical model, as well as 3) in a realistic simulation study with response shapes and forms motivated from a dataset on bottle outlines.
\end{abstract}

\noindent%
\vfill

\newpage
\spacingset{1.5} 

\section{Introduction}

In many imaging data problems, the coordinate system of recorded objects is arbitrary or explicitly not of interest. 
Statistical shape analysis \citep{DrydenMardia2016ShapeAnalysisWithApplications} addresses this point by identifying the ultimate object of analysis as the \textit{shape} of an observation, reflecting its geometric properties invariant under translation, rotation and re-scaling, or as its \textit{form} (or \textit{size-and-shape}) invariant under translation and rotation. 
This paper establishes a flexible additive regression framework for modeling the shape or form of planar (potentially irregularly sampled) curves and/or landmark configurations in dependence on scalar covariates. 
\old{The proposed component-wise Riemannian $L_2$-Boosting algorithm allows estimation of a large number of parameter-intense model terms with inherent model selection.  
We further introduce a novel visualization of functional additive effects using a tensor-product factorization, 
which we consider essential for practical interpretation.}
%
A  rich \new{shape analysis} literature \old{on statistical shape analysis }has been developed for 2D or 3D landmark configurations -- presenting for instance selected points of a bone or face -- which are considered elements of Kendall's shape space \citep[see, e.g.][]{DrydenMardia2016ShapeAnalysisWithApplications}.  
In many 2D scenarios, however, observed points describe a curve reflecting the outline of an object rather than dedicated landmarks \citep{AdamsRholfSlice2013GeomMorphometrics21stCentury}. 
Considering outlines as images of (parameterized) curves shows a direct link to functional data analysis \citep[FDA,][]{RamsaySilverman2005} and, 
in this context, we speak of functional shape/form data analysis. 
As in FDA, functional shape/form data can be observed on a common and often dense grid (\textit{regular/dense} design) or on curve-specific often sparse grids (\textit{irregular/sparse} design).
While in the regular case, analysis often simplifies by treating curve evaluations as multivariate data, more general irregular designs gave rise to further developments in sparse FDA \citep[e.g.][]{Yao:etal:2005, GrevenScheipl2017}, explicitly considering irregular measurements instead of pre-smoothing curves. 
To the best of our knowledge, 
we are the first to consider 
irregular/sparse designs in the context of 
functional shape/form analysis. 

Shapes and forms are examples of manifold data. 
\new{\citet{petersen2019frechet} propose ``Fréchet regression'' for random elements in general metric spaces, which requires estimation of a (potentially negatively) weighted Fréchet mean for each covariate combination. 
	Their implicit rather then explicit model formulation 
	renders model interpretation 
	difficult. 
	More explicit model formulations have been developed for the special case of a Riemannian geometry.  Besides tangent space models \citep{kent2001functional}, extrinsic models \citep{lin2017extrinsic} and models based on unwrapping \citep{jupp1987fitting, mallasto2018wrapped}, a variety of manifold regression models have been designed based on the intrinsic Riemannian geometry.} Starting from geodesic regression \citep{Fletcher2012GeodesicRegressionRiemannianManifold
}, which extends linear regression to curved spaces\new{, these include} \old{,several generalizations of standard regression models to manifold responses have been developed, such as} MANOVA \citep{Huckemann2010intrinsicMANOVA}, polynomial regression \citep{Hinkle2014IntrinsicPolynomials}, \new{smoothing splines \citep{kume2007shape}}, 
regression along geodesic paths with non-constant speed \citep{HongNiethammer2014TimeWarpedGeodesicRegression}, or kernel regression \citep{Davis2010ManifoldKernelRegression} \new{and Kriging \citep{pigoli2016kriging}}. \old{While the 
formulations are more general, landmark shape spaces often serve as an example.} 
However, \new{mostly} only one metric covariate or categorical covariates are considered\new{, possibly in hierarchical model extensions for longitudinal data \citep{muralidharan2012sasaki, schiratti2017bayesian}}. By contrast, \citet{zhu2009intrinsic,shi2009intrinsic,KimEtAl2014RiemannGLM} generalize geodesic regression to regression with multiple covariates \new{focusing on symmetric positive-definite (SPD) matrix responses}\old{but with symmetric positive-definite (SPD) matrices as response rather than shapes}. \citet{CorneaEtAl2017RegRiemannianSymSpaces} 
develop a general generalized linear model (GLM) analogue regression framework for responses in a symmetric \new{manifold} \old{space -- a class of Riemannian manifolds containing, in particular, Kendall's shape space --} and apply it to shape analysis. Recently, \citet{lin2020AdditiveModelsSPD} proposed a Lie group additive regression model for Riemannian manifolds focusing \old{however} on SPD matrices \old{and not on} \new{rather than} shapes.

In FDA, there is a much wider range of developed regression methods \citep[see overviews in][]{Morris2015,GrevenScheipl2017}. 
Among the most flexible models \new{are functional additive models (FAMs)}
for (univariate) functional responses \new{(in contrast to FAMs with functional covariates \citep{ferraty2011recent}) with different strategies existing to model a) response functions and b) smooth covariate effects. For a), basis expansions in spline \citep{BrockhausGreven2015}, functional principal component (FPC) bases \citep{MorrisCarroll2006} or both \citep{Scheipl2015} are employed as well as wavelets \citep{meyer2015bayesian}, sometimes directly expanding functions to model on coefficients and sometimes expanding only predictions while keeping the raw measurements. Other approaches effectively evaluate curves on grids or apply pre-smoothing techniques instead \citep[e.g.,][]{jeon2020additiveHilbertian}.
For b), again penalized spline basis approaches are employed \citep{Scheipl2015, BrockhausGreven2015}, or local linear/polynomial \citep{MullerYao2008FAM, jeon2022locally} or other kernel-based approaches \citep{jeon2020additiveHilbertian, jeon2021additiveNonEuclidean}.
The different approaches come with different theoretical and practical advantages, but similiarities such as regarding asymptotic behavior   are also  known  from scalar nonparametric regression  \citep{LiRuppert2008asymptotics}.
Advantages of the fully basis expansion based approach summarized in \cite{GrevenScheipl2017} include its appropriateness for sparse irregular functional data and its modular extensibility to functional mixed models \citep{Scheipl2015, meyer2015bayesian} and non-standard response distributions \citep{BrockhausGreven2015, Stoecker2019FResponseLSS}.
}\old{
are the functional additive (mixed) models (FA(M)Ms) of  i.a.\ \cite{MorrisCarroll2006,meyer2015bayesian}  for dense functional data 
and the flexible FAMM model class also covering sparse irregular functional data, summarized in \cite{GrevenScheipl2017}.
}
\old{There is less work on regression f}\new{F}or bivariate or multivariate functional responses, which are closest to functional shapes/forms but without   invariances\old{.}\new{,}  
\cite{rosen2009bayesian,zhu2012multivariate,olsen2018simultaneous} consider linear fixed 
effects of scalar covariates, the latter also allowing for warping. 
\cite{
	ZhuEtal2017multiFAMM,BackenrothEtal2018HeteroFPCA} consider one or more random effects for one grouping variable, linear fixed effects and common dense grids for all functions. 
\cite{Volkmann2020MultiFAMM} combine the FA\old{M}M model class of \cite{GrevenScheipl2017} with multivariate \old{functional principal component}\new{FPC} analysis 
\citep{HappGreven2018} to \new{model}\old{obain FA\old{M}Ms for} multivariate (sparse) functional responses.

 \old{For estimation of FAMMs, \citet{BrockhausGreven2015} employ a (component-wise) model-based gradient boosting algorithm. Model-based boosting \citep{HothornBuehlmann2010} is a stage-wise fitting procedure for estimating additive  models with inherent model selection and a slow over-fitting behavior, allowing to efficiently estimate models with a very large number of coefficients. 
Regularization is applied both in each step and via early stopping of the algorithm based, e.g., on cross-validation. The combination of high auto-correlation within functional responses and complex tensor-product covariate effects typical for FAMMs make the double regularization particularly beneficial in this context \citep{Stoecker2019FResponseLSS}. Gradient boosting with respect to (w.r.t.) a least-squares loss is typically referred to as $L^2$-Boosting \citep{Buehlmann2003L2boosting}. 
We introduce a \textit{Riemannian $L^2$-Boosting} algorithm for estimating interpretable additive regression models for response shapes/forms of planar curves.} 

\new{This paper establishes an interpretable FAM framework for modeling the shape or form of planar (potentially irregularly sampled) curves and/or landmark configurations in dependence on scalar covariates, extending $L_2$-Boosting \citep{Buehlmann2003L2boosting, BrockhausGreven2015} to Riemannian manifolds for model estimation.}
The three major contributions of our regression framework are: 
1.~We introduce additive regression with shapes/forms of planar curves and/or landmarks as response, extending FAMs to non-linear response spaces or, vice versa, extending GLM-type regression on manifolds for landmark shapes both to functional shape manifolds and to include 
(non-linear) additive model effects. 
2.~We propose a novel Riemannian $L_2$-Boosting algorithm for estimating regression models for this type of manifold response, 
and 3.~a visualization technique based on tensor-product factorization yielding intuitive  interpretations even of multi-dimensional smooth covariate effects 
for practitioners.
Although related tensor-product model transformations based on higher-order SVD have been used, i.a., in control engineering \citep{Baranyi2013tensor}, 
we are not aware of any comparable application for visualization in FAMs or other statistical models for object data. 
Despite our focus on shapes and forms, transfer of the model, Riemannian $L_2$-Boosting, and 
factorized visualization to other Riemannian manifold responses is intended in the generality of the formulation and the design of the provided R package \texttt{manifoldboost}\if0\blind{ (developer version on \url{github.com/Almond-S/manifoldboost})}\fi.
\old{Moreover, while considering potential invariance w.r.t.\ re-parameterization (``warping'') of curves is beyond the scope of this paper, the presented framework is intended to be combined with an ``elastic'' approach in the square root velocity framework \citep{SrivastavaKlassen2016} in the future.}
The versatile applicability of the approach is illustrated in three different scenarios: 
an analysis of the shape of sheep astragali (ankle bones) represented by both regularly sampled curves and landmarks in dependence on categorical ``demographic'' variables; 
an analysis of the effects of different metric biophysical model parameters (including smooth  interactions) on the form of (irregularly sampled) cell outlines generated from a cellular Potts model; 
and a simulation study with irregularly sampled functional shape and form responses generated from a dataset of different bottle outlines and including metric and categorical covariates.

In Section \ref{chap_diffgeo}, we introduce the manifold geometry of irregular curves modulo translation, rotation and potentially re-scaling, which underlies the intrinsic additive regression model formulated in Section \ref{chap_model}. The Riemannian $L^2$-Boosting algorithm is introduced in Section \ref{chap_boosting}. Section \ref{chap_application} analyzes different data problems,   modeling sheep bone shape responses (Section \ref{sec_sheep}) and cell outlines (Section \ref{sec_cells}). Section \ref{sec_simulation} summarizes the results of simulation studies with functional shape and form responses. We conclude  with a discussion in Section \ref{chap_discussion}.



\section{Geometry of functional forms and shapes}
\label{chap_diffgeo}

\newcommand{\re}[1]{\operatorname{Re}\!\left(#1\right)}
\newcommand{\im}[1]{\operatorname{Im}\!\left(#1\right)}
\newcommand{\dir}{\times} 

Riemannian manifolds of planar shapes (and forms) 
are discussed in various textbooks at different levels of generality, in finite \citep{DrydenMardia2016ShapeAnalysisWithApplications, Kendall1999} or potentially infinite dimensions  \citep{SrivastavaKlassen2016,Klingenberg1995RiemannianGeometry}. 
Starting from the Hilbert space $\cY$ of curve representatives $y$ of a single shape or form observation, we successively characterize its quotient space geometry under translation, rotation and re-scaling including the respective tangent spaces. Building on that, we introduce Riemannian exponential and logarithmic maps and parallel transports needed for model formulation and fitting, and the sample space of (irregularly observed) functional shapes/forms.


\newcommand{\trl}{\operatorname{Trl}}
\newcommand{\scl}{\operatorname{Scl}}
\newcommand{\rot}{\operatorname{Rot}}
\newcommand{\neutral}{\!\mycal{1}}
\newcommand{\zero}{\mycal{0}}

To make use of complex arithmetic, we identify  the two-dimensional plane with the complex numbers, $\R^2 \cong \C$, and consider a planar curve to be a function $y: \R \supset \cT \rightarrow \C$, element of a separable complex Hilbert space $\cY$ with a complex inner product $\langle \cdot,\cdot\rangle$ and corresponding norm $\|\cdot\|$. 
This allows simple scalar expressions for
the group actions of translation $\trl = \{y \overset{\trl_\gamma}{\longmapsto} y + \gamma\, \neutral : \gamma\in\C\}$ with $\neutral\in\cY$ canonically given by $\neutral: t\mapsto \frac{1}{\|t\mapsto 1\|}$ the real constant function of unit norm; re-scaling $\scl = \{y \overset{\scl_\lambda}{\longmapsto} \lambda \cdot (y - \zero_y) + \zero_y : \lambda \in \R^+\}$ around the centroid $\zero_y = \langle \neutral\,, y \rangle \neutral$ (which we consider more natural than using $\zero$, the zero element of $\cY$, mostly chosen in the literature); and rotation $\rot = \{y \overset{\rot_u}{\longmapsto} u\cdot (y-\zero_y) + \zero_y : u \in \S^1\}$ around $\zero_y$ with $\S^1 = \{u\in\C:|u|=1\} = \{\exp(\omega\i):\omega\in\R\}$ 
reflecting counterclockwise rotations by $\omega$ radian measure. Concatenation yields combined group actions $G$ as direct products, such as the rigid motions $G = \trl \dir \rot = \{\trl_\gamma \circ \rot_u  : \gamma \in \C, u \in \S^1 \} \cong \C \times \S^1$ (see  Supplement \ref{sec:groupactions} for more details).
The two real-valued component functions of $y$ are identified with the real part $\re{y}: \cT \rightarrow \R$ and imaginary part $\im{y}:\cT \rightarrow \R$ of $y=\re{y} + \im{y} \i$. 
While the complex setup is used for convenience, the real part of $\langle \cdot, \cdot \rangle$ constitutes an inner product $\re{\langle y_1, y_2 \rangle} = \langle \re{y_1}, \re{y_2} \rangle + \langle \im{y_1}, \im{y_2} \rangle$ for $y_1, y_2 \in \cY$ on the underlying real vector space of planar curves.
Typically $\re{y},\ \im{y}$ are assumed square-intregrable with respect to a measure $\nu$ and we consider the canonical inner product $\langle y_1, y_2 \rangle = \int \conj{y}_1 y_2 d\nu$ where $\conj{y}$ denotes the conjugate transpose of $y$, i.e. $\conj{y}(t) = \re{y}(t) - \im{y}(t)\i$ is simply the complex conjugate, but for vectors $\by\in\C^k$, the vector $\conj{\by}$ is also transposed. 
For curves, we typically assume $\nu$ to be the Lebesgue measure on $\cT = [0,1]$;
for landmarks, a standard choice 
is the counting measure on $\cT = \{1,\dots,k\}$.
 
The ultimate response object is given by the \emph{orbit} $[y]_{G}=\{g(y) : g\in G\}$ (or short $[y]$) of $y\in\cY$, the equivalence class under the respective combined group actions $G$: 
\new{with $G=\trl \dir \rot \dir \scl$,}
$\new{[y] =} [y]_{\trl \dir \rot \dir \scl} = \{\lambda u\, y+ \gamma\, \neutral : \lambda \in \R^+, u \in \S^1, \gamma \in \C\}$ is referred to as the \emph{shape} of $y$ and\new{, for $G=\trl \dir \rot$,}  $\new{[y] =} [y]_{\trl \dir \rot}=\{u y + \gamma\,\neutral :  u \in \S^1, \gamma\in\C\}$ as its \emph{form} or \emph{size-and-shape}.
$\cY_{/G} = \{[y]_G : y\in\cY\}$ denotes the quotient space of $\cY$ with respect to $G$. The description of the Riemannian geometry of $\cY_{/G}$ involves, in particular, a description of the tangent spaces $T_{[y]}\cY_{/G}$ at points $[y]\in \cY_{/G}$, which can be considered local vector space approximations to $\cY_{/G}$ in a neighborhood of $[y]$. For a point $q$ in a manifold $\cM$ the tangent vectors $\beta \in T_{q}\cM$ can, i.a., be thought of as gradients $\dot{c}(0)$ of paths $c:\R \supset (-\delta, \delta) \rightarrow \cM$ at $0$ where they pass through $c(0) = q$. Besides their geometric meaning, they will also play an important role in the regression model, as additive model effects are formulated on tangent space level. 
Choosing suitable representatives $\widetilde{y}^G \in [y]_G \subset \cY$ (or short $\widetilde{y}$) of orbits $[y]_G$, we use an identification of tangent spaces with  suitable linear subspaces $T_{[y]_G}\cY_{/G} \subset \cY$. 

\textbf{Form geometry:}
Starting with translation as the simplest invariance, an orbit $[y]_{\trl}$ can be one-to-one identified with its centered representative $\widetilde{y}^{\trl} =y - \langle y, \neutral\rangle\, \neutral$ 
yielding an identification $\cY_{/\trl} \cong \{y\in\cY : \langle y, \neutral\, \rangle = 0\}$ with a linear subspace of $\cY$. Hence, also $T_{[y]}\cY_{/\trl} = \{y\in\cY : \langle y, \neutral\, \rangle = 0\}$. 
For rotation, by contrast, 
we can only find local identifications with Hilbert subspaces (i.e.\ charts) around reference points $[p]_{\trl \dir \rot}$ we refer to as ``poles''. Moreover, we restrict to $y, p \in \cY^* = \cY \setminus [\zero\,]_{\trl}$ eliminating constant functions as degenerate special cases in the translation orbit of zero. 
For each $[y]_{\trl\dir\rot}$ in an open neighborhood around $[p]_{\trl\dir\rot}$ which can be chosen with $\langle \widetilde{y}^{\trl}, \widetilde{p}^{\trl} \rangle \neq 0$, $y$ can be uniquely rotation aligned to $p$, yielding a one-to-one identification of the form $[y]_{\trl \dir \rot}$ with the aligned representative given by $\widetilde{y}^{\trl\dir\rot} = \frac{\langle \widetilde{y}^{\trl}, \widetilde{p}^{\trl} \rangle}{|\langle \widetilde{y}^{\trl}, \widetilde{p}^{\trl} \rangle|} \widetilde{y}^{\trl} = \argmin[{y' \in [y]_{\trl\dir\rot}}] \|y' - p\|$ (compare Fig. \ref{fig:quotientgeometry}). While $\widetilde{y}^{\trl \dir \rot}$ depends on $p$, we omit this in the notation for simplicity. 
All $\widetilde{y}^{\trl}$ rotation aligned to $\widetilde{p}^{\trl}$ lie on the hyper-plane determined by $\im{\langle \widetilde{y}^{\trl}, \widetilde{p}^{\trl}\rangle} = 0$ (Figure \ref{fig:quotientgeometry}), which yields
$T_{[p]}\cY_{/\trl + \rot}^* = \{y\in\cY : \langle y, \neutral \rangle = 0,\ \im{\langle y, p \rangle} = 0 \}$ \new{with normal vectors $\zeta^{(1)} = \neutral, \zeta^{(2)} = \i\, \neutral, \zeta^{(3)} = \i\, p$}. 
Note that, despite the use of complex arithmetic, $T_{[p]}\cY_{/\trl \dir \rot}^*$ is a real vector space not closed under complex scalar multiplication. 
The geodesic distance of $[y]_{\trl \dir \rot}$ to the pole $[p]_{\trl \dir \rot}$ is given by $d([y]_{\trl\dir\rot}, [p]_{\trl\dir\rot}) = \|\widetilde{y}^{\trl \dir \rot} - \widetilde{p}^{\trl}\| = \underset{y'\in[y]_{\trl\dir\rot}, p'\in[p]_{\trl\dir\rot}}{\operatorname{argmin}}\|y'-p'\|$. It reflects \old{both }the length of the shortest path (i.e. the geodesic) between the forms and the minimum distance between the orbits as sets.

\textbf{Shape geometry:}
To account for scale invariance in shapes $[y]_{\trl \dir \rot \dir \scl}$, 
they are identified with normalized representatives $\widetilde{y}^{\trl \dir \rot \dir \scl} = \frac{\widetilde{y}^{\trl\dir\rot}}{\|\widetilde{y}^{\trl\dir\rot}\|}$. 
Motivated by the normalization, we borrow the well-known geometry of the sphere $\S = \{y \in \cY : \|y\| = 1\}$, where $T_{p}\S = \{y\in \cY : \re{\langle y, p \rangle} = 0\}$ is the tangent space at a point $p\in\S$ and geodesics are great circles.
Together with translation and rotation invariance, the shape tangent space is then given by $T_{[p]}\cY_{/\trl \dir \rot \dir \scl}^* = T_{[p]}\cY_{/\trl \dir \rot}^* \cap T_{p}\S = \{y\in\cY : \langle y, \neutral \rangle = 0,\ \langle y, p \rangle = 0 \}$ \new{with normal vector $\zeta^{(4)} = p$ in addition to $\zeta^{(1)}, \zeta^{(2)}, \zeta^{(3)}$ above}. The geodesic distance $d([p]_{\trl \dir \rot \dir \scl}, [y]_{\trl \dir \rot \dir \scl}) = \arccos |\langle \widetilde{y}^{\trl \dir \rot \dir \scl}, \widetilde{p}^{\trl \dir \rot \dir \scl} \rangle|$ corresponds to the arc-length between the representatives. This distance is often referred to as \textit{Procrustres distance} in statistical shape analysis.\\ 

We may now define the maps needed for the regression model formulation. Let $\widetilde{y}$\old{,} \new{and} $\widetilde{p}$ \old{and $\widetilde{p}'$} be shape/form representatives of $[y]$\old{,} \new{and}  $[p]$ \old{and $[p']$} rotation aligned to the shape/form pole representative $p$. 
Generalizing straight lines to a Riemannian manifold $\cM$, geodesics $c: (-\delta, \delta) \rightarrow \cM$ can be characterized by their ``intercept'' $c(0)\in\cM$ and ``slope'' $\dot{c}(0)\in T_{c(0)}\cM$.
The
\emph{exponential map} $\Exp_{q}: T_q\cM \rightarrow \cM$ at a point $q\in\cM$ is defined to map $\beta \mapsto c(1)$ for $c$ the geodesic with $q=c(0)$ and $\beta=\dot{c}(0)$. It maps $\beta\in T_q\cM $ to a point $\Exp_{q}(\beta) \in \cM$ located $d(q,\Exp_{q}(\beta)) = \|\beta\|$ apart of the pole $q$ in the direction of $\beta$.
On the form space $\cY_{/\trl \dir \rot}$, the exponential map is simply given by $\Exp_{[p]\new{_{\trl\dir\rot}}}(\beta) = \left[\widetilde{p}\new{^{\trl\dir\rot}} + \beta\right]\new{_{\trl\dir\rot}}$. 
On the shape space $\cY_{/\trl \dir \rot \new{\dir}\old{+} \scl}$, identification with exponential maps on the sphere yields $\Exp_{[p]\new{_G}}(\beta) = \left[\cos(\|\beta\|) \widetilde{p}\new{^G} + \sin(\|\beta\|) \frac{\beta}{\|\beta\|}\right]\new{_G}$ \new{ with $G=\trl\dir\rot\dir\scl$}.
In an open neighborhood $\cU$, $q\in\cU\subset\cM$, $\Exp_{q}$ is invertible yielding the $\Log_{q}: \cU \rightarrow T_q\cM$ map from the manifold to the tangent space at $q$. For forms, it is given by $\Log_{[p]_\new{\trl \dir \rot}}([y]_\new{\trl \dir \rot}) = \widetilde{y}^\new{\trl \dir \rot} - \widetilde{p}^\new{\trl \dir \rot}$ and, for shapes, by $\Log_{[p]_\new{G}}([y]_\new{G}) = d([p]_\new{G}, [y]_\new{G}) \frac{\widetilde{y}^\new{G} - \langle \widetilde{p}^\new{G}, \widetilde{y}^\new{G} \rangle \widetilde{p}^\new{G}}{\|\widetilde{y}^\new{G} - \langle \widetilde{p}^\new{G}, \widetilde{y}^\new{G} \rangle \widetilde{p}^\new{G}\|}$ with \new{$G=\trl \dir \rot\dir\scl$}\old{$\widetilde{y}, \widetilde{p}$ the form or shape representatives, respectively}. 
Finally, $\Transp_{q, q'}: T_{q}\cM \rightarrow T_{q'}\cM$ parallel transports tangent vectors $\varepsilon \mapsto \varepsilon'$ isometrically along a geodesic $c(\tau)$ connecting $q$ and $q'\in\cM$ such that the slopes $\new{\Transp_{q, q'}(}\dot{c}(q)\new{)} \new{=}\old{\cong} \dot{c}(q')$ are identified and all angles are preserved. 
For shapes, $\Transp_{[\old{p}\new{y}]_\new{G}, [p\old{'}]_\new{G}}(\varepsilon) = \varepsilon - \langle \varepsilon, {\widetilde{p}\old{'}}^\new{G} \rangle \frac{\old{\widetilde{p}}\new{\widetilde{y}}^\new{G} + {\widetilde{p}\old{'}}^\new{G}}{1+\langle \old{\widetilde{p}}\new{\widetilde{y}}^\new{G}, {\widetilde{p}\old{'}}^\new{G} \rangle}$\new{, with $G = \trl \dir \rot \dir \scl$,} takes the form of the parallel transport on a sphere replacing the real inner product with its complex analogue. 
For forms, it changes only the $\im{\langle\varepsilon, \widetilde{p}\old{'}\rangle}$ coordinate orthogonal to the real $\old{\widetilde{p}}\new{\widetilde{y}}$-$\widetilde{p}\old{'}$-plane as in the shape case, while the remainder of $\varepsilon$ is left unchanged as in a linear space. This yields $\Transp_{[\old{p}\new{y}]\new{_G}, [p\old{'}]\new{_G}}\left(\varepsilon\right) =  \varepsilon -  \im{\langle \widetilde{p}\old{'}\new{^G}/\|\widetilde{p}\new{^G}\old{'}\|, \varepsilon \rangle} \frac{\old{\widetilde{p}}\new{\widetilde{y}^G}/\|\old{\widetilde{p}}\new{\widetilde{y}^G}\| + \widetilde{p}\new{^G}\old{'}/\|\widetilde{p}\new{^G}\old{'}\|}{1+\langle \old{\widetilde{p}}\new{\widetilde{y}^G}/\|\old{\widetilde{p}}\new{\widetilde{y}^G}\|, \widetilde{p}\new{^G}\old{'}/\|\widetilde{p}\new{^G}\old{'}\| \rangle}\i$\new{, with $G=\trl\dir\rot$,} for form tangent vectors. While \old{this or} equivalent expressions for the parallel transport in the shape case can be found, e.g., in \cite{DrydenMardia2016ShapeAnalysisWithApplications, Huckemann2010intrinsicMANOVA}, a corresponding derivation for the form case is given in \old{the} Supplement \ref{sec:transport}\old{. This also involves a more detailed} \new{including a} discussion of the quotient space geometry in differential geometric terms.

Based on this understanding of the response space, we may now proceed to consider a sample of curves $y_1, \dots, y_n \in \cY$ representing orbits $[y_1], \dots, [y_n]$ with respect to group actions $G$. In the functional case, with the domain $\cT = [0,1]$, these curves are usually observed as evaluations $\by_i = (y_i(t_{i1}), \dots, y_i(t_{ik_i}))^\top$ on a finite grid $t_{i1} < \dots < t_{ik_i} \in \cT$ which may differ between observations. In contrast to the \textit{regular} case with common grids, this more general data structure is referred to as \textit{irregular} functional shape/form data. To handle this setting, we replace the original inner product $\langle \cdot, \cdot\rangle$ on $\cY$ by individual 
$\langle y_i , y_i' \rangle_i = \conj{\by}_i \bW_i \by'_i$ providing inner products on the $k_i$-dimensional space $\cY_i=\C^{k_i}$ of evaluations $\by_i, \by_i'$ on the same grid. The symmetric positive-definite weight matrix $\bW_i$ can be chosen to implement an approximation to integration w.r.t.\ the original measure $\nu$ with a numerical integration measure $\nu_i$ such as given by the trapezoidal rule. Alternatively, $\bW_i = \frac{1}{k_i} \bI_{k_i}$ with $k_i \times k_i$ identity matrix $\bI_{k_i}$ presents a canonical choice that is analog to the landmark case for $k_i \equiv k$. 
\new{Moreover, data-driven $\mathbf{W}_i$ could also be motivated from the covariance structure estimated for (potentially sparse) $y_1, \dots, y_n$ along the lines of \citet{Yao:etal:2005, stoecker2022efp}. 
	While this is beyond the scope of this paper, potential procedures are sketched in Supplement \ref{sec::FPC}.}
With the inner products given for $i=1,\dots, n$, the sample space naturally arises as the Riemannian product $\cY^*_{1/G} \times \dots \times \cY^*_{n/G}$ of the orbit spaces, with the individual geometries constructed as described above.

\begin{figure}
	\centering
	\includegraphics[width=0.45\linewidth]{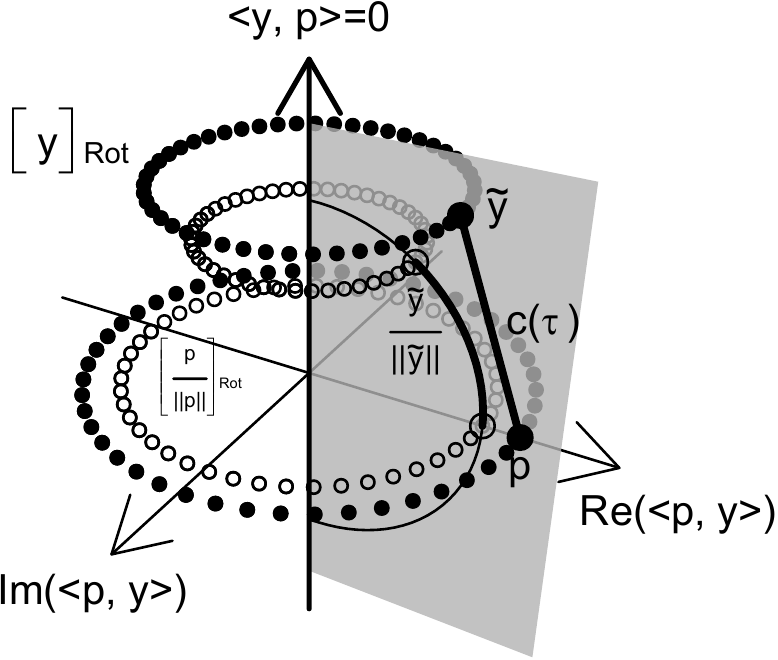}
	\includegraphics[width=0.45\linewidth]{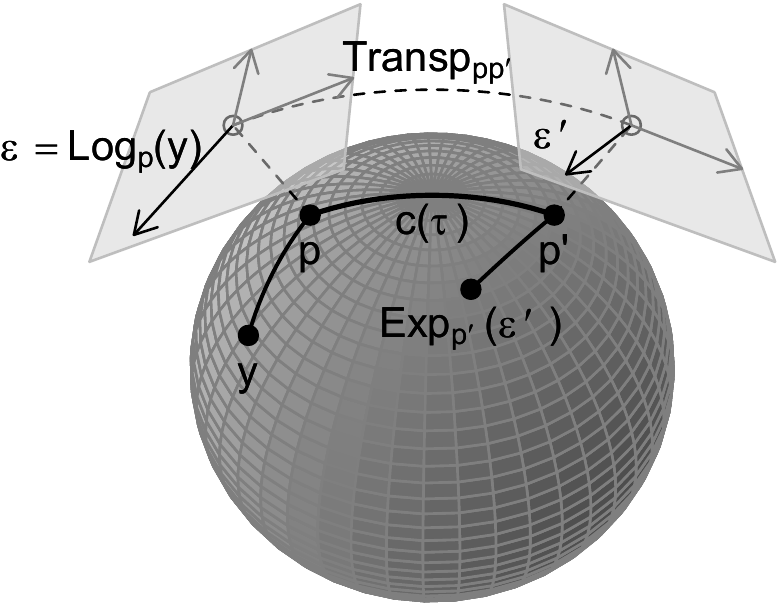}
	\caption[quotientgeometry]{
		\textit{Left:} Quotient space geometry: assuming $p$ and $y$ centered, translation invariance is not further considered in the plot; given pole representative $p$, we express $y = \frac{\re{\langle p, y\rangle}}{\|p\|^2} p + \frac{\im{\langle p, y\rangle}}{\|p\|^2} ip + (y-\frac{\langle p, y\rangle}{\|p\|^2} p) \in\cY$ in its coordinates in $p$ and $ip$ direction, subsuming all orthogonal directions in the third dimension. In this coordinate system, the rotation orbit $[y]_{\rot}$ corresponds to the dotted horizontal circle, and is identified with the aligned $\widetilde{y}:=\widetilde{y}^{\rot}$ in the half-plane of $p$; $[y]_{\rot \dir\scl}$ is identified with the unit vector $\widetilde{y}^{\rot\dir\scl} = \frac{\widetilde{y}}{\|\widetilde{y}\|}$ projecting $\widetilde{y}$ onto the hemisphere depicted by the vertical semicircle. Form and shape distances between $[p]$ and $[y]$ correspond to the length of the geodesics $c(\tau)$ (\textit{thick lines}) on the plane and sphere, respectively. 
		\textit{Right:} Geodesic line $c(\tau)$ between $p=c(0)$ and $p'=c(1)$, Log-map projecting $y$ to $\varepsilon \in T_p\mathcal{M}$, parallel transport $\text{Transp}_{pp'}$ forwarding $\varepsilon$ to $\varepsilon' \in T_{p'}\mathcal{M}$, and Exp-map projecting $\varepsilon'$ onto $\mathcal{M}$ visualized for a sphere. Tangent spaces, identified with subspaces of the ambient space, are depicted as \textit{gray planes} above the respective poles. The parallel transport preserves all angles between tangent vectors and identifies $\dot{c}(0) \cong \dot{c}(1)$.
	}
	\label{fig:quotientgeometry}
\end{figure}

\section{Additive Regression on Riemannian Manifolds}
\label{chap_model}

Consider a data scenario with $n$ observations of a random response covariate tuple $\left(Y, \bX\right)$, where the realizations of $Y$ are planar curves $y_i: \cT \rightarrow \C$, $i=1,\dots,n$, belonging to a Hilbert space $\cY$ defined as above and potentially irregularly measured on individual grids $t_{i1} < \dots < t_{ik_i} \in \cT$. 
The response object $[Y]$ is the equivalence class of $Y$ with respect to translation, rotation and possibly scale and the sample $[y_1],\dots, [y_n]$ is equipped with the respective Riemannian manifold geometry introduced in the previous section. 
For $i=1,\dots,n$, realizations $\bx_{i}\in \cX$ of a covariate vector $\bX$ in a covariate space $\cX$ are observed.  $\bX$ can contain several categorical and/or metric covariates. 

\newcommand{\effid}{j}
\newcommand{\Effid}{J}
For regressing the mean of $[Y]$ on $\bX=\bx$, we model the shape/form $[\mu]$ of \old{a curve }$\mu\in\cY$ as
\begin{equation}
	\label{model_equation}
	[\mu] = \Exp_ {[p]}\left(h(\bx)\right) = \Exp_ {[p]}\left(\sum_{\effid=1}^\Effid h_\effid(\bx)\right), 
\end{equation}
with 
an additive predictor $h: \cX \rightarrow T_{[p]}\cY^*_{/G}$ acting in the tangent space at an ``intercept'' $[p] \in \cY^*_{/G}$.
Generalizing an additive model ``$Y = \mu + \epsilon = p + h(\bx) + \epsilon$'' in a linear space, we implicitly define $[\mu]$ as the conditional mean of $[Y]$ given $\bX=\bx$ by assuming zero-mean ``residuals'' $\epsilon $. In their definition, we follow \cite{CorneaEtAl2017RegRiemannianSymSpaces}
but extend to the functional shape/form and additive case.
We assume local linearized residuals $\varepsilon_{[\mu]} = \Log_{[\mu]}([Y])$ in $T_{[\mu]}\cY^*_{/G}$ to have mean $\E\left(\varepsilon_{[\mu]}\right) = \zero$, 
which corresponds to $\E\left(\varepsilon_{[\mu]}(t)\right) = 0$ for ($\nu$-almost) all $t \in \cT$. 
Here, we assume $[Y]$ is sufficiently close to $[\mu]$ with probability 1 such that $\Log_{[\mu]}$ is well-defined, which is the case whenever $\langle \widetilde{Y}, \widetilde{\mu}\rangle \neq 0$ for centered shape/form representatives $\widetilde{Y}$ and $\widetilde{\mu}$, an un-restrictive and common assumption \cite[compare also][]{CorneaEtAl2017RegRiemannianSymSpaces}. 
However, residuals $\varepsilon_{[\mu]}$ for different $[\mu]$ belong to separate tangent spaces.
To obtain a formulation in a common linear space instead, local residuals are mapped to residuals $\epsilon = \Transp_{[\mu], [p]}(\varepsilon_{[\mu]})$ by parallel transporting them from $[\mu]$ to the common covariate independent pole $[p]$. After this isometric mapping into $T_{[p]}\cY^*_{/G}$, we can equivalently define the conditional mean $[\mu]$ via $\E\left(\epsilon\right) = \zero$ for the transported residuals $\epsilon$.
%
\newline
$\Exp_{[p]}$  maps the additive predictor $h(\bx) = \sum_{\effid=1}^\Effid h_\effid(\bx) \in T_{[p]}\cY^*_{/G}$ to the response space. It is analogous to a response function in GLMs but depends on $[p]$. While other response functions could be used, we restrict to the exponential map here, such that the model contains a geodesic model \citep{Fletcher2012GeodesicRegressionRiemannianManifold} -- the direct generalization of simple linear regression -- as a special case for $h(\bx) = \beta x_1$ with a single covariate $x_1$ and tangent vector $\beta$. 
Typically, it is assumed that $h$ is centered such that $\E\left(h(\bX)\right) = \zero$, and the pole $[p]$ \old{corresponds to}\new{is} the overall mean of $[Y]$ defined, like the conditional mean, via residuals of mean zero.
%

\subsection{Tensor-product effect functions $h_\effid$}
\label{sec_tensorproduct_effects}
\newcommand{\baseid}{r}
\newcommand{\Baseid}{m}
\newcommand{\nb}{0}
\newcommand{\tpid}{{}}
\citet{Scheipl2015} and other authors employ tensor-product (TP) bases for functional additive model terms. This naturally extends to tangent space effects, which we model as
$$h_\effid(\bx) = \sum_{\baseid, l} \theta_j^{(r,l)}\, b_j^{(l)}(\bx)\, \partial_r$$ 
with the TP basis given by the pair-wise products of $\Baseid_\tpid$ linearly independent tangent vectors $\partial_r \in T_{[p]}\cY^*_{/G},\ r=1,\dots, m_\tpid,$ and $\Baseid_j$ basis functions $b_j^{(l)}: \cX \rightarrow \R,\ l=1,\dots,m_j,$ for the $\effid$-th covariate effect depending on one or more covariates. The real \old{basis} coefficients can be arranged as a matrix $\{\theta_j^{(r,l)}\}_{r,l} = \bTheta_j \in \R^{m_\tpid \times m_j}$. 
\new{
Also for infinite-dimensional $T_{[p]}\cY^*_{/G}$ and a general non-linear dependence on $x$, a basis representation approach requires truncation to finite dimensions $m$ and $m_j$ in practice. Choosing the bases to capture the essential variability in the data, their size can be extended with increasing data size and computational resources.}

While, in principle, the basis $\{\partial_r\}_r$ could also vary across effects $j=1,\dots,J$, we assume a common basis for notational simplicity, which presents the typical choice.
Due to the identification of $T_{[p]}\cY^*_{/G}$ with a subspace of the function space $\cY$, 
the $\{\partial_r\}_r$ may be specified using a  function basis commonly used in additive models: Let $b_{\nb}^{(l)}: \cT \rightarrow \R$, $l = 1,\dots, m_\nb$ be a basis of real functions, say a B-spline basis \new{(other typical bases used in the literature include wavelet \citep{meyer2015bayesian} or FPC bases \citep{MullerYao2008FAM})}. Then we construct the tangent space basis as $\partial_r = \sum_{l=1}^{m_\nb} \left( z_p^{(l, r)} + z_p^{(m_\nb+l, r)}\i\right)  b_\nb^{(l)}$, employing the same basis for the $1$- and $\i$-dimension before transforming it with a\old{ suitable} basis transformation matrix $\bZ_p = \{z_p^{(l, r)}\}_{l,r}\in \R^{2m_\nb \times m_\tpid}$ with $\Baseid_\tpid < 2\Baseid_\nb$ implementing the linear tangent space constraints \old{as given in }\new{(}Section \ref{chap_diffgeo}\new{)}.
Practically, $\bZ_p$ \old{can be}\new{is} obtained as null space basis matrix of the matrix $(\re{\bC}, \im{\bC})$ with $\bC=\{\langle b_0^{(l)}, \zeta^{(r)}\rangle\}_{r,l}$ (or with the empirical inner product on the product space of irregular curves instead) constructed from the normal vectors $\zeta^{(r)}\in\cY$, $r=1,\dots, 2m_0-m,$ to $T_{[p]}\cY^*_{/G}$.  
For closed curves, we additionally choose $\bZ_p$ to enforce periodicity, i.e.\ $\partial_r(t) = \partial_r(t+t_0)$ for some $t_0 \in \R$ \citep[compare][]{Hofner2016ConstrainedRegression}. 

\newcommand{\catidx}{\kappa}
\newcommand{\catn}{K}
Given the tangent space basis, we may now modularly specify the usual additive model basis functions $b_j^{(l)}: \cX \rightarrow \R$, $l=1,\dots, m_j$, for the $j$-th covariate effect to obtain the full functional additive model ``tool box'' offered by, e.g., \citet{BrockhausGreven2015}. A linear effect -- linear in the tangent space -- of the form $h_\effid(\bx) = \beta z$ with a scalar (typically centered) covariate $z$ in $\bx$ and $\beta \in T_{[p]}\cY^*_{/G}$ is simply implemented by a single function $b_\effid^{(1)}(\bx) = z$. 
A smooth effect of the generic form $h_\effid(\bx)(t) = f(z, t)$ can be implemented by choosing, e.g., a B-spline basis \new{(Asymptotic properties of penalized B-splines and connections to kernel estimators are discussed, e.g., by \citet{WoodPyaSaefken2016smoothing, LiRuppert2008asymptotics})}.  
For a categorical covariate effect of the form $h_\effid(\bx): \{1,\dots, \catn\} \rightarrow T_{[p]}\cY^*_{/G}$, $\catidx \mapsto \beta_\catidx$, the basis $\bb_\effid(\bx): \catidx \mapsto \be_{\catidx} \in \R^{\catn-1}$ maps category $\catidx$ to a usual contrast vector $\be_{\catidx}$ just as in standard linear models. Here, we typically use effect-encoding to obtain centered effects.
Moreover, TP interactions of the model terms described above, as well as group-specific effects and smooth effects with additional constraints \citep{Hofner2016ConstrainedRegression} can be specified in the model formula, relying on the \texttt{mboost} framework introduced by \citet{HothornBuehlmann2010}, which also allows to define custom effect designs. 
For identification of an overall mean intercept $[p]$, sum-to-zero constraints yielding $\sum_{i=1}^{n} h_j(\bx_i) = \zero$ for observed covariates $\bx_i$\old{, $i=1, \dots, n$,} can be specified, and similar constraints can be used to distinguish linear from non-linear effects and interactions from their marginal effects \citep{KneibHothornTutz2009}.
Different quadratic penalties can be specified for the coefficients $\bTheta_j$, allowing to regularize high-dimensional effect bases and to balance effects of different complexity in the model fit (\old{see also }\new{cf.} Section \ref{chap_boosting}). 

\subsection{Tensor-product factorization}
\label{chap_model::TPfactorization}
The multidimensional structure of the response objects makes it challenging to graphically illustrate and interpret additive model terms, in particular when it comes to non-linear (interaction) effects, or when effect sizes are visually small.
To solve this problem, we suggest to re-write \new{estimated} \old{the} TP effects $\new{\hat{h}}_j$ \new{with estimated} \old{for a given  fixed} coefficient matrix $\new{\widehat{\bTheta}}_j$ 
as
\newcommand{\h}{\new{\hat{h}}}
$$
     \h_j(\bx) = \sum_{\baseid=1}^{\Baseid_j'} \xi^{(\baseid)}_\effid \h_\effid^{(\baseid)}(\bx) 
$$
factorized into $\Baseid_j' = \min(\Baseid_j, \Baseid_0)$ components consisting of covariate effects $\h_\effid^{(\baseid)}: \cX \rightarrow \R,\ \baseid=1,\dots,\Baseid_j',$ in corresponding orthonormal directions $\xi^{(\baseid)}_\effid \in T_{[p]}\cY^*_{/G}$ with $\langle\xi_\effid^{(\baseid)},  \xi_\effid^{(l)} \rangle = \mathds{1}({\baseid = l})$, i.e.\ $1$ if $r=l$ and 0 otherwise. 
Assuming $\E\left(b_j^{(l)}(\bX)^2\right)<\infty$, $l=1,\dots,m_j$, for the underlying effect basis, the  $\h_\effid^{(\baseid)}$ are specified to achieve decreasing component variances $v_j^{(1)} \geq \dots \geq v_j^{(\Baseid_\effid')}\geq 0$ 
given by $v_j^{(\baseid)} = \E\left(\h_\effid^{(\baseid)}(\bX)^2\right)$. In practice, the expectation over the covariates $\bX$ and the inner product $\langle ., .\rangle$ are replaced by empirical analogs (compare Supplement Corollary \ref{corollary:TPfactorization_empirical}). Due to orthonormality of the $\xi_j^{(r)}$, the component variances add up to the total predictor variance $\sum_{\baseid=1}^{\Baseid_j'} v_j^{(\baseid)}= v_j = \E\left(\langle \h_j(\bX), \h_j(\bX)\rangle\right)$.
Moreover, the TP factorization is optimally concentrated in the first components in the sense that for any $l \leq \Baseid_j'$ there is no sequence of $\xi^{(\baseid)}_* \in \cY$ and $\h_*^{(\baseid)}:\cX\rightarrow\R$, such that $\E\left(\|\h_j(\bX) - \sum_{\baseid=1}^l \xi^{(\baseid)}_* \h_*^{(\baseid)}(\bX) \|^2\right) < \E\left(\|h_j(\bX) - \sum_{\baseid=1}^l \xi_\effid^{(\baseid)} \h^{(\baseid)}_\effid(\bX) \|^2\right)$, i.e.\ the series of the first $l$ components yields the best rank $l$ approximation of $\h_j$. 
The factorization relies on SVD of (a transformed version of) the coefficient matrix $\new{\widehat{\bTheta}}_j$ and the fact that it is well-defined is a variant of the  Eckart-Young-Mirsky 
theorem (proof in Supplement \ref{EYM}). \newline 
%
Particularly when large shares of the predictor variance are explained by the first component(s), the decomposition facilitates graphical illustration and interpretation: choosing a suitable constant $\tau\neq0$, an effect direction $\xi^{(\baseid)}_\effid$ can be visualized by plotting the pole representative $p$ together with $\Exp_{p}(\tau\,\xi^{(\baseid)}_\effid)$ on the level of curves, while 
accordingly re-scaled  $\frac{1}{\tau}\h_\effid^{(\baseid)}(\bx)$ is displayed separately in a standard scalar effect plot. 
Adjusting $\tau$ offers an important degree of freedom for visualizing $\xi^{(\baseid)}_\effid$ on an intuitively accessible scale while faithfully depicting $\xi^{(\baseid)}_\effid\h_\effid^{(\baseid)}(\bx)$. 
When based on the same $\tau$, different covariate effects can be compared across the plots sharing the same scale. We suggest $\tau = \max_{j} \sqrt{v_j}$, the maximum total predictor standard deviation of an effect, as a good first choice. 

Besides factorizing effects separately, it can also be helpful to apply TP factorization to the joint additive predictor, yielding
$$
h(\bx) = \sum_{\baseid=1}^{\Baseid'} \xi^{(\baseid)} \h^{(\baseid)}(\bx) = \sum_{\baseid=1}^{\Baseid'} \xi^{(\baseid)} \left( \h_1^{(\baseid)}(\bx) + \dots + \h_J^{(\baseid)}(\bx)\right), \quad \Baseid' = \min(\sum_j m_j, m_\tpid),
$$
with again $\xi^{(\baseid)} \in T_{[p]}\cY^*_{/G}$ orthonormal and the corresponding variance concentration in the first components, but now determined w.r.t.\ entire additive predictors $\h^{(\baseid)} =  \sum_{\effid=1}^\Effid \h_\effid^{(\baseid)}$ 
spanned by all covariate basis functions in the predictor. 
In this representation, the first component yields a geodesic additive model approximation where the predictor moves along a geodesic line $c(\tau) = \Exp_{[p]}\left(\xi^{(1)} \tau\right)$ with the signed distance $\tau \in \R$ from $[p]$, modeled by a scalar additive predictor $\h^{(1)}(\bx)$ composed of covariate effects analogous to the original model predictor. 
In Section \ref{chap_application}, we illustrate its potential in three different scenarios.


\section{Component-wise Riemannian $L_2$-Boosting}
\label{chap_boosting}



Component-wise gradient boosting \citep[e.g.][]{HothornBuehlmann2010} is a step-wise model fitting procedure accumulating predictors from smaller models, so called base-learners, to built an ensemble predictor aiming at minimizing a mean loss function. To this end, the base-learners are fit (via least squares) to the negative gradient of the loss function in each step and the best fitting base-learner is added to the current ensemble predictor. 
Due to its versatile applicability, inherent model selection, and slow over-fitting behavior, boosting has proven useful in various contexts \citep{mayr2014evolution}.
Boosting with respect to the least squares loss function $\ell(y,\mu) = \frac{1}{2}(y-\mu)^2$, $y,\mu\in\R$, is typically referred to as $L_2$-Boosting and simplifies to repeated re-fitting of residuals $\varepsilon = y-\mu = -\nabla_\mu \ell(y, \mu)$ corresponding to the negative gradient of the loss function. For $L_2$-Boosting with a single learner, \cite{Buehlmann2003L2boosting} show how fast bias decay and slow variance increase over the boosting iterations suggest stopping the algorithm early before approaching the ordinary (penalized) least squares estimator. 
\cite{LutzBuehlmann2006BoostingHighMultivariate} prove consistency of component-wise $L^2$-Boosting in a high-dimensional multivariate response linear regression setting and \cite{Stoecker2019FResponseLSS} illustrate in extensive simulation studies how stopping the boosting algorithm early based on curve-wise cross-validation applies desired regularization when fitting (even highly autocorrelated) functional responses with parameter-intense additive model base-learners and, thus, leads to good estimates even in challenging scenarios.\\
When generalizing to least squares on Riemannian manifolds with the loss $\frac{1}{2}d^2([y],[\mu])$ given by the squared geodesic distance, the negative gradient $-\nabla_{[\mu]}\frac{1}{2}d^2([y], [\mu])=\Log_{[\mu]}([y]) = \varepsilon_{[\mu]}$ \citep[compare e.g.\ ][]{Pennec2006BasicTools} corresponds to the local residuals $\varepsilon_{[\mu]}$ defined in Section \ref{chap_model}. This analogy to $L_2$-Boosting motivates the presented generalization where local residuals are further transported to residuals $\epsilon$ in a common linear space. 

Consider the pole $[p]$ known and fixed for now. Assuming its existence, we aim to minimize the population mean loss
$$
 \sigma^2(h) = \E\left( d^2\left([Y], \Exp_ {[p]}\left( h(\bX) \right)\right) \right)
$$
with the point-wise minimizer 
$h^\star(\bx) = \argmin[h:\cX \rightarrow T_{[p]}\cY_{/G}^*] \E\left( d^2\left([Y], \Exp_ {[p]}\left( h(\bX) \right)\right) \mid \bX = \bx \right)$
minimizing the conditional expected squared distance. Fixing a covariate constellation $\bx \in \cX$, the prediction $[\mu] = \Exp_ {[p]}\left( h^\star(\bx) \right)$ corresponds to the Fréchet mean \citep{Karcher1977KarcherMean} of $[Y]$ conditional on $\bX = \bx$. 
In a finite-dimensional context, \citet{Pennec2006BasicTools} show that $\E\left(\varepsilon_{[\mu]}\right)=\zero$ for a Fréchet mean $[\mu]$ if residuals $\varepsilon_{[\mu]}$ are uniquely defined with probability one. This indicates the connection to our residual based model formulation in Section \ref{chap_model}.
We fit the model by reducing the empirical mean loss
$
\hat\sigma^2(h) = \frac{1}{n} \sum_{i=1}^n d_i^2\left([y_i], \Exp_ {[p]}\left(h(\bx_i)\right)\right),
$
where we replace the population mean by the sample mean and compute the geodesic distances $d_i$ with respect to the inner products $\langle\cdot, \cdot\rangle_i$ defined for the respective evaluations of $y_i$. 

A base-learner corresponds to a covariate effect $h_\effid(\bx) = \sum_{\baseid, l} \theta_j^{(r,l)}\, b_j^{(l)}(\bx)\, \partial_r$, $\bTheta_j = \{\theta_j^{(r,l)}\}_{r,l}$, which is repeatedly fit to the transported residuals $\epsilon_1, \dots, \epsilon_n$ by penalized least-squares (PLS) minimizing $\sum_{i=1}^n \|\epsilon_i - h_\effid(\bx_i)\|_i^2 + \lambda_j \tr{\left(\bTheta_j\bP_j\bTheta_j^\top\right)} + \lambda_\tpid \tr{\left(\bTheta^\top\bP_\tpid\bTheta\right)}$. 
Via the penalty parameters $\lambda_\effid, \lambda_\tpid \geq 0$ the effective degrees of freedom of the base-learners are controlled \citep{HofnerSchmid2011} to achieve a balanced ``fair'' base-learner selection despite the typically large and varying number of coefficients involved in the TP effects. 
The symmetric penalty matrices $\bP_j\in\R^{\Baseid_j\times\Baseid_j}$ and $\bP_\tpid\in\R^{\Baseid_\tpid\times\Baseid_\tpid}$ (imposing, e.g., a second-order difference penalty for B-splines in either direction) can equivalently be arranged as a $\Baseid_j\Baseid_\tpid\times\Baseid_j\Baseid_\tpid$ penalty matrix $\bR_\effid=\lambda_\effid (\bP_\effid \otimes \bI_{\Baseid_\tpid}) + \lambda_\tpid (\bI_{\Baseid_j} \otimes \bP_\tpid)$ for the vectorized coefficients $\operatorname{vec}{(\bTheta_j)}=(\theta^{(1,1)}_j, \dots, \theta^{(\Baseid_\tpid,1)}_j, \dots, \theta^{(\Baseid_\tpid,\Baseid_j)})^\top$, where $\otimes$ denotes the Kronecker product. 
The standard PLS estimator is then given by $\operatorname{vec}{(\widehat\bTheta_j)} = \left(\bPsi_\effid + \bR_\effid\right)^{-1} \bpsi_\effid$ with
$\bPsi_\effid = \sum_{i=1}^n  \left\{ \re{\langle b^{(l)}_\effid(\bx_i) \partial_r, b^{(l')}_\effid(\bx_i) \partial_{r'}\rangle_i} \right\}_{\substack{(r,l)=(1,1), \dots, (m_\tpid,1), \dots, (m_\tpid,m_j)\\ (r',l')=(1,1), \dots, (m_\tpid,1), \dots, (m_\tpid,m_j) }} \in \R^{m_\tpid\,m_j \times m_\tpid\,m_j}$ and $\bpsi_\effid = \sum_{i=1}^n  \left\{ \re{\langle b^{(l)}_\effid(\bx_i) \partial_r, \epsilon_i\rangle_i} \right\}_{(r,l)=(1,1), \dots, (m_\tpid, 1), \dots, (m_\tpid,m_j)} \in \R^{m_\tpid\,m_j}$.
In a regular design, using the functional linear array model \citep{BrockhausGreven2015} can  save memory and computation time by avoiding construction of the complete matrices.
The basis construction of $\{\partial_r\}_r$ via a transformation matrix $\bZ_p$ (Section \ref{sec_tensorproduct_effects}) is reflected in the penalty by setting $\bP_\tpid = \bZ_p^\top (\bI_2 \otimes \bP_\nb) \bZ_p$ with $\bP_\nb$ the penalty matrix for the un-transformed basis $\{b_\nb^{(r)}\}_r$. 

In each iteration of the proposed Algorithm \ref{riemmanianL2boosting}, the best-performing base-learner is added to the current ensemble additive predictor $h(\bx)$ after multiplying it with a step-length parameter $\eta\in(0,1]$. Due to the additive model structure this corresponds to a coefficient update of the selected covariate effect. Accordingly, after repeated selection, the effective degrees of freedom of a covariate effect, in general, exceed the degrees specified for the base-learner. They are successively adjusted to the data.
To avoid over-fitting, the algorithm is typically stopped early before reaching a minimum of the empirical mean loss. The stopping iteration is determined, e.g., by re-sampling strategies such as bootstrapping or cross-validation on the level of shapes/forms. 

\newcommand{\steplen}{\eta}

\IncMargin{1em}
\begin{algorithm}
	\SetKwFunction{matrix}{matrix}	
	\SetKwInOut{Input}{Hyper-parameters}
	\SetKwInOut{Output}{Geometry}
	\SetKwInOut{Initialize}{Base-learners}
	\SetKwFor{For}{for}{iterations do}{end}
	\SetKwFor{ForEach}{for}{do}{end}
	\SetKwComment{rc}{\#\ }{}
	\SetKwFunction{Solve}{Solve}
	\DontPrintSemicolon
	\rc*[h]{Initialization:}\\
	\Output{specify geometry (shape/form) and pole representative $p$}
	\Input{Step-length $\steplen \in (0,1]$, number of boosting iterations}
	\Initialize{ $h_j(\bx)$ with penalty matrix $\bR_j$ and\\ initial coefficient matrix $\bTheta_j=\boldsymbol{0}$  }
	\ForEach(\hfill\rc*[h]{Prepare penalized least-squares (PLS)}){$j = 1$ \KwTo $J$}{
		\rc*[h]{set up $m_\tpid\,m_j \times m_\tpid\,m_j$ matrix:}
		$\bPsi_\effid \leftarrow \sum_{i=1}^n  \left\{ \re{\langle b^{(l)}_\effid(\bx_i) \partial_r, b^{(l')}_\effid(\bx_i) \partial_{r'}\rangle_i} \right\}_{\substack{(r,l)=(1,1), \dots, (m_\tpid,1), \dots, (m_\tpid,m_j)\\ (r',l')=(1,1), \dots, (m_\tpid,1), \dots, (m_\tpid,m_j) }}$\; 
	}
	\Repeat(\qquad\rc*[h]{boosting steps}){Stopping criterion (e.g. minimal cross-validation error)}{
		\ForEach(\hfill\rc*[h]{Compute current transported residuals}){$i = 1, \dots, n$}{	
				$[\mu_i] \leftarrow \Exp_{[p]}(h(\bx_i))$\;
				$\varepsilon_{[\mu_i]} \leftarrow \Log_{[\mu_i]}([y_i])$\;
				$\epsilon_i \leftarrow \Transp_{[\mu_i], [p]}(\varepsilon_{[\mu_i]})$\;
		}
		\ForEach(\hfill\rc*[h]{PLS fit to residuals}){$j = 1, \dots, J$}{
			\rc*[h]{$m_\tpid\,m_j$ vector:}
			$\bpsi_\effid \leftarrow \sum_{i=1}^n  \left\{ \re{\langle b^{(l)}_\effid(\bx_i) \partial_r, \epsilon_i\rangle_i} \right\}_{(r,l)=(1,1), \dots, (m_\tpid,1), \dots, (m_\tpid,m_j)}$\;
			$\widehat{\bTheta}_\effid = \{\hat\theta_j^{(r,l)}\}_{r,l} \leftarrow$ \Solve{\ $\left(\bPsi_\effid + \bR_\effid\right) \operatorname{vec}(\bTheta) = \bpsi_\effid$\ }\;
		}
		$\hat\jmath \leftarrow \argmin[\effid\in\{1,\dots,J\}] \sum_{i=1}^n \| \epsilon_i -  \sum_{r,l} \hat{\theta}^{(r,l)}_\effid b_j^{(l)}(\bx) \partial_r \|_i^2$; \hfill\rc*[h]{Select base-learner}\\
		$\bTheta_{\hat\jmath} \leftarrow \bTheta_{\hat\jmath} + \steplen\, \widehat{\bTheta}_{\hat\jmath}$; \hfill\rc*[h]{Update selected model coefficients}\\
	}
	\caption{Component-wise Riemannian $L^2$-Boosting}\label{riemmanianL2boosting}
\end{algorithm}

The pole $[p]$ is, in fact, usually not a priori available. Instead we typically assume $[p]=\argmin[q\in\cY^*]\E\left(d^2([Y], [q])\right)$ is the overall Fréchet mean, also often referred to as \textit{Riemannian center of mass} for Riemannian manifolds or as \textit{Procrustes mean} in shape analysis \citep{DrydenMardia2016ShapeAnalysisWithApplications}. Here, we estimate it as $[p] = \Exp_{[p_0]}(h_0)$ in a preceding Riemannian $L^2$-Boosting routine. The constant effect $h_0\in T_{[p_0]}\cY^*_{/G}$ in the intercept-only special case of our model is estimated with Algorithm \ref{riemmanianL2boosting} based on a preliminary pole $[p_0]\in \cY^*_{/G}$.
For shapes and forms, a good candidate for $p_0$ can be obtained as the standard functional mean of a reasonably well aligned sample $y_1, \dots, y_n \in \cY$ of representatives. 


The proposed Riemannian $L_2$-Boosting algorithm is available in the \texttt{R} \citep{R} package \texttt{manifoldboost}\if0\blind{ (\url{github.com/Almond-S/manifoldboost})}\fi. The implementation is based on the package \texttt{FDboost} \citep{FDboost}, which is in turn based on the model-based boosting package \texttt{mboost} \citep{HothornBuehlmann2010}. 


\section{Applications and Simulation}
\label{chap_application}
\subsection{Shape differences in astragali of wild and domesticated sheep}
\label{sec_sheep}

\newcommand{\coefi}[1]{\bbeta_{\mathrm{#1}_i}}
\newcommand{\coefix}[1]{\bbeta_{{#1}_i}}
\newcommand{\hphi}{\psi}
\newcommand{\bxi}{\boldsymbol{\xi}}

\new{In a geometric morphometric study, }\citet{Poellath2019sheepbones} \old{in a paleoanatomical study }investigate shapes of sheep astragali (ankle bones) \new{to understand the influence of different living conditions on the micromorphology of the skeleton}\old{to understand morphological effects of domestication}. Based on a total of $n=163$ shapes recorded by \citet{Poellath2019sheepbones}
, we model the astragalus shape in dependence on different \new{variables}\old{characteristics}, including domestication status (wild/feral/domesticated), sex (female/male/NA), age (juvenile/subadult/adult/NA), and mobility (confined/pastured/free) of the animals as categorical covariates. The sample comprises sheep of four different \new{populations}\old{breeds}: 
\new{A}\old{a}siatic wild sheep \citep[Field Museum, Chicago;][]{lay1967wildsheep, zeder2006wildsheep}, feral Soay sheep 
\citep[British Natural History Museum, London;][]{Clutton-Brock1990soay}, and domestic \new{sheep of the Karakul and Marsch breed}\old{Karakul sheep 
as well as Marsch sheep} 
\citep[Museum of Livestock Sciences, Halle (Saale);][]{Schafberg2010KarakulMarsch}. Table \ref{tab:data-summary} in Supplement \ref{sec:bones_appendix} shows the distribution of available covariates within the \new{population}\old{breed}s. 
Each sheep astragalus shape, $i = 1,\dots, n$, is represented by a configuration composed of 11 selected landmarks in a vector $\by_{i}^{\text{lm}}\in\C^{11}$ and two vectors of \new{sliding}  semi-landmarks $\by_{i}^{\text{c1}}\in\C^{14}$ and $\by_{i}^{\text{c2}}\in\C^{18}$ evaluated along two outline curve segments, marked on a 2D image of the bone (dorsal \new{view}\old{perspective}). 
Several example configurations are displayed in Supplement Figure \ref{fig:sheepdataexamples}.
In general, we could separately specify 
smooth function bases for the outline segments $y_i^{\text{c1}}$ and $y_i^{\text{c2}}$, respectively. 
Due to their systematic recording, we assume, however, that not only landmarks but also semi-landmarks are regularly observed on a fixed grid, and refrain from using smooth function bases for simplicity. 
Accordingly, shape configurations can directly be identified with their evaluation vectors $\by_i = \vec{\by_i^{\text{lm}\top}, \by_i^{\text{c1}\top}, \by_i^{\text{c2}\top}} \in \C^{43} = \cY$, and the geometry of the response space $\cY^*_{/\trl \dir \rot \dir \scl}$ widely corresponds to the classic Kendall's shape space geometry, with the difference that, considering landmarks more descriptive than single semi-landmarks, we choose a weighted inner product $\langle \by_i, \by_i'\rangle = \by_i^\dagger \bW \by_i'$ with diagonal weight matrix $\bW$ with diagonal $\vec{\boldsymbol{1}_{11}^\top, \frac{3}{14} \boldsymbol{1}_{14}^\top, \frac{3}{18} \boldsymbol{1}_{18}^\top}$ assigning the weight of three landmarks to each outline segment. 
We model the astragalus shapes $[\by_i]\in \cY^*_{/\trl \dir \rot \dir \scl}$ as 
\begin{align*}
	[\bmu_i] &= \Exp_{[\bp]}\left(\coefi{status} + \coefi{\new{pop}\old{breed}} + \coefi{age} + \coefi{sex} + \coefi{mobility}\right)
\end{align*}
with the pole $[\bp]\in \cY^*_{/G}$ specified as overall mean and the conditional mean $[\bmu_i]\in \cY^*_{/\trl \dir \rot \dir \scl}$ depending on the effect coded covariate effects $x_{ij} \mapsto \bbeta_{x_{ij}}\in T_{[\bp]}\cY^*_{/\trl \dir \rot \dir \scl}$. 
For identifiability, the \new{population}\old{breed} and mobility effects are centered around the status effect, as we only have data on different \new{population}\old{breed}s/mobility levels for domesticated sheep. 
All base-learners are regularized to one degree of freedom by employing ridge penalties for the coefficients of the covariate bases $\{b_j^{(l)}\}_l$ 
while the coefficients of the response basis (the standard basis for $\C^{43})$ 
are left un-penalized. 
With a step-length of $\steplen = 0.1$, 10-fold shape-wise cross-validation suggests early stopping after $89$ boosting iterations.
Due to the regular design, we can make use of the functional linear array model \citep{BrockhausGreven2015} for saving computation time and memory, which lead to 8 seconds of initial model fit followed by 47 seconds of cross-validation. 
To interpret the categorical covariate effects, we rely on TP factorization (Figure \ref{fig:bone_effects}). 
The first component of the status effect explains about 2/3 of the variance of the status effect and over 50\% of the cumulative effect variance in the model. In that main direction,  the effect of \textit{feral} is not located between \textit{wild} and \textit{domestic}, as might be naively expected. By contrast, the second component of the effect seems to reflect the expected order and still explains a considerable amount of variance.
Similar to \cite{Poellath2019sheepbones}, we find  little influence of age, sex and mobility on the astragalus shape. Yet, all covariates were selected by the boosting algorithm. 

\begin{figure}[!h]
	\centering 
	\includegraphics[width=1.5in]{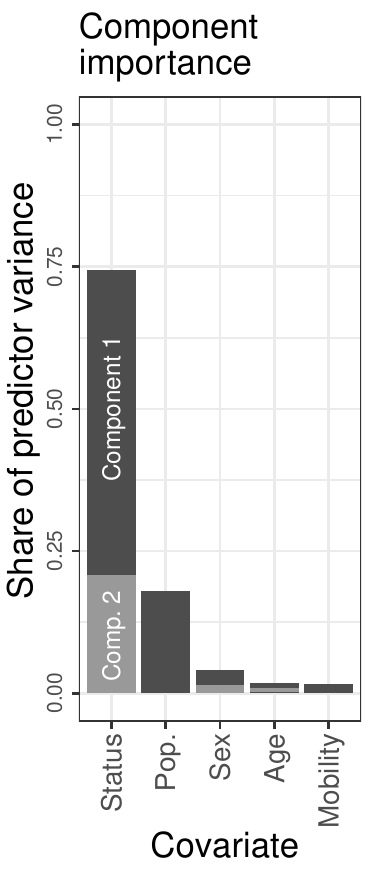}
	\includegraphics[width=3.5in]{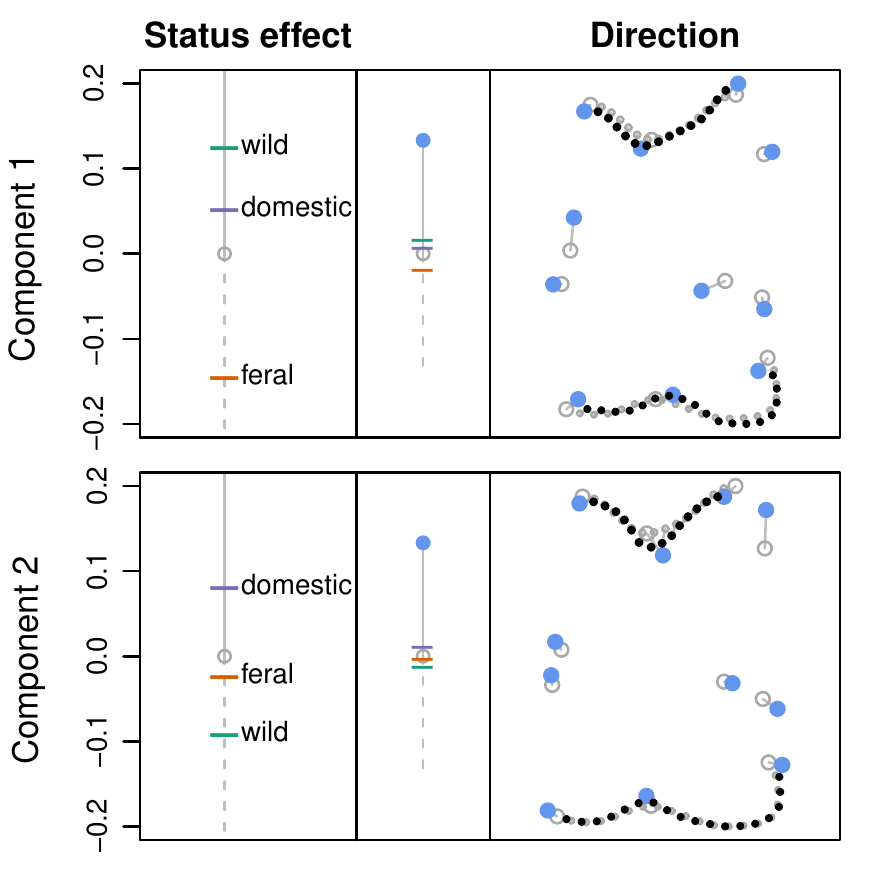}
	\caption{
		\textit{Left}: Shares of different factorized covariate effects in the total predictor variance. 
		\textit{Right}: Factorized effect plots showing the two components of the status effect (\textit{rows}): in the \textit{right column}, the two first directions $\bxi^{(1)}_1, \bxi^{(2)}_1 \in T_{[\bp]}\cY^*_{/\trl + \rot + \scl}$ are visualized via line-segments originating at the overall mean shape (\textit{empty circles}) and ending in the shape resulting from moving 1 unit into the target direction (\textit{solid circles}; \textit{large}: landmarks; \textit{small}: semi-landmarks along the outline); in the  \textit{left column}, the status effect in the respective direction is depicted. As illustrated in the \textit{middle} plot, an effect of 1 would correspond to the full extend of the direction shown to the right.}\label{fig:bone_effects}
\end{figure}

Visually, differences in estimated mean shapes are rather small, which is, in our experience, quite usual for shape data. With differences in size, rotation and translation excluded by definition, only comparably small variance remains in the observed shapes. 
Nonetheless, TP factorization provides accessible visualization of the effect directions and allows to partially order the effect levels in each direction.

\subsection{Cellular Potts model parameter effects on cell form}
\label{sec_cells}

The stochastic biophysical model proposed by \cite{Thueroff2019SingleCellCollectiveDynamics}, a cellular Potts model (CPM), simulates migration dynamics of cells (e.g. wound healing or metastasis) in two dimensions.
The progression of simulated cells is the result of many consecutive local elementary events sampled with a Metropolis-algorithm according to a Hamiltonian. 
Different parameters controlling the Hamiltonian have to be calibrated to match real live cell properties \citep{schaffer2021phd}. Considering whole cells, parameter implications on the cell form are not obvious. 
To provide additional insights, we model the cell form in dependence on four CPM parameters considered particularly relevant: the bulk stiffness $x_{i1}$, membrane stiffness $x_{i2}$, substrate adhesion $x_{i3}$, and signaling radius $x_{i4}$ are subsumed in a vector $\bx_i$ of metric covariates for $i=1,\dots,n$.
Corresponding sampled cell outlines $y_i$ were provided by Sophia Schaffer in the context of \cite{schaffer2021phd}, who ran underlying CPM simulations and extracted outlines. 
Deriving the intrinsic orientation of the cells from their movement trajectories, we parameterize $y_i: [0,1] \rightarrow \C$, clockwisely relative to arc-length such that $y_i(0) = y_i(1)$ points into the movement direction of the barycenter of the cell. With an average of $k = \frac{1}{n}\sum_{i=1}^{n}k_i \approx 43$ samples per curve \citep[after sub-sampling preserving $95\%$ of their inherent variation, as described in][Supplement]{Volkmann2020MultiFAMM}, the evaluation vectors $\by_i \in \C^{k_i}$ are equipped with an inner-product implementing trapezoidal rule integration weights. Example cell outlines are depicted in Supplement Figure \ref{fig:celldataexamples}. The results shown below are based on cell samples obtained from 30 different CPM parameter configurations. For each configuration, 33 out of 10.000 Monte-Carlo samples were extracted as approximately independent. This yields a dataset of $n = 990 = 30 \times 33$ cell outlines. 

As positioning of the irregularly sampled cell outlines $y_i$, $i=1,\dots,n$, in the coordinate system is arbitrary, we model the cell forms $[y_i]\in\cY^*_{/\trl + \rot}$. Their estimated overall form mean $[p]$ serves as pole in the additive model
\newcommand{\jalt}{{\ddot{\jmath}}}
\begin{align*}
	[\mu_i] = \Exp_{[p]}\big( h(\bx_i) \big) = \Exp_{[p]}\big( \sum_{j} \beta_j x_{ij} + \sum_{j} f_j(x_{ij}) + \sum_{j\neq \jalt} f_{j\jalt}(x_{ij}, x_{i\jalt}) \big)
\end{align*}
where the conditional form mean $[\mu_i]$ is modeled in dependence on tangent-space linear effects with coefficients $\beta_j \in T_{[p]}\cY_{/\trl \dir \rot}$ and non-linear smooth effects $f_j$ for covariate $j=1,\dots, 4$, as well as smooth interaction effects $f_{j\jalt}$ for each pair of covariates $j \neq \jalt$. All involved (effect) functions are modeled via a cyclic cubic P-spline basis $\{b_{0}^{(r)}\}_r$ with $7$ (inner) knots and a ridge penalty, and quadratic P-splines with $4$ knots for the covariates $x_{ij}$ equipped with a second order difference penalty for the $f_j$ and ridge penalties for interactions. 
Covariate effects are mean centered and interaction effects $f_{j\jalt}(x_j, x_\jalt)$ are centered around their marginal effects $f_{j}(x_j), f_\jalt(x_\jalt)$, which are in turn centered around the linear effects $\beta_j x_j$ and $\beta_\jalt x_\jalt$, respectively. 
Resulting predictor terms involve 69 (linear effect) to 1173 (interaction) basis coefficients but are penalized to a common degree of freedom of 2 to ensure a fair base-learner selection.
We fit the model with a step-size of $\steplen = 0.25$ 
and stop after 2000 boosting iterations observing no further meaningful risk reduction, since no need for early-stopping is indicated by 10-fold form-wise cross-validation. 
Due to the increased number of data points and coefficients, the irregular design, and the increased number of iterations, the model fit takes considerably longer than in Section \ref{sec_sheep}, with about 50 initial minutes followed by 8 hours of cross-validation. However, as usual in boosting, model updates are large in the beginning and only marginal in later iterations, such that fits after $1000$ or $500$ iterations would already yield very similar results.   

Observing that the most relevant components point into similar directions, we jointly factorize the \old{additive model }predictor as $\h(\bx_i) = \sum_r \xi^{(r)} \h^{(r)}(\bx_i)$ with TP factorization. The first component explains about 93\% of the total predictor variance (Supplement Fig. \ref{fig:cellmodelfactorizedvarimp}), indicating that, post-hoc, a good share of the model can be reduced to the geodesic model $[\new{\hat{\mu}}_i] = \Exp_{[p]}(\xi^{(1)} \h^{(1)}(\bx_i))$ illustrated in Figure \ref{fig:cells_factorized_effect_plots}. A positive effect in the direction $\xi^{(1)}$ makes
cells larger and more keratocyte / croissant shaped, a negative effect -- pointing into the opposite direction -- makes them smaller and more mesenchymal shaped / elongated.
\begin{figure}[t]
	\includegraphics[width=5.5in]{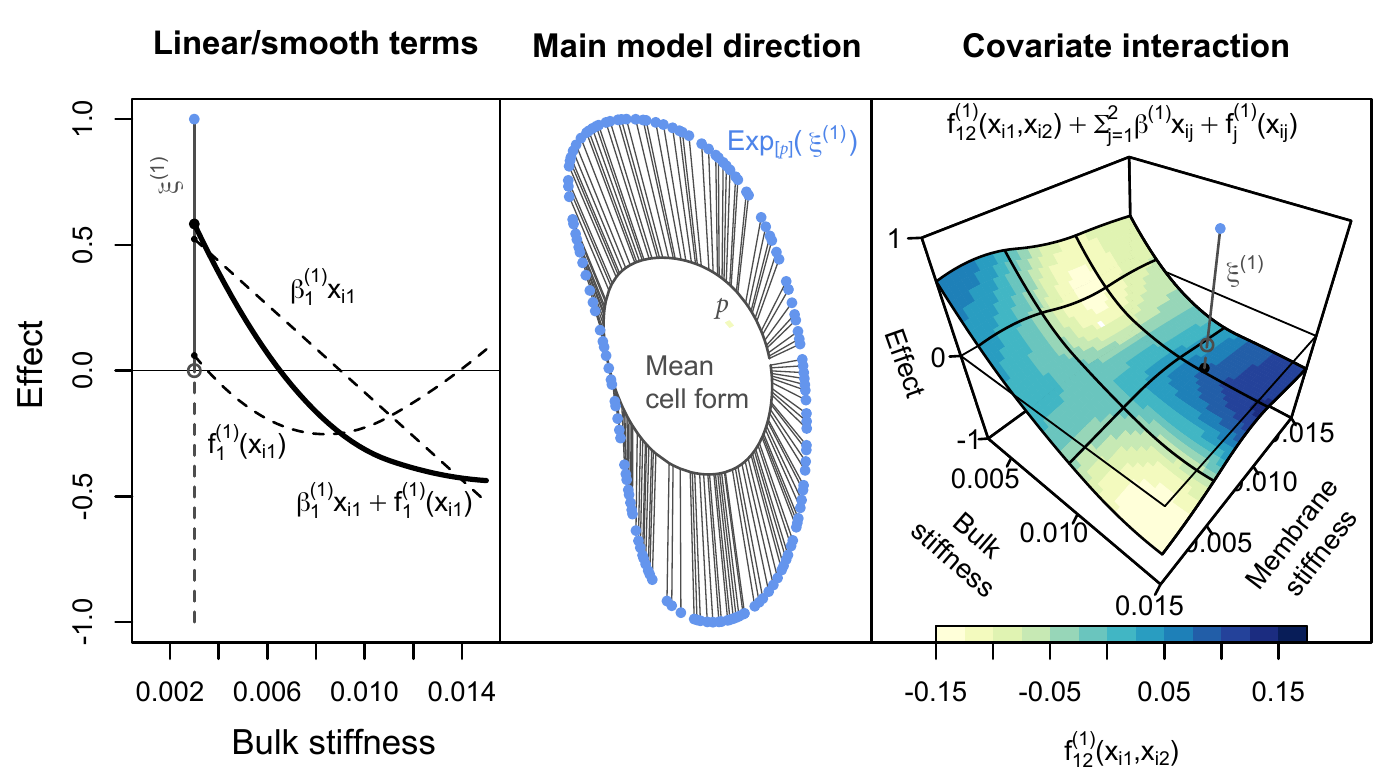}
	\caption{\textit{Center:} the main direction $\xi^{(1)}$ of the model illustrated as vectors pointing from the overall mean cell form $[p]$ (\textit{grey curve}) to the form $\Exp_{[p]}(\xi^{(1)})$ (\textit{blue dots}), which are both oriented as cells migrating rightwards.
	\textit{Left:} Effects of the bulk stiffness $x_{i1}$ into the direction $\xi^{(1)}$. 
	A vertical line from 0, corresponding to $[p]$, to 1, corresponding to the full extent of $\xi^{(1)}$, underlines the connection between the plots and helps to visually asses the amount of change for a given value of $x_{i1}$.
	\textit{Right:} 
	The overall effect of $x_{i1}$ and membrane stiffness $x_{i2}$, comprising linear, smooth and interaction effects, as a 3D surface plot. The heat map plotted on the surface shows only the interaction effect $f^{(1)}_{12}(x_{i1}, x_{i2})$ illustrating deviations from the marginal effects, which are of particular interest for CPM calibration.
	}\label{fig:cells_factorized_effect_plots}
\end{figure}
%
The bulk stiffness $x_{i1}$ turns out to present the most important driving factor behind the cell form, explaining over 75\% of the cumulative variance of the effects (Supplement Fig. \ref{fig:cellcovfactorizedvarimp}). Around 80\% of its effect are explained by the linear term reflecting gradual shrinkage at the side of the cells with increasing bulk stiffness.  

\subsection{Realistic shape and form simulation studies}
\label{sec_simulation}

To evaluate the proposed approach, we conduct \old{a series of }simulation studies for both form and shape regression for irregular curves. 
We compare sample sizes $n\in\{54, 162\}$ and average grid sizes $k=\frac{1}{n}\sum_{i=1}^n k_i \in\{40,100\}$ as well as an extreme case with $k_i=3$ for each curve but $n=720$, i.e.\ where only random triangles are observed (yet, with known parameterization over $[0,1]$). We additionally investigate the influence of nuisance effects and compare different inner product weights. 
While important results are summarized in the following, comprehensive visualizations can be found in Supplement\old{ary}~\ref{sec:bottles_appendix}.

\textbf{Simulation design:} We simulate models of the form $[\mu] = \Exp_{[p]}\left(\beta_\catidx + f_1(z_1)\right)$ with overall mean $[p]$, a binary \old{covariate}\new{effect} with levels $\catidx\in\{0,1\}$ and a smooth effect of \old{a metric covariate }$z_1 \in [-60,60]$. 
We choose a cyclic cubic B-spline basis with $27$ knots for $T_{[p]}\cY^*_{/G}$, \new{placing them irregularly at 1/27-quantiles of unit-speed parameterization time-points of the curves. 
C}ubic B-splines with $4$ \new{regularly placed} knots \new{are used} for covariates in smooth effects. 
True models are \old{obtained }based on the \texttt{bot} dataset from R package \texttt{Momocs} \citep{Bonhommeetal2014RPackMomocs} comprising outlines of 20 beer ($\catidx = 0$) and 20 whiskey ($\catidx = 1$) bottles of different brands. A \old{metric covariate}\new{smooth} effect is induced by the 2D viewing transformations resulting from tilting the planar outlines in a 3D coordinate system along their longitudinal axis by an angle of up to 60 degree towards the viewer ($z_1 = 60$) and away ($z_1 = -60$) (i.e. in a way not captured by 2D rotation invariance). \old{Besides e}\new{E}stablishing \old{the }ground truth models based on a fit to the bottle data, we simulate new responses $[y_1], \dots, [y_n]$ via residual re-sampling (Supplement \ref{sec:bottles_appendix}) 
to preserve realistic autocorrelation. 
Subsequently, we randomly translate, rotate and scale $y_1, \dots, y_n \in \cY$ somewhat around the aligned form/shape representatives to obtain realistic samples.\\
The implied residual variance $\frac{1}{n}\sum_{i=1}^n \|\epsilon_i\|_i^2 = \frac{1}{n}\sum_{i=1}^n d_i^2([y_i],[\mu_i])$ on simulated datasets ranges around 105\% of the predictor variance $\frac{1}{n}\sum_{i=1}^n \|h(\bx_i)\|_i^2 = \frac{1}{n}\sum_{i=1}^n d_i^2([\mu_i],[p])$ in the form scenario and around 65\% in the shape scenario. All simulation\new{s} \old{settings }were repeated 100 times, \old{fitting the generated data with}\new{fitting} models \old{including}\new{with} the model \old{components}\new{terms} specified above and three additional nuisance effects: a linear effect $\beta z_1$ (orthogonal to $f_1(z_1)$), an effect $f_2$ of the same structure as $f_1$ but depending on an independently uniformly drawn variable $z_2$, and a constant effect $h_0\in T_{[p]}\cY^*_{/G}$ to test centering around $[p]$. \old{For the model fit, all b}\new{B}ase-learners are regularized to 4 degrees of freedom \new{(step-length $\steplen = 0.1$)}. \old{We specify a step-length of $\steplen = 0.1$, and e}\new{E}arly-stopping is based on 10-fold cross-validation.\\
\begin{figure}[!h]
	\centering
	\includegraphics[width = 0.73\textwidth]{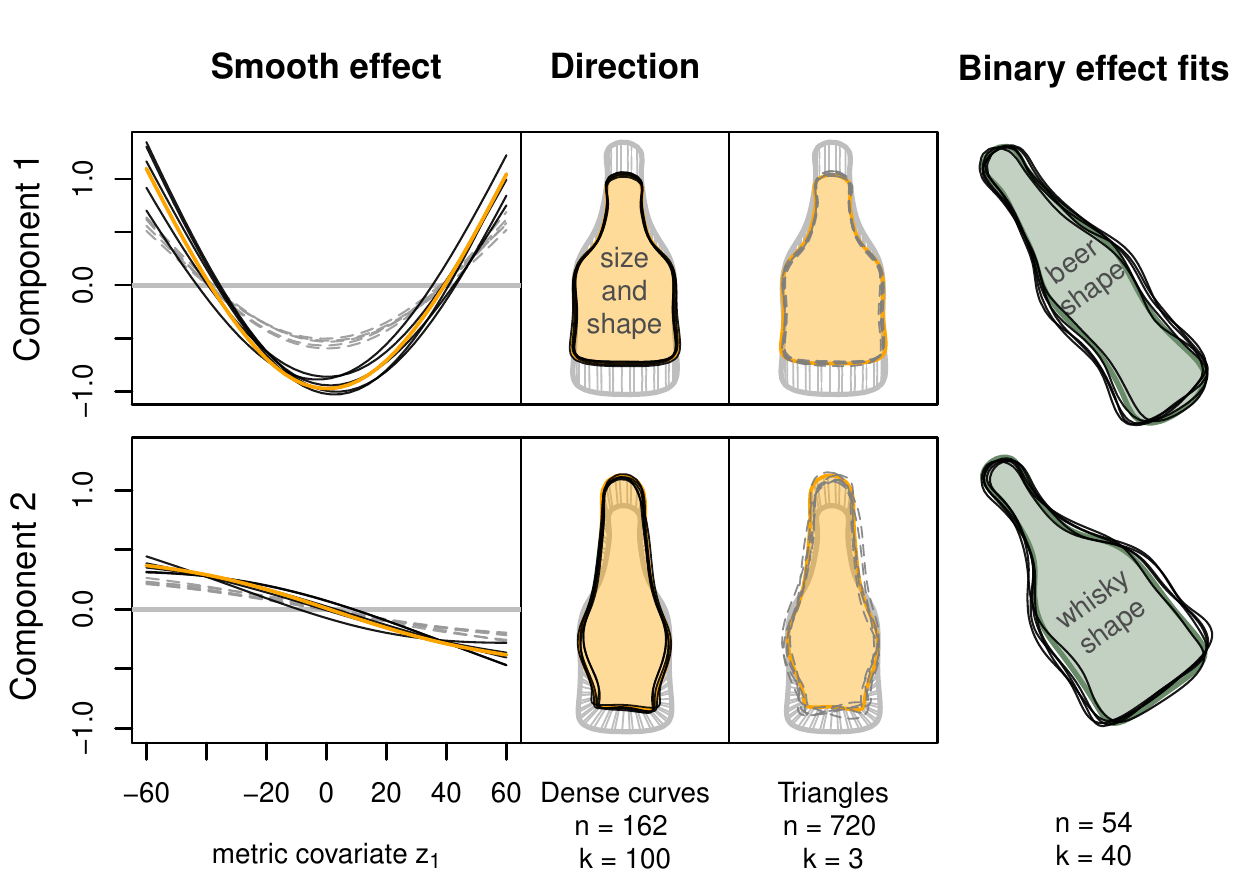} 
	\includegraphics[width = 0.256\textwidth]{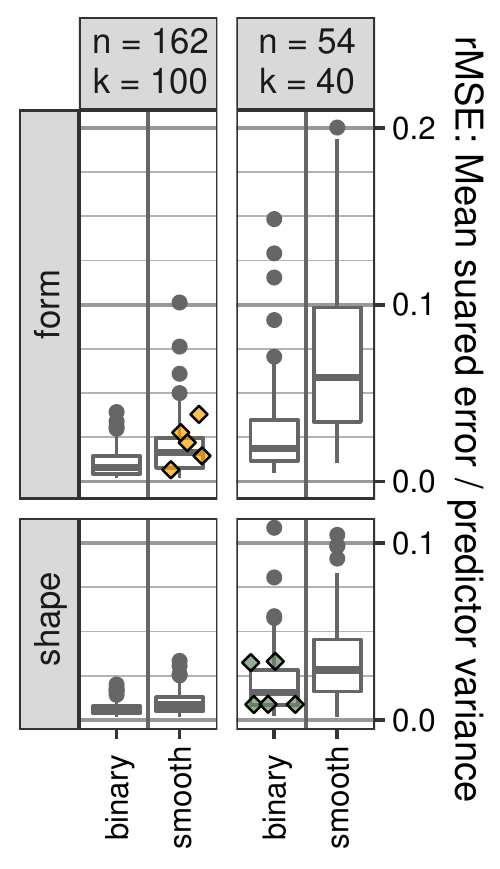}
	\caption{\textit{Left:} First \textit{(row 1)} and second \textit{(row 2)} main components of the smooth effect $f_1(z_1)$ in the form scenario obtained from TP factorization. Normalized component directions are visualized as bottle outlines after transporting them to the true pole \textit{(gray solid outline)}. Underlying truth \textit{(orange solid lines / areas)} are plotted together with
		five example estimates for $n=162$ and $k=100$ \textit{(black solid lines)} and the extremely sparse $k_i=3$ setting \textit{(gray dashed lines)}. 
		\textit{Center:} Conditional means for both bottle types with fixed metric covariate $z_1=0$ in the shape scenario with $n=54$ and $k=40$. Five example estimates \textit{(black solid outlines)} are plotted in front of the underlying truth \textit{(olive-green areas)}.
		\textit{Right:} rMSE of shown example estimates \textit{(jittered colored diamonds)} contextualized with boxplots of rMSE distributions observed in respective simulation scenarios. 
	}\label{fig:sim_effect_plots} 
\end{figure}
%
\textbf{Form scenario:} In the form scenario, the smooth covariate effect $f_1$ offers a particularly clear interpretation. TP factorization decomposes the true effect into its two relevant components, where the first (major) component corresponds to the bare projection of the tilted outline in 3D into the 2D image plane and the second to additional perspective transformations (Fig.\ \ref{fig:sim_effect_plots}). For this effect, we observe a median relative mean squared error  $\operatorname{rMSE}(\hat{h}_j) = \sum_{i=1}^n \|\hat{h}_j(\bx_i) - h_j(\bx_i)\|_i^2 / \sum_{i=1}^n \|h(\bx_i)\|_i^2$ of about 3.7\% of the total predictor variance for small data settings with $n=54$ and ${k} = 100$ (5.9\% with $ k = 40$), which reduces to 1.5\% for $n=162$ (for both $k = 40$ and $ k = 100$). It is typical for functional data that, from a certain point, adding more (highly correlated) evaluations per curve leads to distinctly less improvement in the model fit than adding further observations \citep[compare, e.g., also][]{Stoecker2019FResponseLSS}. 
In the extreme $k_i=3$ scenario, we obtain an rMSE of around 15\%, which is not surprisingly considerably higher than for the moderate settings above. 
Even in this extreme setting (Fig.\ \ref{fig:sim_effect_plots}), the effect directions are captured well, while the size of the effect is underestimated. Rotation alignment based on only three points (which are randomly distributed along the curves) might considerably differ from the full curve alignment, and averaging over these sub-optimal alignments masks the full extend of the effect. Still, results are very good given the sparsity of information in this case.
Having a simpler form, the binary effect $\beta_\catidx$ is also estimated more accurately with an rMSE of around 1.5\% for $n=54$, $ k = 100$ (1.9\% for $ k = 40$) and less than 0.8\% for $n=162$ (for both $ k = 40$ and $ k = 100$). The pole estimation accuracy varies on a similar scale. 

\textbf{Shape scenario:}
Qualitatively, the shape scenario shows a similar picture. For $ k = 40$, we observe median rMSEs of 2.8\% ($n = 54$) and 2.2\% ($n = 162$) for $f_1(z_1)$, and 1.5\% and 0.6\% for the binary effect $\beta_\catidx$. For $k = 100$, accuracy is again slightly higher.

\textbf{Nuisance effects and integration weights:} 
Nuisance effects in the model where generally rarely selected and, if selected at all, only lead to a marginal loss in accuracy. The constant effect is only selected sometimes in the extreme triangle scenarios, when pole estimation is difficult. 
We refer to \citet{brockhaus2017boosting}, who perform gradient boosting with functional responses and a large number of covariate effects with stability selection, for simulations with larger numbers of nuisance effects and further discussion in a related context, as variable selection is not our main focus here. 
Finally, simulations indicate that inner product weights implementing a trapezoidal rule for numerical integration  are slightly preferable for typical grid sizes ($ k = 40, 100$), whereas weights of $1/k_i$ equal over all grid points within a curve gave slightly better results in the extreme $k_i=3$ settings.

All in all, the simulations show that Riemannian $L_2$-Boosting can adequately fit both shape and form models in a realistic scenario and captures effects reasonably well even for a comparably small number of sampled outlines or evaluations per outline. 

\section{Discussion and Outlook}
\label{chap_discussion}

Compared to existing (landmark) shape regression models, the presented approach extends  linear predictors to more general additive predictors including also, e.g., smooth nonlinear model terms and interactions, and yields the first regression approach for functional shape as well as  form responses. 
Moreover, we propose novel visualizations based on TP factorization that, similar to \new{FPC}\old{functional principal component} analysis, enable a systematic decomposition of the variability explained by an additive effect on tangent space level. 
Yielding meaningful coordinates for model effects,
its potential for visualization will be useful also for \new{FAMs}\old{functional additive models} in linear spaces\old{.} \new{and also beyond our model framework, such as we exemplarily illustrate for the non-parametric approach of \citet{jeon2020additiveHilbertian} in Supplement \ref{sec:kernelTPfactorization}}. 

Instead of operating on the original evaluations $\by_i\in\C^{k_i}$ of response curves $y_i$ as in all applications above, another frequently used approach expands $y_i$, $i=1,\dots,n$, in a common basis first, before carrying out statistical analysis on coefficient vectors (compare \old{e.g. }\citet{RamsaySilverman2005, Morris2015} \new{and \citet{MullerYao2008FAM} for smoothing spline, wavelet or FPC representations} in FDA or \citet{Bonhommeetal2014RPackMomocs} in\old{ the} shape analysis\old{ literature}). \old{Additive s}\new{S}hape/form regression on the coefficients is, in fact, a special case of our approach, where the inner product is evaluated on the coefficients instead of \old{on curve }evaluations\old{, since the entire quotient space geometry and model is determined by the Hilbert space $\cY$} (\old{for more details see }Supplement \ref{sec:coefficientlevel}). 

The proposed model is motivated by geodesic regression. 
However, in the multiple linear predictor, a linear effect of a single covariate does, in general, not describe a geodesic for fixed non-zero values of other covariate effects. Or put differently, $\Exp_{[p]}\left( h_1 + h_2\right) \neq \Exp_{\Exp_{[p]}\left(h_1\right)}\left(h_2\right) \neq \Exp_{\Exp_{[p]}\left(h_2\right)}\left(h_1\right)$ in general. 
Thus, hierarchical geodesic effects of the form $\Exp_{\Exp_{[p]}\left(h_1\right)}\left(h_2\right)$, relevant, i.a., in mixed models for hierarchical/longitudinal study designs \citep{KimEtAl2018RiemannianNonlinearMixedEffectsModels}, present an interesting future extension of our model.\old{ For estimation, developments in  model-based boosting for generalized additive models for location, scale and shape \citep[GAMLSS,][]{Thomas2016} might offer the needed infrastructure.} 
Moreover, an ``elastic'' extension\old{ for functional shape and form regression} based on the square-root-velocity framework \citep{SrivastavaKlassen2016} presents a promising direction for future research\new{, as do other manifold responses}. 
\old{The generality of underlying concepts also encourages  investigation of our approach for response objects living in other Riemannian manifolds, such as SPD matrices occurring, e.g., as covariance matrices in functional connectivity studies or diffusion tensor imaging. }

\section*{Acknowledgement}
We sincerely thank Nadja P\"ollath for providing carefully recorded sheep astragalus data and important insights and comments, \new{and} Sophia Schaffer for running and discussing cell simulations and providing fully processed cell outlines\new{.}\old{, and Lisa Steyer for frequent and fruitful discussions.}
Moreover, we gratefully acknowledge funding by grant GR 3793/3-1 from the German research foundation (DFG).

\bigskip
\begin{center}
{\large\bf SUPPLEMENTARY MATERIAL}
\end{center}
Supplementary material with further details is provided in an online supplement.



\bibliographystyle{chicago}
\bibliography{literatur}

\appendix
\setcounter{section}{18}
\setcounter{figure}{0}
\makeatletter 
\renewcommand{\thefigure}{S\@arabic\c@figure}
\renewcommand{\thetable}{S\@arabic\c@table}
\makeatother

\section{Online Supplementary Material to}
{\Large Functional additive models on manifolds of planar shapes and forms}\newline
by Almond St\"ocker\new{, Lisa Steyer} and Sonja Greven

\subsection{Geometry of functional forms and shapes}

\subsubsection{Translation, rotation and re-scaling as normal subgroups}
\label{sec:groupactions}
We consider the following invariances of a response curve $y \in\cY$ with respect to the transformations $\cY \rightarrow \cY$ given by 
the group actions of translation $\trl = \{y \overset{\trl_\gamma}{\longmapsto} y + \gamma\,\neutral : \gamma\in\C\}$ with some $\neutral\in\cY\setminus\{\zero\}$ (for curves typically $\neutral: t\mapsto \frac{1}{\|t\mapsto 1\|}$ the real constant function of unit norm), re-scaling $\scl = \{y \overset{\scl_\lambda}{\longmapsto} \lambda \cdot (y - \zero_y) + \zero_y : \lambda \in \R^+\}$ around a reference point $\zero_y \in \C$, and rotation $\rot = \{y \overset{\rot_u}{\longmapsto} u\cdot (y-\zero_y) + \zero_y : u \in \S^1\}$ around $\zero_y$ with $\S^1 = \{u\in\C:|u|=1\} = \{\exp(\omega\i):\omega\in\R\}$ the circle group reflecting counterclockwise rotations by $\omega$ radian measure.
In the literature, the reference point is usually omitted setting $\zero_y = \zero$, which can be done without loss of generality under translation invariance (i.e. in particular for shapes/forms). 
However, keeping other possible combinations of invariances in mind, we explicitly refer to an individual reference point and suggest 
the centroid $\zero_y = \langle {\neutral}\,, y \rangle \neutral$ or, more generally, $\zero_y = a(y)\, \neutral$ for some linear functional $a:\cY \rightarrow \R$.
Assuming $\trl_\gamma(0_{y}) = 0_{\trl_\gamma(y)}$, as for the centroid, the definition of re-scaling and rotation around $\zero_y$ ensures that $\trl_\gamma$, $\scl_\lambda$ and $\rot_u$ commute -- and that $\trl$, $\rot$ and $\scl$ present normal subgroups of the combined group actions $\{y \mapsto \lambda u y + \gamma : \gamma\in\C, \lambda\in\R^+, u\in\S^1\}$ of shape invariances. 
Thus, the combined group actions can be written as the direct product (or direct sum) $\trl \dir \scl \dir \rot = \{\trl_\gamma \circ \scl_\lambda \circ \rot_u: \gamma\in\C, \lambda\in\R^+, u\in\S^1\} \cong \C \times \R^+ \times \S^1$ and invariances with respect to $\trl$, $\scl$, $\rot$ can be modularly accounted for in arbitrary order. $\trl \dir \rot$, for instance, describe rigid motions.
The ultimate response object is then given by the \emph{orbit} $[y]_{G}=\{g(y) : g\in G\}$ (or short $[y]$), i.e. the equivalence class with respect to the direct sum $G$ generated by the chosen combination of $\trl$, $\scl$ and $\rot$. 
$[y]_{\trl \dir \scl \dir\rot}$ is referred to as the \emph{shape} of $y$ and $[y]_{\trl \dir\rot}$ as its \emph{form} or \emph{size-and-shape} \citepsup[compare][]{DrydenMardia2016ShapeAnalysisWithApplications}; studying  $[y]_{\scl}$ is closely related to \emph{directional} data analysis \citepsup{MardiaJupp2009DirectionalStatistics} where the direction of $y$ is analyzed independent of its size $\|y\|$. 

\subsubsection{Parallel transport of form tangent vectors}
\label{sec:transport}

To confirm that the parallel transport in the form space $\cY^*_{/\trl \dir \rot}$ can be carried out via representatives in $\cY$ as described in the main manuscript, we closely follow \citesup{Huckemann2010intrinsicMANOVA} in their derivation of shape parallel transport. Necessary differential geometric notions and statements are briefly introduced in the following before stating the main result in Lemma \ref{lemma:formtransport}. For a more profound introduction, we recommend \citetsup[in particular, p. 21, 43, 93, 124, 337, 402]{Lee2018RiemannianManifolds}, as well as \citetsup{tu2011manifolds} for an illustrative introduction into some of the concepts, and \citetsup[in particular, p.103- 107]{Klingenberg1995RiemannianGeometry} for an introduction in the light of potentially infinite dimensional manifolds.


\newtheorem{theorem}{Theorem}
\newtheorem{corollary}{Corollary}
\newtheorem{lemma}{Lemma}
\newtheorem{proposition}{Proposition}
\newtheorem{remark}{Remark}

\newcommand{\diff}{d}
\newcommand{\ree}{{\operatorname{Re}}}
\newcommand{\imm}{{\operatorname{Im}}}
\newcommand{\cN}{\mathcal{N}}
\newcommand{\cH}{\mathcal{H}}

\newcommand{\pullb}[1]{\widetilde{#1}}
\newcommand{\pullback}[1]{\pullb{\left(#1\right)}}

The entire argument crucially relies on properties known for Riemannian submersions between differentiable manifolds $\widetilde{\cM}$ and $\cM$, which allow to relate the structure of $\cM$ back to  $\widetilde{\cM}$. A \textit{submersion} is a smooth surjective function $\Phi: \pullb{\cM} \rightarrow \cM$, for which also the differential $\diff\Phi: T_{\widetilde{q}} \widetilde{\cM} \rightarrow T_{q} \cM$, $\widetilde{q}\in\widetilde{\cM}$, $q=\Phi{(\widetilde{q})}$, is surjective at each $\widetilde{q}\in\widetilde{\cM}$. For $q \in \cM$, the $\Phi^{-1}(\{q\})$ are submanifolds of $\cM$, and $T_{\widetilde{q}} \widetilde{\cM} = T_{\widetilde{q}}\Phi^{-1}(\{q\}) \oplus H_{\widetilde{q}}\widetilde{\cM}$ can be decomposed into the \textit{vertical space} $T_{\widetilde{q}}\Phi^{-1}(\{q\}) = \ker(d\Phi)$ and its orthogonal complement $H_{\widetilde{q}}\widetilde{\cM}$, the \textit{horizontal space}. When restricted to the horizontal space, $d\Phi\big|_{H_{\widetilde{q}}\widetilde{\cM}} : H_{\widetilde{q}}\widetilde{\cM} \rightarrow T_{q} \cM$ presents a linear isomorphism. A submersion $\Phi$ is called \textit{Riemannian submersion} if $d\Phi\big|_{H_{\widetilde{q}}\widetilde{\cM}}$ is also isometric. It gives rise to an identification $T_{q} \cM \cong H_{\widetilde{q}}\widetilde{\cM}$ of tangent spaces of $\cM$ with horizontal spaces on $\widetilde{\cM}$. Such an identification underlies the presentation of the response geometry in Section \ref{chap_diffgeo} of the main manuscript. \newline
By construction, the quotient map $\Phi: \cY^* \rightarrow \cY^*_{/\trl \dir \rot}, y \mapsto [y]$ presents a Riemannian submersion: 
Since $[p]=\{u p + \gamma \neutral: u\in\S^1, \gamma \in \C\}$ embeds $\S^1\times \R^2$ in $\cY$, and, since the tangent spaces of $\S^1$ and $\R^2$ are well-known, the vertical space is given by $T_p[p] \cong \{\lambda p \i + \gamma\, \neutral : \lambda \in  \R , \gamma \in  \C \} \subset \cY$, with orthogonal complement $H_p\cY^* \cong \{y\in\cY : \langle y, \neutral \rangle = 0,\ \im{\langle y, p \rangle} = 0 \}$ (see also Figure \ref{fig:quotientgeometry} in the main manuscript for an illustration).
While $\Phi$ is obviously surjective, surjectivity and isometry of $d\Phi|_{H_{p}\cY^*}$ can be seen by expressing $\Phi$ in terms of the charts for $\cY^*_{/\trl \dir \rot}$: for a given $p \in [p] \in \cY_{/\trl \dir \rot}^*$, the map $\widetilde{(\cdot)}: [y] \mapsto \widetilde{y}^{\trl \dir \rot}$ provides a chart $\cU_{[p]} \rightarrow \cV_{p}$, i.e.\ an isomorphism from $\cU_{[p]} = \{[y] \in \cY_{/\trl \dir \rot}^* : \langle \widetilde{p}^{\trl}, \widetilde{y}^{\trl}\rangle \neq 0\}$ to $\cV_{p} = \{y\in\cY : \im{\langle p, y\rangle} = 0, \re{\langle p, y\rangle} > 0, \langle \neutral, y\rangle = 0\}$ used to establish the differential structure on $\cY_{/\trl \dir \rot}^*$. Expressed in this chart, $\widetilde{\Phi}(y) = \widetilde{(\cdot)} \circ \Phi(y) =  \widetilde{y}^{\trl \dir \rot}$ is the identity for all $y\in\cV_{p}\subset\Phi^{-1}\left(\cU_{[p]}\right)$. Thus, since $T_p\cV_p = H_p\cY^*$, also $d\widetilde{\Phi}\big|_{H_p\cY^*}$ is the identity, which is obviously an isometric isomorphism. The latter carries over to $d\Phi\big|_{H_p\cY^*}$ independent of the given chart.\newline
The isometric isomorphism $d\Phi\big|_{H_p\cY^*}: H_p\cY^* \rightarrow T_{[p]}\cY_{/\trl \dir \rot}^*$ yields the identification \linebreak 
$T_{[p]}\cY_{/\trl \dir \rot}^* \cong \{y\in\cY : \langle y, \neutral \rangle = 0,\ \im{\langle y, p \rangle} = 0 \}$, which we rely on in the main manuscript.
Unlike there, we denote $d\Phi\big|_{H_y\cY^*}^{-1}: \xi \mapsto \widetilde{\xi}$ also for tangent vectors in the following, to make the identification of $\xi=d\Phi(\pullb{\xi}) \in T_{[y]}\cY_{/\trl \dir \rot}^*$ with the corresponding $\pullb{\xi} \in H_y\cY^*$, usually referred to as \textit{horizontal lift}, explicit in the notation. 

The \textit{covariant derivative} (Levi-Civita connection) $\nabla^\cM_V W \in T\cM$ of a vector-field $W\in T\cM$ along a vector-field $V\in T\cM$ provides a derivative of vector-fields in the tangent bundle $T\cM = \{T_q\cM : q\in\cM\} $ of a Riemannian manifold $\cM$. As a derivation in $W$ and a linear function in $V$, $\nabla^\cM$ fulfills a set of properties identifying it as unique generalization of ordinary directional derivatives of the components of $W:q\mapsto W_q\in T_q\cM$ into the direction $V_q \in T_q\cM$. 
For a submanifold $\cM$ of a linear space $\cY$, $\nabla^\cM_V W$ corresponds to the ordinary directional derivative orthogonally projected into $T_q\cM$. 
For the linear case (with $\cM=\cY$), the covariant derivative of 
a vector field $W(\tau) := W_{c{(\tau)}}$ along a differentiable curve $c(\tau)$ is directly given as
\begin{equation}
	\label{covariant_derivative}
	\nabla^\cY_{\dot{c}(t)} W(\tau) = \dot{W}(\tau) = \frac{d}{d\tau} W(\tau).
\end{equation}

In analogy to straight lines, geodesic curves $c(\tau)$ are characterized by 
\begin{equation*}
	\nabla^\cM_{\dot{c}(\tau)} \dot{c}(\tau) = \zero,
\end{equation*}
i.e. curves with zero `second derivative'.
More generally, a vector-field $W$ is called parallel along a curve $c(\tau)$ if
\begin{equation}
	\label{parallel}
	\nabla^\cM_{\dot{c}(\tau)} W(\tau) = \zero.
\end{equation}
According to that the parallel transport $\Transp_{q, q'}^c : T_q\cM \rightarrow T_{q'}\cM$ along a curve $c: [\tau_0, \tau_1] \rightarrow \cM$ between $c(\tau_0) = q,\ c(\tau_1) = q' \in \cM$ is defined to map tangent vectors $\varepsilon = W(\tau_0) \mapsto \varepsilon' = W(\tau_1)$ for some  vector field $W$ parallel along $c$ (fulfilling Equation \ref{parallel}). If the curve $c$ is clear from context, we omit it in the notation. This is especially the case in the following, where $c$ can be chosen as the unique geodesic between two forms $[p]$ and $[p']$ with $\langle p, p' \rangle\neq 0$, yielding a canonical connection (in this case, $c$ corresponds to the line between $p$ and the aligned $\tilde{p}'$; for $\langle p, p'\rangle=0$, by contrast, it is easy to see that for each $u\in\S^1$ the line between $p$ and $up'$ corresponds to a different geodesic; the second case can, however, be neglected).

The possibility to effectively carry out the parallel transport between forms $[p], [p']$ on suitable representatives $p, p'\in\cY^*$ stems from the following theorem and subsequent Corollary \citepsup[compare, e.g,][p. 103-105]{Klingenberg1995RiemannianGeometry}.

\begin{theorem}
	\label{submersion_theorem}
	Let $\Phi: \pullb{\cM} \rightarrow \cM$ be a Riemannian submersion between manifolds $\widetilde{\cM}$ and $\cM$, and $V,W \in T\cM$ vector-fields. Then
	\begin{equation*}
		\nabla^{\pullb{\cM}}_{\pullb{V}} \pullb{W} = \pullback{\nabla^\cM_{V}W} + \frac{1}{2}[\pullb{V}, \pullb{W}]^\perp
	\end{equation*}
	where $\pullb{Z} \in H\pullb{\cM}$ denotes the horizontal lift of $Z  \in T\cM$ to the horizontal bundle $H\pullb{\cM} = \{H_{\pullb{p}}\pullb{\cM} : \pullb{p} \in \pullb{\cM} \}$, $Z= \diff\Phi \left(\pullb{Z}\right)$, and $[\pullb{V}, \pullb{W}]^\perp$ is the the Lie bracket $[\pullb{V}, \pullb{W}] = \pullb{V} \circ \pullb{W} - \pullb{W} \circ \pullb{V}$ orthogonally projected $(\cdot)^\perp: T\pullb{\cM}\rightarrow ker\left(\diff\Phi\right)$ to the vertical space.
\end{theorem}

\begin{corollary}
	\label{corollary}
	Let $\Phi: \pullb{\cM} \rightarrow \cM$ be a Riemannian submersion and $c: (\tau_0, \tau_1) \rightarrow \cM$ a smooth curve on $\cM$ with $\pullb{c}: (\tau_0, \tau_1) \rightarrow \pullb{\cM}$ its horizontal lift, i.e., $\Phi\circ\pullb{c} = c$, $d\Phi\circ\dot{\pullb{c}} = \dot{c}$ and $\dot{\pullb{c}}(\tau)\in H_{\pullb{c}(\tau)}\pullb{\cM}$ horizontal (i.e. $\dot{\pullb{c}} = \pullb{\dot{c}}$). Then
	\begin{enumerate}[label = \roman*)]
		\item \label{corollary_parallel} a vector-field $W = \diff \Phi \circ \pullb{W} \in T \cM$ along $c$ is parallel if and only if
		\begin{equation*}
			\nabla_{\dot{\pullb{c}}}^{\pullb{\cM}} \pullb{W} = \frac{1}{2}[\dot{\pullb{c}}, \pullb{W}]^{\perp}
		\end{equation*}
		for the horizontal vector-field $\pullb{W} \in T\pullb{\cM}$ along $\pullb{c}$.
		\item \label{corollary_geodesic}$c$ is a geodesic if and only if $\pullb{c}$ is a geodesic.
	\end{enumerate}
\end{corollary}

While \ref{corollary_parallel} yields the basis for confirming the parallel transport computation, \ref{corollary_geodesic} is the underlying fact behind the identification of geodesics in form and shape spaces with geodesics of suitably aligned representatives. 
Note that, while \citetsup{Huckemann2010intrinsicMANOVA} generally restrict their discussion to finite dimensional manifolds, the theorem does in fact not have this restriction. Based on these preparations, we can now verify the presented parallel transport along the lines of \citetsup{Huckemann2010intrinsicMANOVA}, but for forms rather than shapes and explicitly based on a separable Hilbert space $\cY$ rather than on $\C^k$. Note that, while identifying $H_p\cY^* \cong T_{[p]}\cY_{/\trl \dir \rot}^*$ and $\varepsilon \cong \diff \Phi(\varepsilon)$ in the main manuscript, they are distinguished here for clarity.

\begin{lemma}
	\label{lemma:formtransport}
	Let $p, p' \in \cY^*$ with $\langle p, p' \rangle \neq 0$ centered and mutually rotation aligned representatives of forms $[p], [p'] \in \cY^*_{/\trl \dir \rot}$ (i.e. $p=\widetilde{p}^{{\trl \dir \rot}}$ for notational simplicity and $p'$ accordingly), let $\varepsilon \in T_{[p]}\cY^*_{/\trl \dir \rot}$ with horizontal lift $\tilde{\varepsilon}\in H_p \cY^*$, and let $\Phi:y \mapsto [y]$ denote the quotient map. Then  
	\begin{equation}
		\label{transport_equation}
		\Transp_{[p], [p']}\left(\varepsilon\right) = \diff\Phi\left( \tilde{\varepsilon} - \im{\langle p'/\|p'\|, \tilde{\varepsilon}\rangle} \frac{p/\|p\| + p'/\|p'\|}{1+\langle p/\|p\|, p'/\|p'\|\rangle} \i \right)
	\end{equation}
	implements the form parallel transport via its horizontal lift. 
	
\end{lemma}
\begin{proof}
	For $\langle p, p' \rangle \neq 0$ aligned and centered, the unique unit-speed geodesic (uniqueness can be seen using Corollary \ref{corollary} \ref{corollary_geodesic}) between $[p]$ and $[p']$ is described by $\tau \rightarrow [p + \tau \frac{p'-p}{\|p'-p\|}]$. Yet, to simplify the argument, we choose a unit-angular speed parameterization instead. It takes the form $c(\tau) := [\pullb{c}(\tau)] :=  [\rho(\tau) \gamma(\tau)]$ with $\gamma(\tau) = \cos(\tau) \beta + \sin(\tau) \beta'$ where $\beta = \frac{p}{\|p\|}$ and $\beta' = \frac{p' - \langle \beta, p'\rangle \beta}{\|p' - \langle \beta, p'\rangle \beta\|}=\frac{\frac{p'}{\|p'\|} - \langle \frac{p}{\|p\|}, \frac{p'}{\|p'\|}\rangle \frac{p}{\|p\|}}{\|\frac{p'}{\|p'\|} - \langle \frac{p}{\|p\|}, \frac{p'}{\|p'\|}\rangle \frac{p}{\|p\|}\|}$ form an orthonormal basis of the real plain containing the horizontal geodesic.
	With $\widetilde{c}(0) = p$ and $\widetilde{c}(\arccos \langle \frac{p}{\|p\|}, \frac{p'}{\|p'\|}\rangle)=p'$, $\widetilde{c}(\tau)$ describes the line connecting $p$ and $p'$ in polar coordinates.
	$[\gamma(\tau)]_{\trl \dir \rot \dir \scl}$ corresponds to the shape geodesic between $[p]_{\trl \dir \rot \dir \scl}$ and $[p']_{\trl \dir \rot \dir \scl}$, and $\rho(\tau) = \|\pullb{c}(\tau)\|$ reflects the size of the geodesic $c(\tau)$. An explicit definition of $\rho(\tau)$ is not needed.
	
	Due to the alignment of $p$ and $p'$,  $\dot{\gamma}(\tau)$ and also
	\begin{equation}
		\label{W}
		\pullb{W}(\tau) := \tilde{\varepsilon} + \im{\langle \beta', \tilde{\varepsilon}\rangle} \left(\dot{\gamma}(\tau) - \beta' \right) \i
	\end{equation}
	are horizontal along $\pullb{c}(\tau)$, i.e. $\widetilde{W}(\tau) \in H_{\widetilde{c}(\tau)} \cY^*$ for each $\tau$, if $\tilde{\varepsilon}$ is horizontal, i.e. if $\im{\langle p, \tilde{\varepsilon}\rangle} = \langle \neutral, \tilde{\varepsilon} \rangle =  0$. More concretely, this holds as
	\begin{align*}
		\im{\langle \pullb{c}(\tau), \pullb{W}(\tau) \rangle} 
		&= \rho(\tau)\; \left( \im{\langle \gamma(\tau), \tilde{\varepsilon} \rangle}  + \im{\langle \beta', \tilde{\varepsilon}\rangle} \re{\langle \gamma(\tau), \dot{\gamma}(\tau) - \beta' \rangle} \right)\\  
		&\overset{\tilde{\varepsilon} \text{ horizontal}}{=} \rho(\tau)\; \left(\sin(\tau) \im{\langle \beta', \tilde{\varepsilon} \rangle} + \im{\langle \beta', \tilde{\varepsilon}\rangle} (0 - \sin(\tau) \underbrace{\|\beta'\|^2}_{=1}) \right) = 0
	\end{align*}
	and, obviously, also $\langle \neutral, \pullb{W}(\tau) \rangle = 0$ as this is the case for all involved vectors.
	Moreover, $\pullb{W}$ is smooth and $\tilde{\varepsilon} \mapsto \pullb{W}\left(\arccos \langle \frac{p}{\|p\|}, \frac{p'}{\|p'\|}\rangle\right)$ yields the transport formulated in Equation \eqref{transport_equation}, which follows from basic trigonometric relations.
	In detail, it follows from plugging
	\begin{align*}
		\dot{\gamma}\left(\arccos \langle \frac{p}{\|p\|}, \frac{p'}{\|p'\|}\rangle\right) -\beta' &= \langle \frac{p}{\|p\|}, \frac{p'}{\|p'\|}\rangle\ \beta' - \sqrt{1-\langle \frac{p}{\|p\|}, \frac{p'}{\|p'\|}\rangle^2}\ \beta - \beta'\\ 
		&=  \left(\langle \frac{p}{\|p\|}, \frac{p'}{\|p'\|}\rangle - 1\right) \overbrace{\frac{\frac{p'}{\|p'\|} - \langle \frac{p}{\|p\|}, \frac{p'}{\|p'\|}\rangle \frac{p}{\|p\|}}{\sqrt{1-\langle \frac{p}{\|p\|}, \frac{p'}{\|p'\|}\rangle^2}}}^{\beta'}\  - \sqrt{1-\langle \frac{p}{\|p\|}, \frac{p'}{\|p'\|}\rangle^2}\ \frac{p}{\|p\|}\\
		&=  \frac{\left(\langle \frac{p}{\|p\|}, \frac{p'}{\|p'\|}\rangle - 1\right) \frac{p'}{\|p'\|}}{\sqrt{1-\langle \frac{p}{\|p\|}, \frac{p'}{\|p'\|}\rangle^2}} \\
		&\qquad +  \frac{ - \langle \frac{p}{\|p\|}, \frac{p'}{\|p'\|}\rangle^2 \frac{p}{\|p\|} + \langle \frac{p}{\|p\|}, \frac{p'}{\|p'\|}\rangle \frac{p}{\|p\|} - \left(1-\langle \frac{p}{\|p\|}, \frac{p'}{\|p'\|}\rangle^2\right) \frac{p}{\|p\|}}{\sqrt{1-\langle \frac{p}{\|p\|}, \frac{p'}{\|p'\|}\rangle^2}} \\
		&= \frac{-\left(1-\langle \frac{p}{\|p\|}, \frac{p'}{\|p'\|}\rangle\right) \left(\frac{p'}{\|p'\|} +  \frac{p}{\|p\|}\right)}{\sqrt{1-\langle \frac{p}{\|p\|}, \frac{p'}{\|p'\|}\rangle^2}} 
	\end{align*}
	and 
	\begin{align*}
		\im{\langle \beta', \tilde{\varepsilon} \rangle} \overset{\text{$\tilde{\varepsilon}$ horizontal}}{=} \frac{\im{\langle p', \tilde{\varepsilon}\rangle}}{\sqrt{1-\langle \frac{p}{\|p\|}, \frac{p'}{\|p'\|}\rangle^2}}
	\end{align*}
	into the definition of $W(\tau)=d\Phi (\widetilde W(\tau))$ using \eqref{W}.\\

	\newcommand{\imform}{\diff^\mathcal{I}}
	\newcommand{\reform}{\diff^\mathcal{R}}
	\newcommand{\imvec}{\partial_\mathcal{I}}
	\newcommand{\revec}{\partial_\mathcal{R}}

	Hence, due to Corollary \ref{corollary} \ref{corollary_parallel}, we mainly need to show
	\begin{equation}
		\label{targetequation}
		\nabla_{\dot{\pullb{c}}}^{\cY^*} \pullb{W} = \frac{1}{2} [\dot{\pullb{c}}, \pullb{W}]^\perp 
	\end{equation}
	where the left-hand side may directly be computed as
	\begin{align*}
		\nabla_{\pullb{\dot{c}}(\tau)}^{\cY^*} \pullb{W}(\tau) &\overset{\eqref{covariant_derivative}}{=} \dot{\pullb{W}}(\tau) \overset{\eqref{W}}{=} -\im{\langle \beta', \tilde{\varepsilon}\rangle} \gamma(\tau)  \i
	\end{align*}
	since $\ddot{\gamma}(\tau) = - \gamma(\tau)$.\\ 
	
	On the right-hand side, the orthogonal projection of a vector-field $V(\tau):= V_{\pullb{c}(\tau)} \in T_{\pullb{c}(\tau)}\cY^*$ along $\pullb{c}(\tau)$ into the vertical spaces (of which $\{ \gamma(\tau), \neutral, \i\neutral\}$ constitute an orthonormal basis) is given by
	\begin{align*}
		V^\perp(\tau) &= \frac{\re{\langle \i\, \pullb{c}(\tau), V(\tau) \rangle}}{\|\pullb{c}(\tau)\|^2}  \pullb{c}(\tau) \i + \langle \neutral, V(\tau) \rangle\, \neutral \nonumber\\
		&= \frac{\omega^{\rot}(V(\tau))}{\rho(\tau)} \gamma(\tau) \i  + \omega^{\trl}(V(\tau))\, \neutral
	\end{align*}
	with the 1-forms $\omega^{\rot}$ and $\omega^{\trl}$ defined as
	\begin{align*}
		\omega^{\rot}(V_p) &:= \re{\langle \i\, p, V_p \rangle} = \re{-\i\, \langle p, V_p \rangle} \\ 
		&= \re{-\i\, \left(\re{\langle p, V_p \rangle}+\im{\langle p, V_p \rangle} \i\right)}\\
		&= \re{-\i\, \re{\langle p, V_p \rangle} + \im{\langle p, V_p \rangle}}\\
		&= \im{\langle p, V_p \rangle}.
	\end{align*}
	and $\omega^{\trl}(V_p) = \langle \neutral, V_p \rangle$ for $p\in\cY^*$.
	
	Thus, to confirm \eqref{targetequation} and complete the proof, it remains to show $\omega^{\rot}([\dot{\pullb{c}}, \pullb{W}]) = - 2 \im{\langle\beta', \tilde{\varepsilon}\rangle} \rho$ and $\omega^{\trl}([\dot{\pullb{c}}, \pullb{W}]) = 0$. For this, we use some statements on the exterior derivative $d\omega$ of a 1-form $\omega$ subsumed in the following auxiliary lemma (proven later):  
	
	\begin{lemma}
		\label{auxiliary}
		Let $V, W$ be smooth vector-fields.
		\begin{enumerate}[label = \roman*)]
			\item \label{auxiliary:lie} For any smooth 1-form $\omega$ it holds that 
			$\omega\left([V,W]\right) = V\left(\omega\left(W\right)\right) - W\left(\omega\left(V\right)\right) - \diff \omega \left(V,W\right)$.
			\item \label{auxiliary:omega} For $\omega^{\rot}$ defined above, $d\omega^{\rot}(V,W) = 2\, \im{\langle V, W\rangle}$.
			\item For $\omega^{\trl}$ defined above, $d\omega^{\trl}(V,W) = 0$.
		\end{enumerate}
	\end{lemma}
	
	Using further that
	\begin{align*}
		\omega^{\rot}\left(\dot{\pullb{c}}(\tau)\right)
		&= \im{\langle \gamma(\tau), \rho(\tau)\, \dot{\gamma}(\tau)\rangle} + \im{\langle \gamma(\tau), \dot{\rho}(\tau)\, \gamma(\tau)\rangle} = 0
	\end{align*}
	and  
	\begin{align*}
		\omega^{\rot}\left(\pullb{W}(\tau)\right)
		&= \rho(\tau) \im{\langle \gamma(\tau), \tilde{\varepsilon}\rangle} + \rho(\tau) \im{\langle\beta', \tilde{\varepsilon}\rangle}\underbrace{\langle \gamma(\tau),  \dot{\gamma}(\tau) - \beta'\rangle}_{\in\R}\\
		&= \rho(\tau)\left(\sin(\tau) \im{\langle \beta', \tilde{\varepsilon}\rangle} - \im{\langle \beta', \tilde{\varepsilon}\rangle} \sin(\tau)\right) = 0
	\end{align*}
	 we then have 
	\begin{align*}
		\omega^{\rot}([\dot{\pullb{c}}(\tau), \pullb{W}(\tau)]) &= 
		\underbrace{\dot{\pullb{c}}(\tau) \left( \omega^{\rot}(\pullb{W}(\tau)) \right)}_{=0} - 
		\underbrace{\pullb{W}(\tau) \left( \omega^{\rot}(\dot{\pullb{c}}(\tau)) \right)}_{=0} -
		d\omega^{\rot} \left(\dot{\pullb{c}}(\tau), \pullb{W}(\tau) \right)\\
			&= -2\,\big( \im{\langle \rho(\tau) \dot{\gamma}(\tau), \pullb{W}(\tau)\rangle} + \underbrace{\im{\langle \dot{\rho}(\tau) \gamma(\tau), \pullb{W}(\tau)\rangle}}_{=0 \text{  ($\pullb{W}$ horizontal, $\rho$ and $\dot{\rho}$ real)}} \big)\\
			&= -2\rho(\tau)\, \im{\langle \cos(\tau) \beta' - \sin(\tau) \beta, \tilde{\varepsilon}\rangle + \langle \dot{\gamma}(\tau), \im{\langle\beta', \tilde{\varepsilon}\rangle}\left(\dot{\gamma}(\tau) - \beta'\right) \i\rangle}\\
			&= -2\rho(\tau)\, \big( \cos(\tau) \im{\langle\beta', \tilde{\varepsilon}\rangle} +
			\im{\langle\beta', \tilde{\varepsilon}\rangle} \underbrace{\langle \dot{\gamma}(\tau), \dot{\gamma}(\tau) - \beta'\rangle}_{\in\R} \big) \\
			&= -2\rho(\tau)\, \big( \cos(\tau) \im{\langle\beta', \tilde{\varepsilon}\rangle} + \im{\langle\beta', \tilde{\varepsilon}\rangle} - \im{\langle\beta', \tilde{\varepsilon}\rangle} \cos(\tau) \big)\\ 
			&= -2\rho(\tau)\, \im{\langle\beta', \tilde{\varepsilon}\rangle} 
	\end{align*}
	and
	\begin{align*}
		\omega^{\trl}([\dot{\pullb{c}}(\tau), \pullb{W}(\tau)]) &= 
		\underset{=\langle \neutral, \pullb{W}\rangle=0}
		{\dot{\pullb{c}}(\tau) \left( \underbrace{\omega^{\trl}\left(\pullb{W}(\tau)\right)} \right)} - 
		\underset{=\langle \neutral, \dot{\pullb{c}}\rangle=0}
		{\pullb{W}(\tau) \left( \underbrace{\omega^{\trl}\left(\dot{\pullb{c}}(\tau)\right)} \right)} -
		\underbrace{d\omega^{\trl} \left(\dot{\pullb{c}}(\tau), \pullb{W}(\tau) \right)}_{= 0} = 0,\\
	\end{align*}
	where tangent vectors $\dot{\pullb{c}}(\tau)$ and $\pullb{W}(\tau)$ are interpreted as directional derivatives. These are the two equations that remained to show.
\end{proof}

\begin{proof}[Proof of Lemma \ref{auxiliary}]
	\begin{enumerate}[label = \roman*)]
		\item See, e.g., \citetsup{Lee2018RiemannianManifolds},  Proposition B.12 on page 402. This is a standard result. Note that based on an alternative (yet also common) definition of the wedge product and, hence, the exterior derivative, \citetsup{Huckemann2010intrinsicMANOVA} and others write $\omega\left([V,W]\right) = V\left(\omega\left(W\right)\right) - W\left(\omega\left(V\right)\right) - 2\, \diff \omega \left(V,W\right)$ instead. In this case, we also have $d\omega^{\rot}(V,W) = \, \im{\langle V, W\rangle}$ in \ref{auxiliary:omega} compensating for the different factor in the proof of Lemma \ref{lemma:formtransport}. 
		\item Let $\{e_r\}_r$ be an orthonormal $\C$-basis of $\cY$ (a complete orthonormal system existing since $\cY$ is separable) and $\{\vartheta^{(r)}(y)\}_r = \langle e_r, y \rangle$ the corresponding dual basis. The tangent vectors $\partial_{\ree, r}\big|_p \cong e_r$ and $\partial_{\imm, r}\big|_p \cong \i\, e_r$, $p\in\cY$ together form an $\R$-basis of $T_p\cY^* \cong \cY$. The dual 1-forms are given by $\diff^{\ree, r}(V_p) := V_p\left(\ree \circ \vartheta^{(r)}\right) \cong \ree \circ \vartheta^{(r)}(V_p)$ and $\diff^{\imm, r}(V_p) := V_p\left(\imm \circ \vartheta^{(r)}\right) \cong \imm\circ\vartheta^{(r)}(V_p)$ where we identify tangent vectors either with directional derivatives  $V_p(f)=\frac{d}{d\tau}\left( f\circ\Exp_{p}(\tau V_p)\right)\big|_{\tau=0}$ of functions $f:\cM\rightarrow\R$ or with elements of $\cY$, and the equality follows from $\cM=\cY^*$, and $\ree\circ\vartheta^{(r)}$, $\imm\circ\vartheta^{(r)}$ linear.
		With this given, we have 
		\begin{align}
			\label{oneform}
			\omega^{\rot}(V_p) &= \im{\langle \sum_r \langle e_r, p\rangle e_r, V_p\rangle}\\ \nonumber
			&= \sum_r \im{\langle p, e_r\rangle \langle e_r, V_p\rangle}\\ \nonumber
			&= \sum_r \re{\langle e_r, p\rangle} \im{\langle e_r, V_p\rangle} - \im{\langle e_r, p\rangle} \re{\langle e_r, V_p\rangle}\\  \nonumber
			&= \sum_r \ree\circ\vartheta^{(r)}(p)\, \diff^{\imm, j}\left(V_p\right) - \imm\circ\vartheta^{(r)}(p)\, \diff^{\ree, j}\left( V_p\right)  
		\end{align}
		and thus, expressing the exterior derivative in terms of wedge products 
		\begin{align*}
			\diff\omega^{\rot} &= \sum_r \sum_l \partial_{\ree, r} \left(\ree\circ\vartheta^{(r)}\right)\ \diff^{\ree, l}\wedge\diff^{\imm, r} + 
			\partial_{\imm, r} \left(\ree\circ\vartheta^{(r)}\right)\ \diff^{\imm, l}\wedge\diff^{\imm, r}\\
			&\quad -\partial_{\ree, r} \left(\imm\circ\vartheta^{(r)}\right)\ \diff^{\ree, l}\wedge\diff^{\ree, r} - 
			\partial_{\imm, r} \left(\imm\circ\vartheta^{(r)}\right)\ \diff^{\imm, l}\wedge\diff^{\ree, r}\\
			&= \sum_r \sum_l \diff^{\ree, l}\left(\partial_{\ree, r}\right) \ \diff^{\ree, l}\wedge\diff^{\imm, r} + 
			\diff^{\ree, l}\left(\partial_{\imm, r}\right) \diff^{\imm, l}\wedge\diff^{\imm, r}\\
			&\quad -\diff^{\imm, l}\left(\partial_{\ree, r}\right) \diff^{\ree, l}\wedge\diff^{\ree, r} - 
			\diff^{\imm, l}\left(\partial_{\imm, r}\right) \diff^{\imm, l}\wedge\diff^{\ree, r}\\
			&=\sum_r  \diff^{\ree, r}\wedge\diff^{\imm, r} - \diff^{\imm, r}\wedge\diff^{\ree, r}\\
			&=2 \sum_r  \diff^{\ree, r}\wedge\diff^{\imm, r}
		\end{align*}
		which evaluates to
		\begin{align*}
			d\omega^{\rot}(V,W) &= 2 \sum_r \Big( \diff^{\ree, r}\left(V\right) \diff^{\imm, r}\left(W\right) - \diff^{\imm, r}\left(V\right) \diff^{\ree, r}\left(W\right) \Big) \\
			&= 2\, \im{\langle V, W\rangle}
		\end{align*}
		where the last equation follows from a computation analogous to \eqref{oneform}.
		
		\item By choosing w.l.o.g. $e_1 = \neutral$, we obtain
		\begin{align*}
			\omega^{\trl}(V) = \diff^{\ree, 1} + \i\, \diff^{\imm, 1}
		\end{align*}
		which immediately yields $\diff\langle \neutral, \cdot \rangle = 0$, since $d\,\diff^{\ree, 1} = d\,\diff^{\imm, 1} = 0$.	 
	\end{enumerate}
\end{proof}

%
%


\subsection{Tensor-product factorization} \label{EYM}

\newcommand{\tp}[2]{\, #1 \otimes #2\,}

The optimality of the proposed tensor-product factorization follows from the Eckart-Young-Mirsky theorem (EYM) which can be found, e.g., in \citepsup[page 139]{Gentle2007MatrixAlgebra} for matrices and, in more general terms, in \citepsup[page 111]{Hsing2015TheoreticalFoundations} for Hilbert-Schmidt operators. In the following, we present a tensor-product version of EYM designed for our needs. The optimality of the tensor-product factorization is then illustrated in two corollaries -- first in a theoretical model setting and second for the empirical decomposition on evaluations which can be practically conducted on given data.
Consider two real vector spaces $\cB_j$, $j \in \{0,1\}$, with positive semi-definite bilinear forms $\langle\cdot,\cdot\rangle_j: \cB_j \times \cB_j \rightarrow \R$ inducing semi-norms $\|\cdot\|_j$. Assuming $\cB_1$ to be, in fact, a function space of functions $f: \cX \rightarrow \R$ on some set $\cX$, 
the (vector space) tensor product $\tp{\cB_1}{\cB_0}$ of $\cB_0$ and $\cB_1$ is the vector space spanned by all $f \otimes y: \cX \rightarrow \cB_0, \bx \mapsto f(\bx)\, y$ with $f\in\cB_1$ and $y\in\cB_0$.
By linear extension, a symmetric positive semi-definite bilinear form on $\tp{\cB_1}{\cB_0}$ is defined by $\langle \tp{f}{y}, \tp{f'}{y'} \rangle_{\tp{\cB_1}{\cB_0}} = \langle f, f' \rangle_{1} \, \langle y, y' \rangle_{0}$ for all $f,f' \in \cB_1, y,y' \in \cB_0$. It induces a semi-norm $\|\cdot\|_{\tp{\cB_1}{\cB_0}}$ on the tensor product space.
 
\newcommand{\rank}{\operatorname{rank}} 

\begin{theorem}[Eckart-Young-Mirsky for finite-dimensional tensor-products]
	\label{EckartYoungMirsky}
	Let $\cB_0, \cB_1$ be semi-normed vector spaces as defined above and $h = \sum_{r=1}^{m_0}  \sum_{l=1}^{m_1} \theta^{(r,l)} \tp{b_1^{(l)}}{b_0^{(r)}} \in \tp{\cB_1}{\cB_0}$ expressed as a finite linear-combination with $b_1^{(1)}, \dots, b_1^{(m_1)} \in \cB_1$,  $b_0^{(1)}, \dots, b_0^{(m_0)} \in \cB_0$, and coefficient matrix $\{\theta^{(r,l)}\}_{r,l} = \bTheta \in \R^{m_0\times m_1}$. 
	Then we can optimally decompose $h = \sum_{r=1}^m d^{(r)} \tp{\xi_{1}^{(r)}}{\xi_{0}^{(r)}}$ with $m = \min\{m_0, m_1\}$, $d^{(1)} \geq \dots \geq d^{(m)} \geq 0$ and $\langle \xi_{j}^{(r)}, \xi_{j}^{(l)} \rangle_j = \ind(r=l)$ for $\xi_{j}^{(r)}\in\cB_j$
	, in the sense that for any $L\leq m$ 
	\begin{equation}
		\label{optimalapprox}
		\| h - \sum_{r=1}^L d^{(r)} \tp{\xi_{1}^{(r)}}{\xi_{0}^{(r)}} \|_{\tp{\cB_1}{\cB_0}} \ \leq \ \| h - \sum_{r=1}^L d^{(r)}_\star \tp{\xi_{1\star}^{(r)}}{\xi_{0\star}^{(r)}} \|_{\tp{\cB_1}{\cB_0}}
	\end{equation}
	for all $d^{(r)}_\star\in \R$ and $ \xi_{j\star}^{(r)} \in \cB_j$, $j\in\{0,1\}$, $r=1,\dots, L$.
	Arranging $\bD = \operatorname{diag}(d^{(1)}, \dots, d^{(m)})$ and expressing $\xi_{j}^{(r)} = \sum_{l} u_{j}^{(l,r)} b_{j}^{(l)}, r=1,\dots, m,$ with the coefficient matrices $\{u_{j}^{(l,r)}\}_{l,r}=\bU_j \in \R^{m_j \times m}$, $j\in\{0,1\}$, an optimal decomposition is obtained as follows: 
		 
	\begin{enumerate}[label = \roman*)]
		\item If for $j\in\{0,1\}$ the Gram matrices $\bG_j = \{\langle b_{j}^{(r)}, b_{j}^{(l)}\rangle_j\}_{r,l}$ are the identity $\bG_j = \bI_{m_j}$, the matrices $\bD$ and $\bU_j$, $j\in\{0,1\}$, are directly determined via SVD of the coefficient matrix $\bTheta = \bU_0 \bD \bU_1^\top$.
		
		\item In general, there are suitable matrices $\bM_j\in\R^{\rank \bG_j \times m_j}$, $j\in\{0,1\}$, such that $\bXi = \bV_0 \bD \bV_1^\top$ is the SVD of the matrix $\bXi = \bM_0\bTheta\bM_1^{\top}$ and $\bU_j = \bM_j^{-} \bV_j$ with generalized inverse $\bM_j^- = \bM_j^\top (\bM_j \bM_j^\top)^{-1}$. 
		\begin{enumerate}
			\item In general, a suitable matrix is given by $\bM_j = \sqrt{\bG_j}^{\top}$ with $\bG_j = \sqrt{\bG_j}\sqrt{\bG_j}^\top$ a Cholesky decomposition.
			
			\item If the $b_{j}^{(r)}$ can be identified with vectors $\bb_{j}^{(r)}\in \R^{m_j'}$ of some length $m_j'$, arranged as column vectors of a ``design matrix'' $\bB_j \in \R^{m_j'\times m_j}$, such that $\langle b_{j}^{(r)}, b_{j}^{(l)} \rangle_j = (\bb_{j}^{(r)})^\top \bW_j \bb_{j}^{(l)}$, with $r,l = 1,\dots, m_j$, for a symmetric positive definite weight matrix $\bW_j$, we may equivalently set $\bM_j = \bR_j$ based on the QR-decomposition $\sqrt{\bW_j}^\top \bB_j = \bQ_j \bR_j$. In this case, design matrices $\bE_j$ of vector representatives $\bxi_{j}^{(r)}$ for the $\xi_{j}^{(r)}$ can, alternatively, be obtained as $\bE_j = \sqrt{\bW_j}^{-\top}\bQ_j\bV_j$ (where $\bW_j$ is typically diagonal and, hence, $\sqrt{\bW_j}^{-\top}$ fast to compute).
		\end{enumerate}
	\end{enumerate}
\end{theorem}

\begin{proof}
	\begin{enumerate}[label = \textit{\roman*)}]
		\item \label{factorization_orthonormal} For $j\in\{0, 1\}$, denote the column vectors of $\bU_j$ by $\bu_{j}^{(r)}$, $r=1,\dots, m$, and consider the space of $m_0\times m_1$ matrices equipped with the inner product $\langle \bTheta_1, \bTheta_2 \rangle_{F} = \tr{\left(\bTheta_1^\top \bTheta_2\right)}$, for $\bTheta_1, \bTheta_2 \in \R^{m_0\times m_1}$, inducing the Frobenius norm $\|\cdot \|_F$. 
		
		The EYM for matrices \citepsup[e.g.][page 139]{Gentle2007MatrixAlgebra} states that the matrix $\bTheta_L = \sum_{r=1}^{L} d^{(r)} \bu_{0}^{(r)} (\bu_{1}^{(r)})^\top$ is the best rank $L$ approximation of $\bTheta$, in the sense that
		$$
		\| \bTheta - \bTheta_L \|_F \leq \| \bTheta - \sum_{r=1}^L d_\star^{(r)} \bu_{0\star}^{(r)} (\bu_{1\star}^{(r)})^\top\|_F \text{ for any } d_\star^{(r)} \in \R, \bu_{j\star}^{(r)} \in \R^{m_j}, r = 1, \dots, m.
		$$
		
		To apply the theorem, we point out that, provided the Gram matrices $\bG_j=\bI_{m_j}$, the $\{b_{j}^{(r)}\}_{r=1,\dots,m_j}$ and, thus, also $\{\tp{b_1^{(r)}}{b_0^{(l)}}\}_{r,l}$ are orthonormal bases of finite-dimensional subspaces $\cA_j \subset \cB_j$ and $\tp{\cA_1}{\cA_0} \subset \tp{\cB_1}{\cB_0}$, respectively, forming Hilbert spaces. 
		Hence, the basis representation map sending $\xi_j^{(r)} \mapsto \bu_j^{(r)}$ to its coefficient vector w.r.t.\ $\{b_{j}^{(r)}\}_r$ presents an isometric isomorphism from $\cA_j$ to $\R^{m_j}$.
		Accordingly, the basis representation $\tp{\cA_1}{\cA_0} \rightarrow \R^{m_0 \times m_1}$, $h \mapsto \bTheta$ presents an isometric isomorphism identifying $\tp{\xi_1^{(r)}}{\xi_0^{(l)}}$ with $\bu_{0}^{(l)} (\bu_{1}^{(r)})^\top$. The isometry follows from $\langle h_1, h_2 \rangle = \sum_{r, l, r', l'} \theta_1^{(r,l)} \theta_2^{(r',l')} \langle \tp{b_1^{(r)}}{b_0^{(l)}} , \tp{b_1^{(r')}}{b_0^{(l')}} \rangle = \sum_{r, l} \theta_1^{(r,l)} \theta_2^{(r,l)} = \tr{\left(\bTheta_1^\top\bTheta_2\right)}$ for basis representations $h_1 \mapsto \bTheta_1$ and $h_2 \mapsto \bTheta_2$.
		This lets us carry over the EYM for matrices to $\tp{\cA_1}{\cA_0}$ yielding the desired inequality \eqref{optimalapprox} restricted to $\xi_{j\star}^{(r)}\in \cA_j \subset \cB_j$. The property $d_1\geq \dots \geq d_m \geq 0$ and orthonormality of the $\xi_j^{(r)}$ are also inherited from the SVD.
		
		Moreover, we can project any $\xi_{j\star}^{(r)} \in \cB_j$ as $\xi_{j\parallel}^{(r)} = \sum_l \langle b_j^{(l)}, \xi_{j\star}^{(r)} \rangle_j\, b_j^{(l)}$ into $\cA_j$ and define $\xi_{j\perp}^{(r)} = \xi_{j\star}^{(r)} - \xi_{j\parallel}^{(r)}$, which yields an analogous decomposition $h_\star =  \sum_{r=1}^L d^{(r)}_\star \tp{\xi_{1\star}^{(r)}}{\xi_{0\star}^{(r)}} = h_\parallel + h_\perp$ with $h_\parallel\in\tp{\cA_1}{\cA_0}$ and $\langle h, h_\perp \rangle_{\tp{\cB_1}{\cB_0}} = \langle h_\parallel, h_\perp \rangle_{\tp{\cB_1}{\cB_0}} = 0$. Thus, we have $\| h - h_\star \|_{\tp{\cB_1}{\cB_0}}^2 = \| h - h_\parallel \|_{\tp{\cB_1}{\cB_0}}^2 + \|h_\perp\|_{\tp{\cB_1}{\cB_0}}^2 \geq \| h - h_\parallel \|_{\tp{\cB_1}{\cB_0}}^2 \overset{\text{EYM}}{\underset{\text{on} \tp{\cA_1}{\cA_0}}\geq} \| h - \sum_{r=1}^L d^{(r)} \tp{\xi_{1}^{(r)}}{\xi_{0}^{(r)}} \|_{\tp{\cB_1}{\cB_0}}$, which completes the proof.
		
		\item We represent $b_{j}^{(r)} = \sum_{l=1} M_{j}^{(l,r)} a_j^{(l)}$ in an orthonormal basis $\{a_j^{(l)}\}_l$ of the Hilbert space $\cA_j \subset \cB_j$ spanned by $\{b_{j}^{(r)}\}_r$ as in \ref{factorization_orthonormal} with the coefficients forming the matrix $\bM_j = \{M_{j}^{(l,r)}\}_{l,r}$, for $j\in\{0,1\}$, such that $\bXi = \bM_0 \bTheta \bM_1^{\top}$ is the coefficient matrix of $h$ w.r.t.\ $\{a_{j}^{(l)}\}_l$. Hence, due to \ref{factorization_orthonormal}, the matrices $\bXi = \bV_0 \bD \bV_1^\top$ obtained by SVD fulfill the desired properties where the $\bV_j$ are the coefficient matrices of the $\{\xi_{j}^{(r)}\}_r$ w.r.t.\ $\{a_{j}^{(r)}\}_r$. We may set $\bU_j = \bM_j^- \bV_j$ to represent $\{\xi_{j}^{(r)}\}_r$ in the original basis $\{b_{j}^{(r)}\}_r$ instead,  since, due to $\bM_j \bM_j^- = \bI_{\rank \bG_j}$, we have $a_{j}^{(r)} = \sum_{l=1} M_{j}^{-(l,r)} b_j^{(l)}$ for $\bM^- = \{M_{j}^{-(l,r)}\}_{l,r}$.
		\begin{enumerate}[label = \textit{\alph*)}]
			\item \label{factorization_general} Constructing the orthonormal basis $\{a_{j}^{(r)}\}_r$ via $a_{j}^{(r)} = \sum_{l=1} M_{j}^{-(l,r)} b_j^{(l)}$ with $\bM_j^- = \{M_{j}^{-(l,r)}\}_{l,r} = \sqrt{\bG_j}^{\top-}$ is straight forward yielding 
			\begin{align*}
				\{\langle a_{j}^{(r)}, a_{j}^{(l)} \rangle \}_{r,l} &= \bM_j^{-\top} \bG_j \bM_j^{-}\\ 
			 &=\left(\sqrt{\bG_j}^\top\sqrt{\bG_j}\right)^{-1} \sqrt{\bG_j}^\top 
			\sqrt{\bG_j} \sqrt{\bG_j}^{\top}
			\sqrt{\bG_j} \left(\sqrt{\bG_j}^\top\sqrt{\bG_j}\right)^{-1}\\ 
			&= \bI_{\rank \bG_j}.
			\end{align*}
			\item As in this case, $\bG_j = \bB_j^\top \sqrt{\bW}_j \sqrt{\bW}_j^\top \bB_j = \bR_j^\top \bQ_j^\top\bQ_j \bR_j \overset{\bQ_j \text{ orthogonal}}{=} \bR_j^\top \bR_j$ the choice $\bM_j = \bR_j$ is equivalent to \ref{factorization_general}. Accordingly, $\bU_j = \bR_j^{-} \bV_j$ and thus $\bE_j = \bB_j \bU_j = \bB_j \bR_j^{-} \bV_j = \sqrt{\bW}^{-\top}\bQ_j \bV_j$. 
		\end{enumerate}
	\end{enumerate}
\end{proof}


\begin{corollary}[Tensor-product factorization]
	Let $\{b_0^{{(r)}}\}_{r=1,\dots,m_0}$ elements of a Hilbert space $\cY$ with norm $\|\cdot\|$ and $b_{1}^{(l)} \in \cL^2(\cX)=\{f\!:\!\cX \rightarrow \R\, : f\!\circ\!\bX \text{ measurable, } \E\left(\|f(\bX)\|^2\right)<\infty\}$, $l=1,\dots, m_1$, square-integrable functions of a random covariate vector $\bX$ taking values in $\cX$. Let further $h(\bx) = \sum_{r=1}^{m_0}\sum_{l=1}^{m_1} \theta^{(r,l)} \, b_{1}^{(l)}(\bx) \, b_{0}^{(r)}$ for $\bx \in \cX$.
	Then we can optimally decompose $h(\bx) = \sum_{r=1}^m h^{(r)}(\bx)\, \xi^{(r)}$ with $m = \min\{m_0, m_1\}$, $\xi^{(1)}, \dots, \xi^{(m)}$ orthonormal and $h^{(r)}\in\cL^2(\cX)$ with $\E\left(h^{(1)}(\bX)^2\right)\geq\dots\geq \E\left(h^{(m)}(\bX)^2\right)$,  in the sense that for any $L\leq m$ 
	$$
	\E\left( \| h(X) - \sum_{r=1}^L h^{(r)}(X)\, \xi^{(r)} \|^2 \right) \leq \E\left( \| h(X) - \sum_{r=1}^L h^{(r)}_\star(X)\, \xi^{(r)}_\star \|^2 \right),
	$$
	for any other $\xi^{(r)}_\star \in \cY$ and $h^{(r)}_\star \in \cL^2(\cX)$, $r=1,\dots, L$.
	An optimal decomposition is obtained by specifying $\xi^{(r)} = \xi_0^{(r)}$ and $h^{(r)} = d^{(r)}\, \xi_1^{(r)}$ as in Theorem \ref{EckartYoungMirsky} with $\langle \cdot, \cdot \rangle_0=\langle \cdot, \cdot \rangle$ the inner product of $\cY$ and $\langle f, f' \rangle_1 = \E\left( f(X) \, f'(X)\right)$ for $f, f'\in\cL^2(\cX)$.
\end{corollary}

\begin{proof}
	After applying Theorem \ref{EckartYoungMirsky}, it remains to check that $\|h\|_{\tp{\cL^2(\cX)}{\cY}}^2 = \E\left(\|h(X)\|^2\right)$. Indeed, this holds for all simple $h = \tp{f}{y}$, since
	$$
		\langle \tp{y}{f}, \tp{y'}{f'} \rangle_{\tp{\cL^2(\cX)}{\cY}} = \langle y, y' \rangle \, \E\left( f(X) \, f'(X)\right) = \E\left( \langle f(X)\, y, f'(X)\, y' \rangle \right)
	$$
	for any $y,y'\in\cY$ and $f,f'\in\cL^2(\cX)$, and, therefore, carries over to all $h \in \tp{\cL^2(\cX)}{\cY}$ in the vector space.
\end{proof}


\newcommand{\di}{\ddot{\imath}}

\begin{corollary}[Tensor-product factorization, empirical version]
\label{corollary:TPfactorization_empirical}
	Let $\cF(\cX, \R)$ and $\cF(\cT, \C)$ denote the sets of functions $\cX \rightarrow\R$ and $\cT \rightarrow \C$, respectively, which are both considered real vector spaces.
	Let $b_{0}^{(r)}\in\cF(\cT, \C)$, $r=1,\dots, m_0$, and $b_{1}^{(l)}\in\cF(\cX, \R)$, $l=1,\dots, m_1$. 
	Consider $h(\bx)(t) = \sum_{r=1}^{m_0}\sum_{l=1}^{m_1} \theta^{(r,l)} \, b_{1}^{(l)}(\bx) \, b_{0}^{(r)}(t)$ for $\bx \in \cX$, $t \in \cT$ evaluated, for $i=1,\dots,n$, at $\bx_i \in \cX$ and $t_{i,\iota} \in \cT, \iota= 1,\dots,k_i$. 
	Then we can decompose $h(\bx) = \sum_{r=1}^m h^{(r)}(\bx)\, \xi^{(r)}$ with $m = \min\{m_0, m_1\}$ optimally, in the sense that for any $L\leq m$ and any other functions $\xi^{(r)}_\star: \cT \rightarrow \C$ and $h^{(r)}_\star: \cX \rightarrow \R$, $r=1,\dots, L$, 
	\begin{align}
	\label{tensorfactorization_empirical}
	    \sum_{i=1}^n w_{1i} \frac{1}{n} \sum_{\di=1}^n \sum_{\iota=1}^{k_i} w_{0\di\iota} | h(\bx_i)&(t_{\di\iota}) - \sum_{r=1}^L h^{(r)}(\bx_i)\, \xi^{(r)}(t_{\di\iota}) |^2\nonumber\\ 
	    &\leq \\
	    \sum_{i=1}^n w_{1i} \frac{1}{n} \sum_{\di=1}^n \sum_{\iota=1}^{k_i} w_{0\di\iota} | h(\bx_i)&(t_{\di\iota}) - \sum_{r=1}^L h_\star^{(r)}(\bx_i)\, \xi_\star^{(r)}(t_{\di\iota}) |^2,\nonumber    
	\end{align}
	with integration/sample weights $w_{0i\iota} \geq 0$ and $w_{1i} \geq 0$. An optimal decomposition is obtained by specifying $\xi^{(r)} = \xi_{0}^{(r)}$ and $h^{(r)} = d^{(r)}\, \xi_{1}^{(r)}$, $r=1,\dots, m$, specified as in Theorem \ref{EckartYoungMirsky} with $\langle y, y' \rangle_0 = \frac{1}{n}\sum_{\ddot{\imath}=1}^n \sum_{\iota=1}^{k_i} w_{0\di\iota} \re{y^\dagger(t_{\di\iota}) y'(t_{\di\iota})}$ for $y, y'\in\cF(\cT, \C)$ and $\langle f, f' \rangle_1 = \sum_{i=1}^n  w_{1i} f(x_i) f'(x_i)$ for $f, f'\in\cF(\cX, \R)$.
\end{corollary}

\begin{proof}
	Again, we confirm $\|h\|_{\tp{\cF(\cX, \R)}{\cF(\cT, \C)}}^2 = \sum_{i=1}^n w_{1i} \sum_{\iota=1}^{k_i} w_{0\di\iota} \left( h(\bx_i)(t_{\di\iota}) \right)^2$ by showing
	\begin{align*}
		\langle \tp{y}{f}, \tp{y'}{f'} \rangle_{\tp{\cF(\cX, \R)}{\cF(\cT, \C)}} 
			&= \langle y, y' \rangle_0 \, \sum_{i=1}^n  w_{1i} f(x_i) f'(x_i) = \sum_{i=1}^n  w_{1i} \langle f(x_i) y, f'(x_i) y'\rangle_1\\ 
			&= \sum_{i=1}^n  w_{1i} \frac{1}{n}\sum_{\di=1}^n\sum_{\iota=1}^{k_i} w_{0\di\iota} \re{\left( f(x_i) y(t_{\di\iota}) \right)^\dagger f'(x_i) y'(t_{\di\iota})}
	\end{align*}
	for any $y,y'\in\cF(\cT, \C)$ and $f,f'\in\cF(\cX, \R)$.
\end{proof}

\begin{remark}
For the regular case with $k_1 = \dots = k_n =: k$ and for all $\iota = 1,\dots,k$ also $t_{i\iota} = t_{1\iota} =: t_{\iota}$ and $w_{0i\iota} = w_{01\iota} =: w_{0\iota}$ equal for all observations $i = 1,\dots, n$, Inequality \eqref{tensorfactorization_empirical} simplifies to 
\begin{align*}
	    \sum_{i=1}^n w_{1i} \sum_{\iota=1}^{k_i} w_{0\iota} | h(\bx_i)&(t_{\iota}) - \sum_{r=1}^L h^{(r)}(\bx_i)\, \xi^{(r)}(t_{\iota}) |^2\nonumber\\ 
	    &\leq \\
	    \sum_{i=1}^n w_{1i} \sum_{\iota=1}^{k_i} w_{0\iota} | h(\bx_i)&(t_{\iota}) - \sum_{r=1}^L h_\star^{(r)}(\bx_i)\ \xi_\star^{(r)}(t_{\iota}) |^2.\nonumber    
	\end{align*}
\end{remark}


\subsection{Shape differences in astragali of wild and domesticated sheep}
\label{sec:bones_appendix}

\begin{table}[H]
	
	\caption{\label{tab:data-summary}Distribution of covariate levels over the sheep \new{population}\old{breed}s in the data set.}
	\centering
	\begin{tabular}[t]{l|r|r|r|r|r|r|r}
		\hline
		\multicolumn{1}{c|}{ } & \multicolumn{3}{c|}{Sex} & \multicolumn{4}{c}{Age\_group} \\
		\cline{2-4} \cline{5-8}
		& female & male & na & juvenile & subadult & adult & na\\
		\hline
		Karakul & 21 & 19 & 1 & 1 & 5 & 35 & 0\\
		\hline
		Marsch & 18 & 5 & 0 & 5 & 5 & 13 & 0\\
		\hline
		Soay & 21 & 25 & 12 & 7 & 8 & 13 & 30\\
		\hline
		Wild\_sheep & 21 & 20 & 0 & 5 & 18 & 14 & 4\\
		\hline
	\end{tabular}
\end{table}


\begin{table}[H]
	\centering
	\begin{tabular}[t]{l|r|r|r|r|r|r}
		\hline
		\multicolumn{1}{c|}{ } & \multicolumn{3}{c|}{Mobility} & \multicolumn{3}{c}{Status} \\
		\cline{2-4} \cline{5-7}
		& confined & pastured & free & domestic & feral & wild\\
		\hline
		Karakul & 31 & 10 & 0 & 41 & 0 & 0\\
		\hline
		Marsch & 23 & 0 & 0 & 23 & 0 & 0\\
		\hline
		Soay & 0 & 0 & 58 & 0 & 58 & 0\\
		\hline
		Wild\_sheep & 0 & 0 & 41 & 0 & 0 & 41\\
		\hline
	\end{tabular}
\end{table}

\begin{figure}[H]
	\centering
	\includegraphics[width=0.9\linewidth]{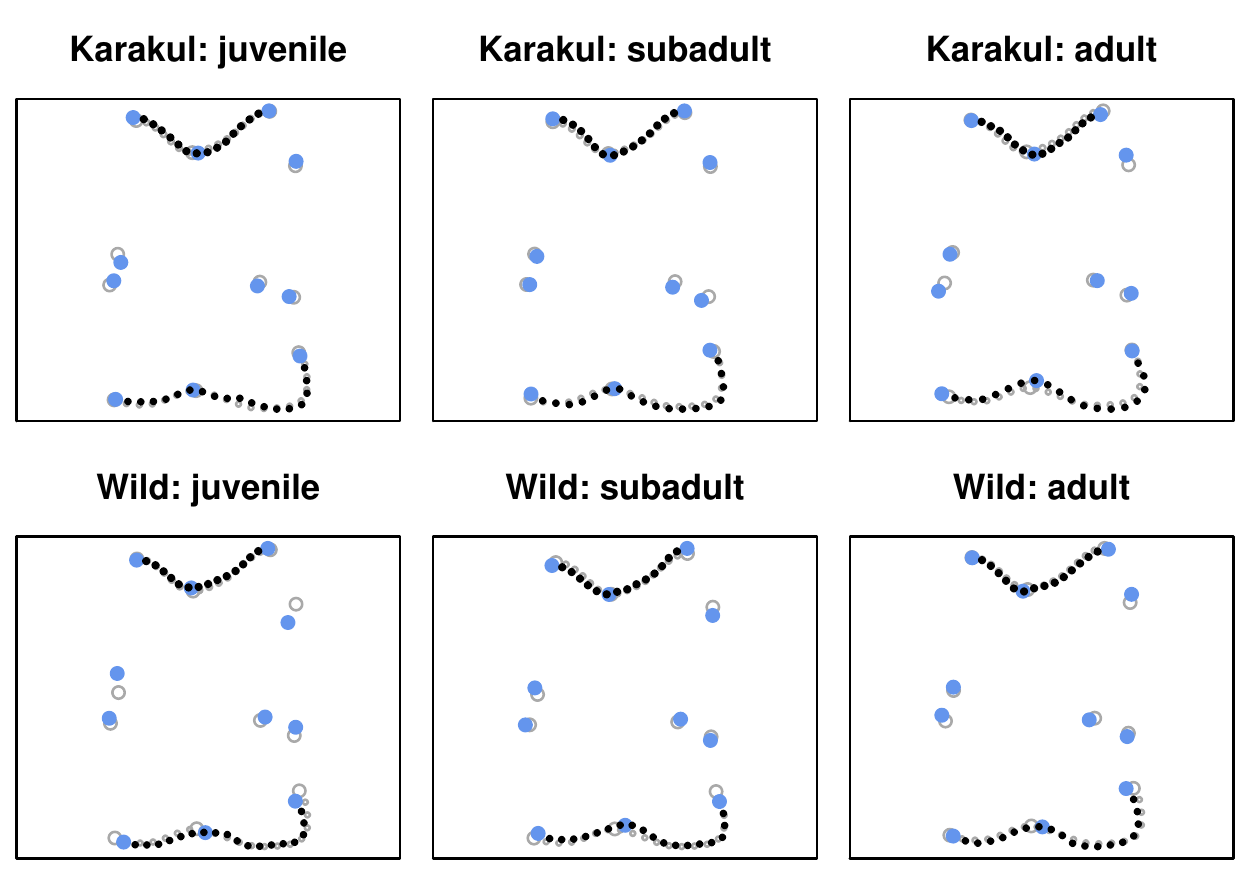}
	\caption[Sheep bone examples]{Six example sheep astragalus shape configurations consisting of landmarks \textit{(blue dots)} and semi-landmarks describing two outline curves \textit{(black dots)} recorded in male Karakul and wild sheep of different age. Points are weighted such that the total weight of each curve corresponds to three landmarks \textit{(weights reflected in point-size)}. Shapes are depicted aligned to their overall mean shape \textit{(grey circles)}.}
	\label{fig:sheepdataexamples}
\end{figure}


\subsection{Cellular Potts model parameter effects on cell form}

In the graphics below, the CPM parameters are abbreviated as
\begin{itemize}
	\item[\textsf{b}:] bulk stiffness $x_{i1} \in [0.003, 0.015]$
	\item[\textsf{m}:] membrane stiffness $x_{i2} \in [0.001, 0.015]$
	\item[\textsf{a}:] substrate adhesion $x_{i3} \in [30, 70]$
	\item[\textsf{r}:] signaling radius $x_{i4} \in [5, 40]$ 
\end{itemize}

\begin{figure}[H]
	\centering
	\includegraphics[width=0.9\linewidth]{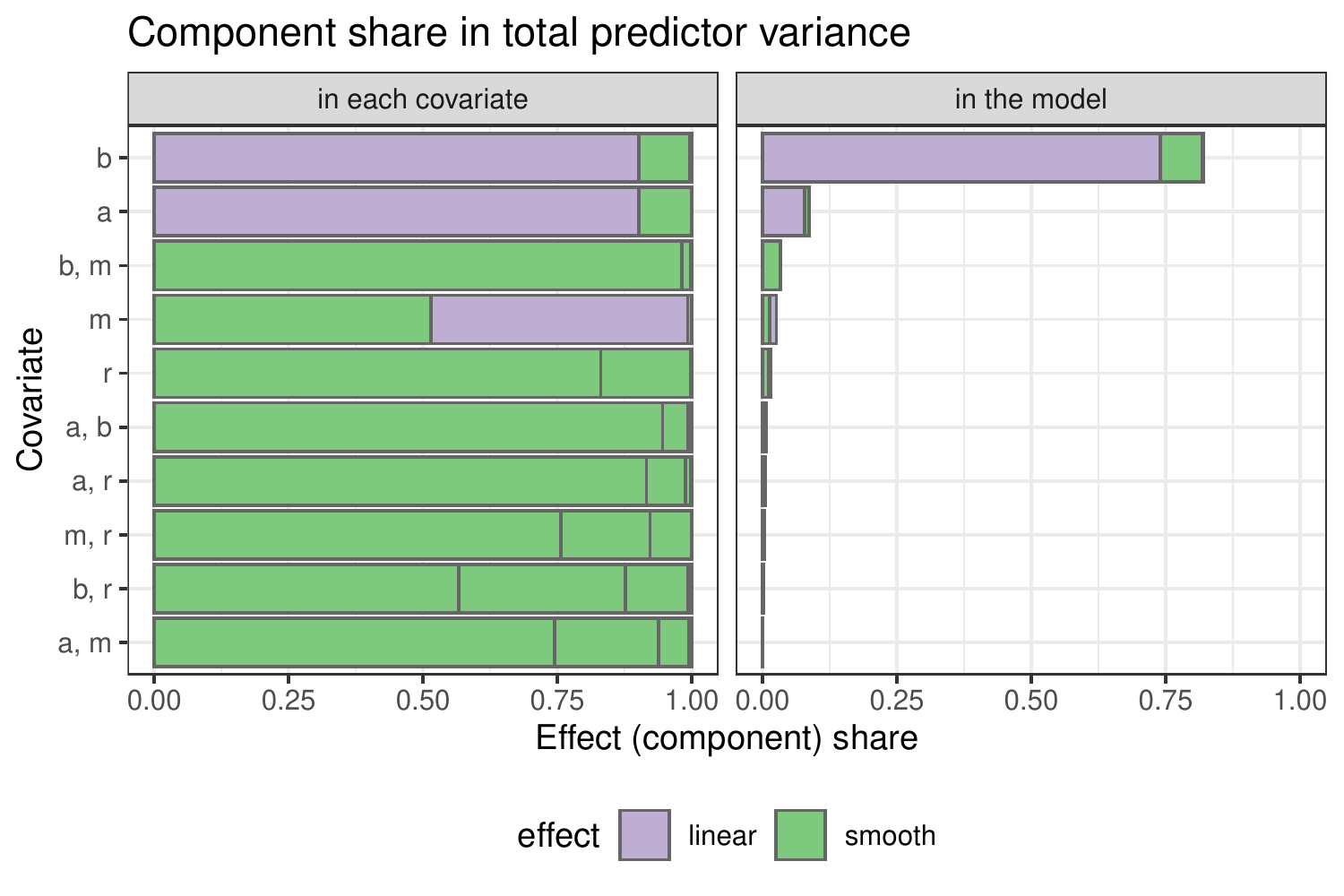}
	\caption[cells: factoriced effect varimp]{
	Tensor-product effect factorization: Predictor variance share explained
by each effect direction \textit{(separated by vertical lines)} relative to the total predictor variance of the effects of
each covariate \textit{(left)} and of the overall model \textit{(right)}. 
Linear effect components are presented together with the respective nonlinear effects of a covariate -- they point, however, in individual directions. Interaction effects are listed separately. We observe that for many covariates the nonlinear
effect is already almost entirely captured by its first component.}
	\label{fig:cellcovfactorizedvarimp}
\end{figure}

\begin{figure}[H]
	\centering
	\includegraphics[width=0.9\linewidth]{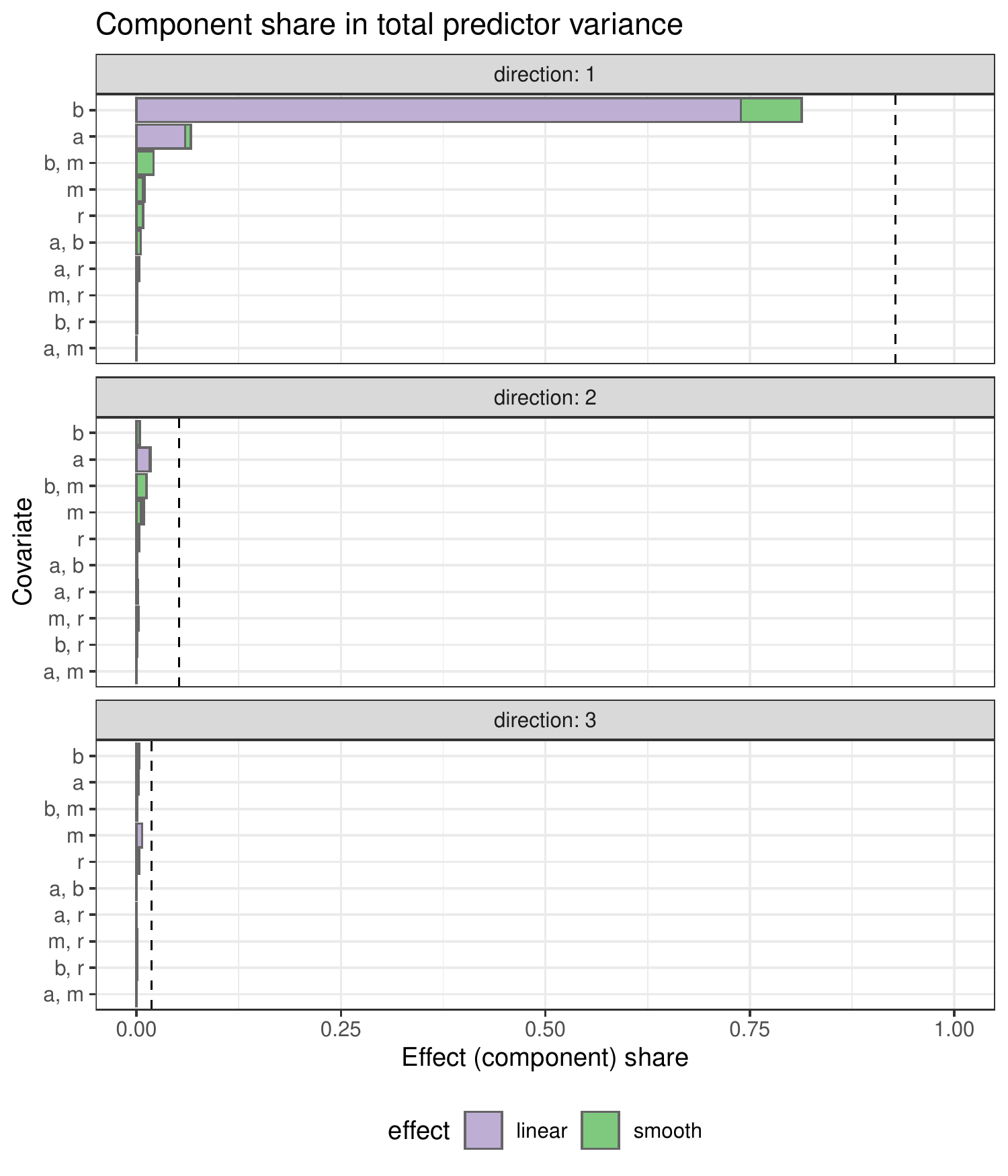}
	\caption[cells: factorized model varimp]{
	Tensor-product model factorization: Predictor variance shares into the first three directions \textit{(dashed vertical lines)} resulting from joint model factorization (unlike individual factorization of effects in Figure \ref{fig:cellcovfactorizedvarimp}).
	Horizontal bars reflect the variance of the single covariate effects within each model predictor component. They roughly -- but due to potential correlation not precisely -- add up to the predictor component variance shares. 
	} 
	\label{fig:cellmodelfactorizedvarimp}
\end{figure}

\begin{figure}[H]
	\centering
	\includegraphics[width=0.6\linewidth]{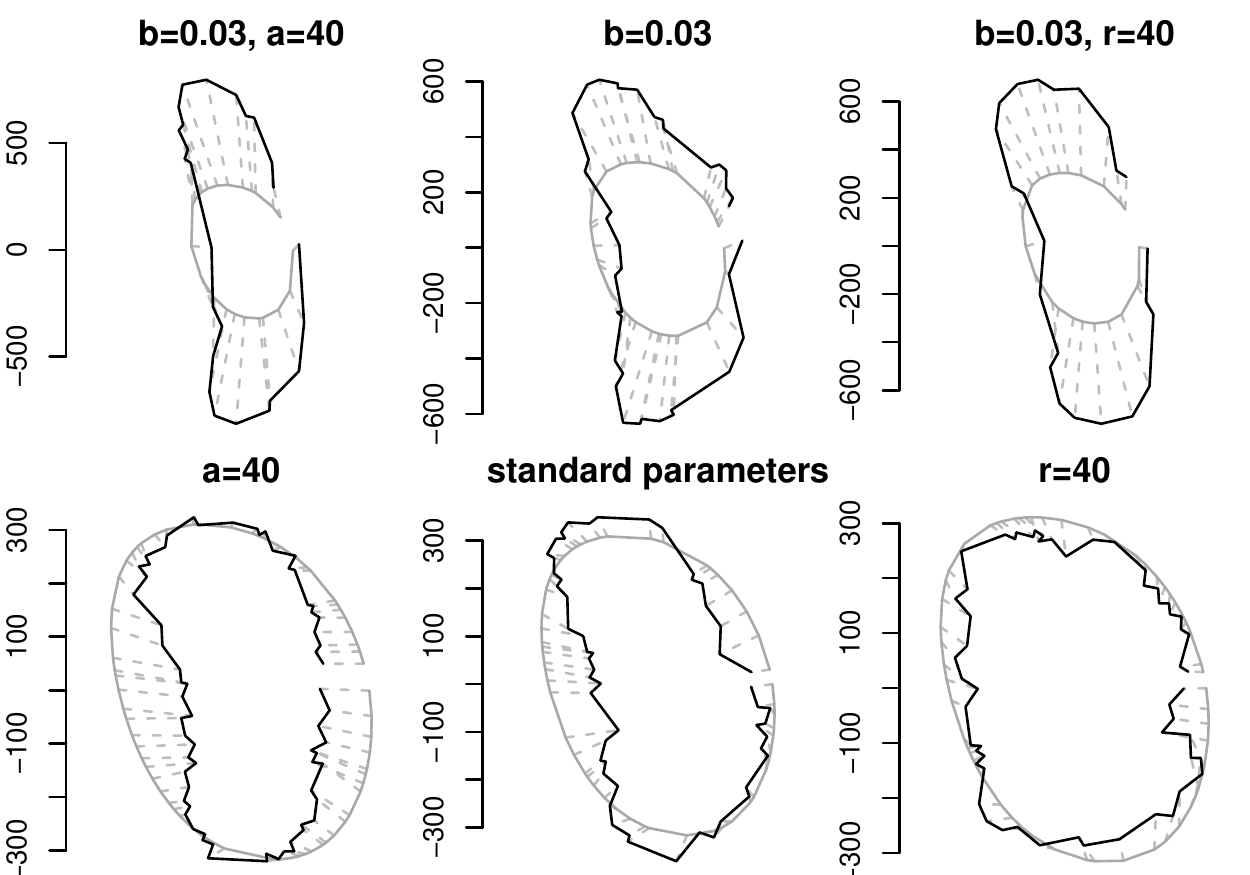}
	
	\vspace{-10pt}
	\rule{0.6\linewidth}{1pt}
	\vspace{10pt}
	
	\includegraphics[width=0.6\linewidth, page=2]{figure/cell_dataexamples}
	\caption[Cell data examples]{
		\textit{Top:}  Example cell outline \textit{(black)}, one randomly selected out of 33 for each of six different CPM parameter (covariate) configurations chosen for visualization, aligned to the overall mean form \textit{(grey)}. Note that while panel scales are individually adjusted for better visibility, contrasting plotted forms with the overall mean, which is equal in all plots, also allows to compare their sizes across panels. Headers show parameter deviations from a standard configuration with $b = 0.009$, $m = 0.003$, $a = 50$ and $r = 20$. Dashed lines indicate point correspondences. Cell outlines are oriented as cells migrating rightwards and not connected between $y(0)$ and the point left of it (while outlines are modeled as closed forms in the model).
		\textit{Bottom:} Predictions for the corresponding mean form  of our cell form model described in Section \ref{sec_cells}.}
	\label{fig:celldataexamples}
\end{figure}


\subsection{Realistic shape and form simulation studies}
\label{sec:bottles_appendix}

\subsubsection{Sampling of response observations}

Response curves are generated separately for the shape and form scenario as follows: we obtain the true underlying models by fitting original beer and whisky bottles and 3D rotated versions of them, four successively rotated towards the viewer and four away from the viewer, and compute transported residuals $\epsilon_i$ of a total of $N=360$ bottle outlines $y_1, \dots, y_N$ (20 whisky and 20 beer brands, each from 9 different angles $z_1$). 
For each simulated dataset, a sample of the desired size $n$ is randomly drawn (with replacement) from the model residuals $\epsilon_1, \dots, \epsilon_N$. To obtain irregular data with an average grid length $k=\frac{1}{n}\sum_{i=1}^n k_i$, we subsample the original evaluations $\epsilon_i(t_{i1}), \dots, \epsilon_i(t_{iK_i})$, with original grid sizes $K_i \geq 123$, in two steps: first we randomly pick three evaluations as minimal sample size; then we draw evaluations independently with $\frac{k-3}{K_i-3}$ probability to enter the dataset. 
To preserve the original covariate distribution of the data, covariates are not randomly picked but we select batches of 9 beer and 9 whisky bottles with $z_1 \in [-60, 60]$ as in the original dataset. Sample sizes $n$ are, therefore, multiples of 18. 
With the conditional means $[\mu_i]$ determined by the covariates, the evaluated residuals $\epsilon_i$ (on $k_i$ points) are parallel transported to $\varepsilon_{[\mu_i], i}\in T_{[\mu_i]}\cY^*_{i/G}$, into the tangent space of the true conditional mean, to generate the simulated shape/form dataset $[y_i] = \Exp_{[\mu_i]}(\varepsilon_{[\mu_i], i})$, $i=1, \dots, n$. 

\subsubsection{Simulation results}

In order to systematically and efficiently assess model behavior, we vary key aspects of the model setup and compare fitting performance in selected settings. 
Here, we list the different aspects and how they are referred to in subsequent graphical visualizations:
\begin{itemize}
	\item \textsf{Scenario}: Shape or form responses.
	\item \textsf{Sample size} $n$ of curves and mean \textsf{grid size} $k$ that curves are evaluated on. 
	\item \textsf{Setting}: simulations adjusted in an additional aspect compared to a \textsf{default} setup
	\begin{itemize}[leftmargin=1.5in]
		\item[\textsf{equal weight}:] Constant inner product weights $w_{i\iota} = \frac{1}{k_i}$, $\iota = 1,\dots, k_i$, are utilized for curve evaluations $y_i(t_{i1}), \dots, y_i(t_{ik_i})$ instead of trapezoidal rule weights (\textsf{default}).
		\item[\textsf{no nuisance}:] No constant and smooth nuisance effects $h_0$ and $f_2(z_2)$ are included into the model, which are included by \textsf{default}. 
		\item[\texttt{pre-aligned}:] This setting concerns the pre-alignment of the curves $y_1, \dots, y_n$ representing the forms/shapes in the simulated data. 
		Note, however, that due to alignment to the pole $p$ in the very beginning of the Riemannian $L^2$-Boosting algorithm, all of this only effects the preliminary pole $p_0$ used for estimation of $p$. 
		In the models fit in the paper, we estimated $p_0$ by using a  functional $L^2$-Boosting algorithm (without any alignment), 
		which makes sense for typical data where the curves occur roughly aligned. Consequently, this aspect translates to a ``good or worse starting point $p_0$'', which is then replaced by $p$ in the actual model fit.
		In \textsf{pre-aligned} settings, simulated response curves $\tilde{y}_i = \Exp_{\mu_i}(\varepsilon_{\mu_i, i})$ are directly used for fitting. 
		In the \textsf{default}, by contrast, the model is fit on random representatives of $[y_i]$ to mimic realistic scenarios, where $y_i = \lambda u \tilde{y}_i + \gamma \in [\tilde{y}_i]$ with 
		$u = \exp(\i \omega)$, $\omega \sim N(0, \frac{\pi}{20})$, 
		with $\gamma = \sigma_1\gamma_1 + \sigma_2^2\gamma_2 \i$, $\gamma_1, \gamma_2 \sim N(0, 1)$, $(\sigma_1^2, \sigma_2^2) = \frac{1}{n k} \sum_{i=1}^n \sum_{\iota=1}^{k_i} (\re{\tilde{y}_i(t_\iota)}, \im{\tilde{y}_i(t_\iota)})$, 
		and with $\lambda = 1$ for forms and $\lambda \sim \operatorname{Gamma}(10^2, 10^{-2})$ for shapes. 
		
	\end{itemize}
\end{itemize}

\begin{figure}[H]
	\centering
	\includegraphics[width=.8\linewidth]{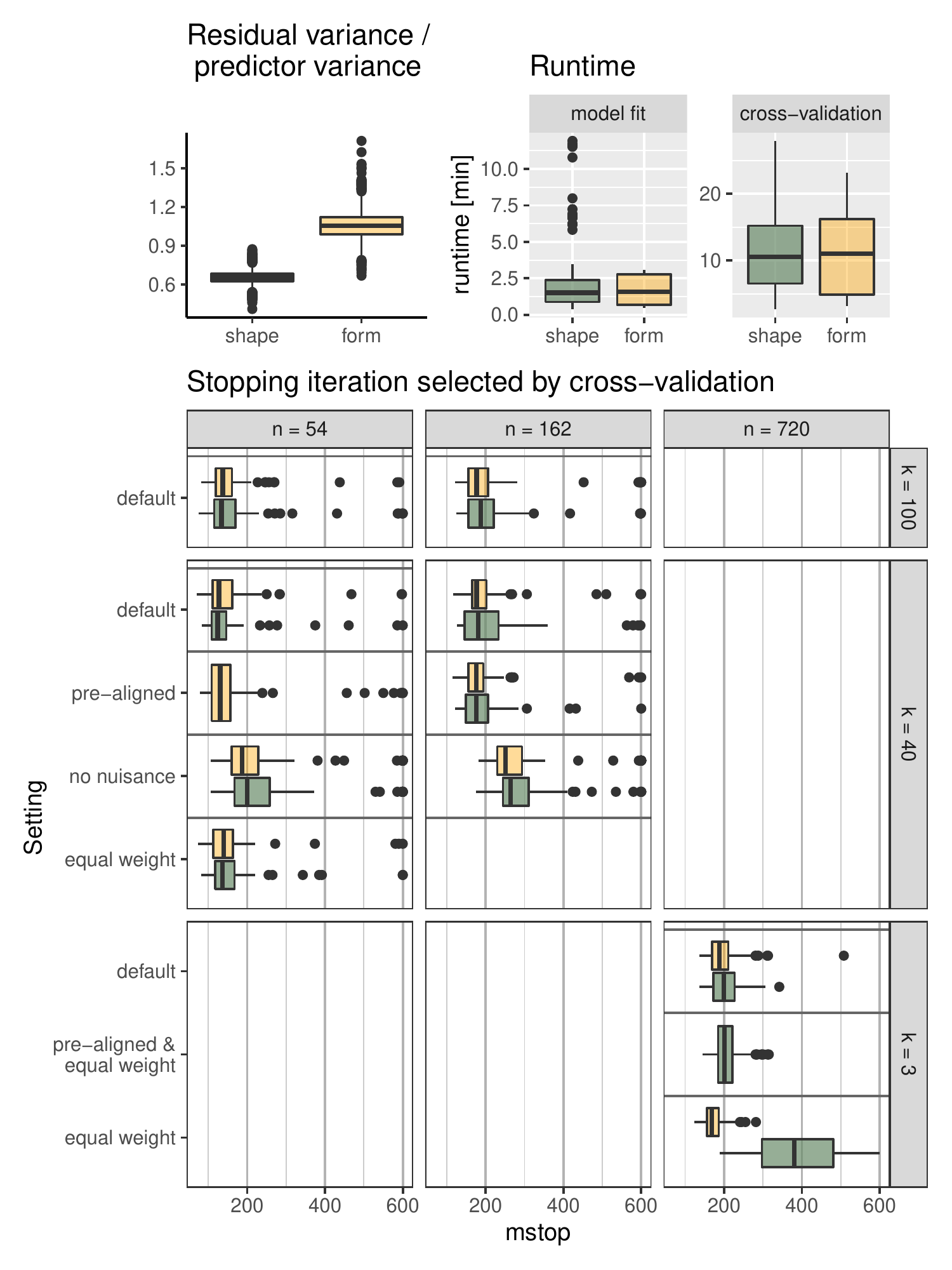}
	\caption[Simulations: Meta-Info]{\textit{Top, left:} Noise-to-signal ratio: distribution of empirical residual variance / predictor variance ratio in all simulations. \textit{Top, right:} Runtime distribution of model fits and subsequent cross-validations (always running 600 boosting iterations). \textit{Bottom:} Distribution of stopping iteration $m_{stop}$ selected by 10-fold curve-wise cross-validation for different simulation settings. All plots displayed separately for the shape and form scenario (top and bottom row within sub-panel, respectively).}
	\label{fig:simplotmetainfo}
\end{figure}

\begin{figure}[H]
	\centering
	\includegraphics[width=6in, page=3]{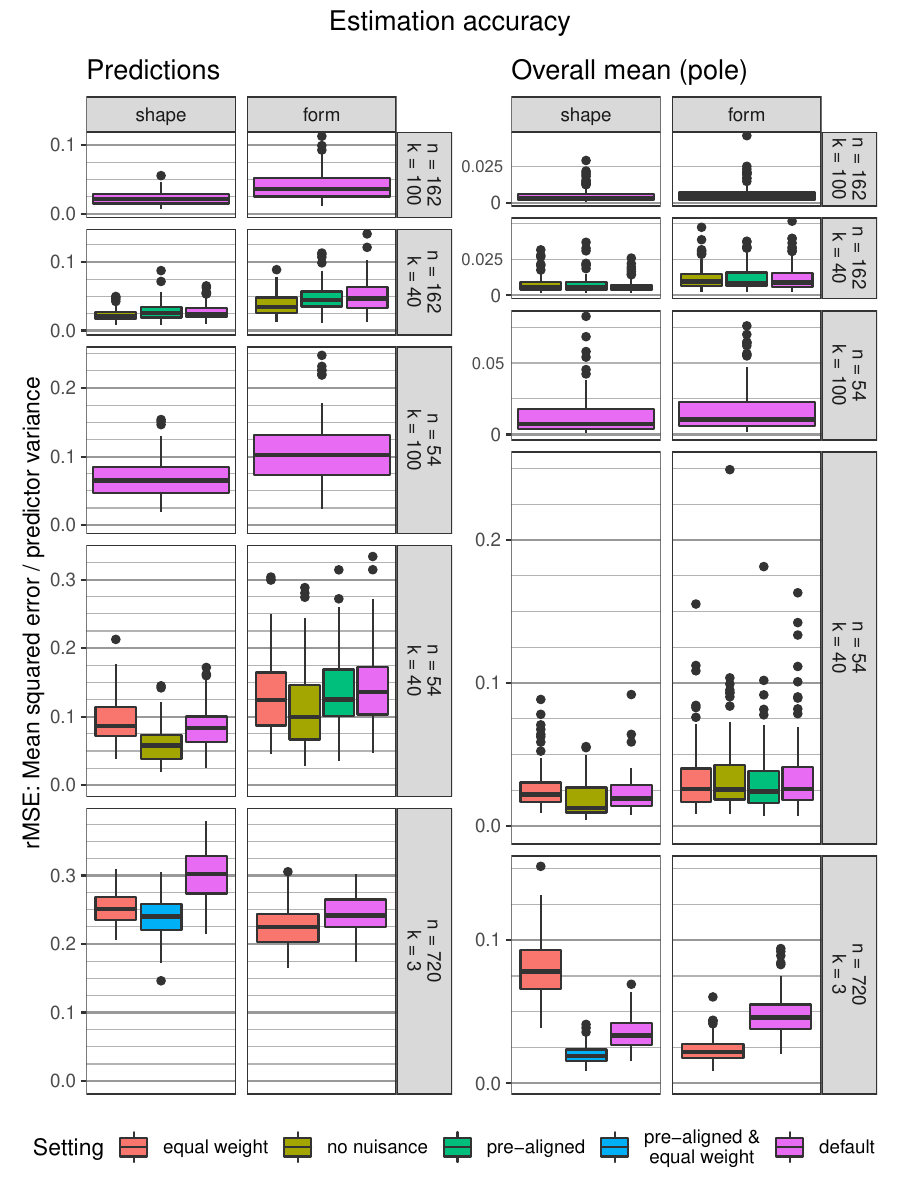}
	\caption[Simulations: Prediction Accuracy]{Accuracy in estimating the unconditional mean (pole) and conditional means (predictions), where the MSE is averaged over the covariate values in the dataset.}
	\label{fig:simplotprediction}
\end{figure}

\begin{figure}[H]
	\centering
	\includegraphics[width=6in]{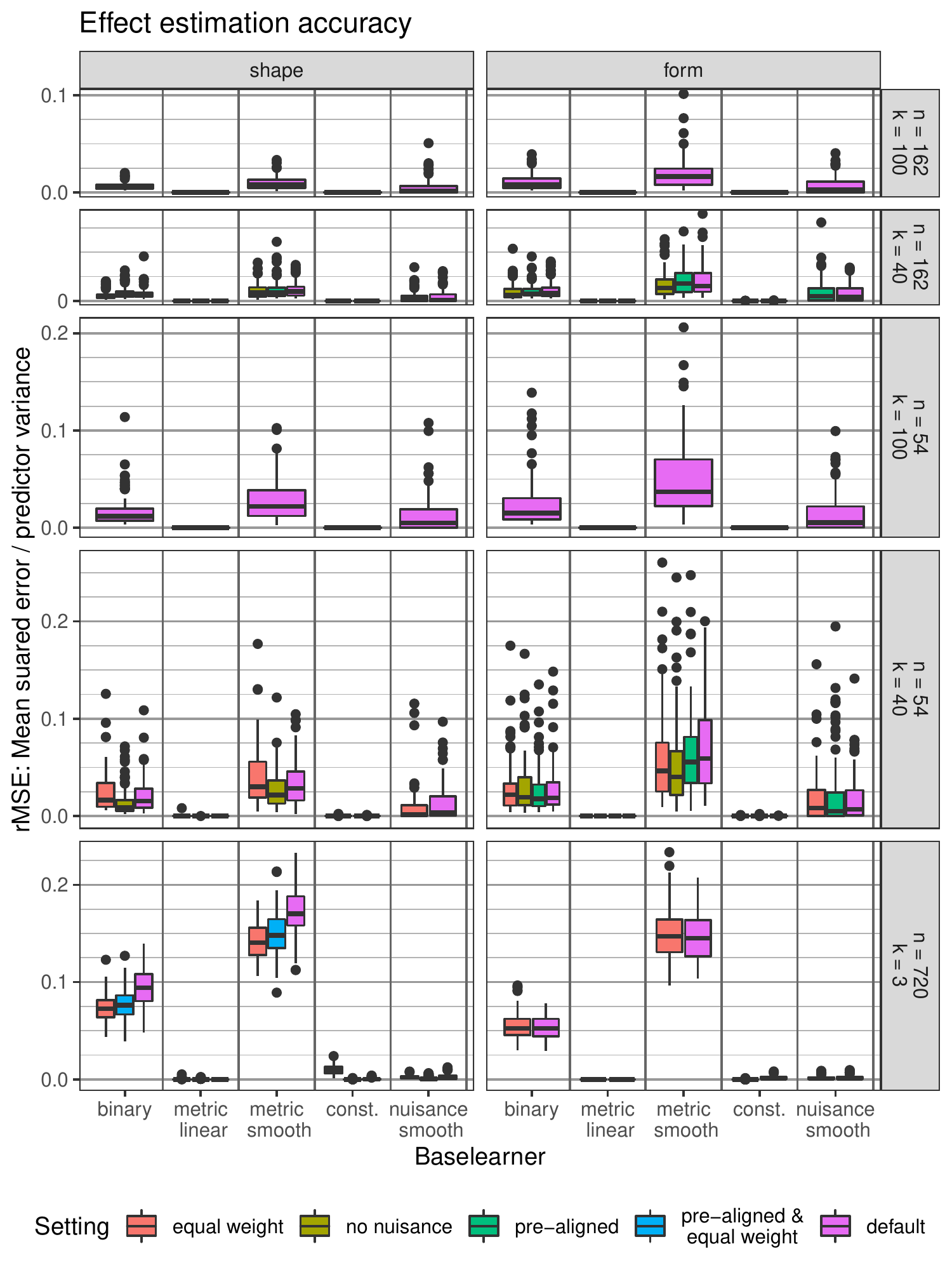}
	\caption[Simulations: Effect estimation accuracy]{Accuracy of estimated effects on tangent space level.}
	\label{fig:simploteffects}
\end{figure}

\newpage
\subsection{Coefficient level modeling}
\label{sec:coefficientlevel}

\newcommand{\ycoef}{\check{y}}
\newcommand{\bycoef}{\check{\by}} 

In the main manuscript, we consider the space of complex valued functions $\cY$ mostly a vector space over $\R$ and utilize real coefficients to formulate the tensor-product effect structure in Section \ref{sec_tensorproduct_effects} in corresponding bases.
In particular for form tangent spaces, identified with real linear subspaces that do not correspond to complex subspaces, this is useful to implement respective constraints via basis transforms. 
By contrast, here we represent $h_\effid(\bx) = \sum_{r, l} \vartheta_{\effid}^{(r,l)} b_\effid^{(l)}(\bx) b_{\nb}^{(r)}$ with complex coefficients $\vartheta_{\effid}^{(r,l)} \in \C$, $r=1,\dots,\Baseid_\nb$, $l=1,\dots, \Baseid_\effid$, with (possibly all real-valued) basis functions $b_\nb^{(1)}, \dots, b_\nb^{(m_\nb)}\in \cY$  corresponding to the basis used for construction of the tangent space basis $\{\partial_r\}_r$ in Section \ref{sec_tensorproduct_effects}. This representation lets us illustrate the link between evaluation level and coefficient level modeling of shapes and forms:\newline
Consider the case where $y_i \in \cY$, $i=1,\dots, n$, can be expanded as $y_i = \sum_{r=1}^{m_\nb} \ycoef_i^{(r)} b_\nb^{(r)}$ in the same basis with 
complex coefficient vectors $\bycoef_i = (\ycoef_i^{(1)}, \dots, \ycoef_i^{(m_\nb)})^\top \in \C^{m_\nb}$, and let also  
the pole $[p] = [\sum_{\baseid=1}^{m_\nb} \check{p}^{(\baseid)} b_\nb^{(r)}]$, $\check{\bp}_i = (\check{p}^{(1)}, \dots, \check{p}^{(m_\nb)})^\top$, be expanded accordingly.  
With $\check{\neutral} = (\check{\neutral}^{(1)}, \dots, \check{\neutral}^{(\Baseid_\nb)})$ the coefficient vector of $\neutral = \sum_{\baseid=1}^{m_\nb} \check{\neutral}^{(r)} b_\nb^{(r)}$ (for B-splines simply $\check{\neutral} = \frac{1}{|\cT|}(1,\dots, 1)^\top$), we have $u y_i + \gamma\,\neutral = \sum_{\baseid=1}^{m_\nb} (u\, \ycoef_i^{(r)} + \gamma\,\check{\neutral}\,) b_\nb^{(r)}$ for $u,\gamma \in \C$, such that basis representation yields an isomorphism between shapes/forms $[y]$ of curves and the shapes/forms $[\bycoef]$ of their coefficients as alternative ``landmarks''. Moreover, when choosing inner products on $\cY$ and $\C^{m_\nb}$ such that $y \rightarrow \bycoef$ is isometric, it follows that $[y] \rightarrow [\bycoef]$ presents an isometric isomorphism. 

Under these assumptions, modeling the mean shape/form $[\check{\bmu}]= \Exp_{[\check{\bp}]}\left(\check{\bh}(\bx)\right)$ of the coefficients $\bycoef_i$, with predictor $\check{\bh}(\bx)=\sum_{\effid = 1}^\Effid \check{\bh}_\effid(\bx) \in \C^{m_0}$ and $\check{\bmu}_i = (\check{\mu}_i^{(1)}, \dots, \check{\mu}_i^{(m_\nb)})^\top \in \C^{m_0}$, is equivalent to our presented model on the original level of curves, if coefficient level effects $\check{\bh}_\effid(\bx) = \sum_{r,l} \vartheta_\effid^{(r,l)} b_\effid^{(l)}(\bx) \be_r$ are specified with the canonical basis $\be_r = \left(\ind\left(r = 1\right), \dots, \ind\left(r = m_\nb\right)\right)^\top$, since  
$[\mu] = [\sum_{\baseid=1}^{m_\nb} \check{\mu}^{(\baseid)} b_\nb^{(r)}] \overset{(*)}= \Exp_{[\sum_{\baseid=1}^{m_\nb} \check{p}^{(\baseid)} b_\nb^{(r)}]}\left(\sum_{\baseid=1}^{m_\nb} \check{h}^{(r)}(\bx) b_\nb^{(r)} \right) = \Exp_{[p]}\left(h(\bx)\right)$ with $\check{\bh}(\bx) = (\check{h}^{(1)}(\bx), \dots, \check{h}^{(m_\nb)}(\bx))^\top$. For shapes, equality $(*)$ follows from 
	\begin{align*}
		\Exp_{\sum_{\baseid=1}^{m_\nb} \check{p}^{(\baseid)} b_\nb^{(r)}}\left(\sum_{\baseid=1}^{m_\nb} \check{h}^{(r)}(\bx) b_\nb^{(r)} \right) &= 
	 \cos(\|\check{\bh}(\bx)\|) \sum_{\baseid=1}^{m_\nb} \check{p}^{(r)} b_\nb^{(r)} + \sin(\|\check{\bh}(\bx)\|)  \frac{\sum_{\baseid=1}^{m_\nb}\check{h}^{(r)}(\bx)b_\nb^{(r)}}{\|\check{\bh}(\bx)\|}  \\ 
		&= 
	\sum_{\baseid=1}^{m_\nb} \left(\cos(\|\check{\bh}(\bx)\|) \check{p}^{(r)} + \sin(\|\check{\bh}(\bx)\|) \frac{\check{h}^{(r)}(\bx)}{\|\check{\bh}(\bx)\|} \right) b_\nb^{(r)}\\ 
		&= \sum_{\baseid=1}^{m_\nb} \be_r^\top \Exp_{\check{\bp}}\left(\check{\bh}(\bx) \right) b_\nb^{(r)} = \sum_{\baseid=1}^{\Baseid_\nb}\check{\mu}^{(\baseid)} b_0^{(\baseid)}
	\end{align*}
where $\Exp$ is the exponential map on the sphere (first on the function space and then on the coefficient level), we use that due to the isometry $\|\check{\bh}(\bx)\| = \|\sum_{\baseid=1}^{m_\nb} \check{h}^{(r)}(\bx) b_\nb^{(r)}\|$ and, we assume w.l.o.g. $\|p\| = \|\check{\bp}\| = 1$ and $\langle p, \neutral\rangle = \sum_r \check{p}^{(r)} = 0$. Accordingly for forms.

However, the expansion $y_i \approx \sum_{r=1}^{m_\nb} \ycoef_i^{(r)} b_\nb^{(r)}$ is typically only approximate. In terms of the inner product, $\langle y_i , y_i' \rangle_i^\nb = \bycoef_i^{\dagger} \check{\bW} \bycoef_i'$, with $\check{\bW}$ the Gramian matrix of $\{b_\nb^{(r)}\}_r$, presents an alternative empirical substitute for the inner product $\langle y_i, y_i' \rangle$ of curves $y_i, y_i'\in\cY$, which is computed on the coefficients instead of $\langle y_i , y_i' \rangle_i = \conj{\by}_i \bW_i \by'_i$ computed on evaluation vectors $\by_i = (y_i(t_{i1}), \dots, y_i(t_{i k_i}))^\top,\by'_i = (y_i'(t_{i1}), \dots, y_i'(t_{i k_i}))^\top$ as suggested in Section \ref{chap_diffgeo}. 
When, for dense grids, it can be assumed that both $\langle y_i , y_i' \rangle_i^\nb \approx \langle y_i , y_i' \rangle_i \approx \langle y_i, y_i' \rangle$ approximate the inner product on the level of curves well, the approach based on the coefficients $\bycoef_i$ may be computationally preferable, guaranteeing regular and typically more sparse representations that necessitate operations on smaller design matrices (in particular when utilizing the linear array framework \citepsup{BrockhausGreven2015}). By contrast, in comparably sparse irregular scenarios, expanding single observed $y_i$ in a basis in a first step might involve unwanted pre-smoothing. To give a consistent presentation on the original level of curves, we rely on an evaluation based approach in all applications presented in the main manuscript.

\new{
	\subsection{Functional Principal Component Representation}
	\label{sec::FPC}
	
	Various approaches in the literature \citepsup[e.g.,][]{MullerYao2008FAM, Scheipl2015, cederbaum2015functional, Volkmann2020MultiFAMM} have employed 
	functional principal component (FPC) basis representations for modeling functional responses in regression models.
	In combination with covariance smoothing \citepsup[e.g.,][]{Yao:etal:2005, cederbaum2018fast} 
	this can be particularly useful in sparse/irregular scenarios, allowing to estimate the functional covariance structure from single curve evaluations. 
	In fact, two variants of corresponding approaches directly fit into our proposed framework, either a) representing  curves using predicted FPC scores or b) estimating inner products based on the covariance structure. 
	In the following, we outline both approaches and briefly discuss related perspectives beyond the scope of this paper. 
	
	Prediction of FPC scores and inner products are carried out along the lines of \citetsup{Yao:etal:2005} and, in the complex case, \citetsup{stoecker2022efp}. 
	Assume we have given (an estimate of) the complex covariance surface $C(s,t) = \mathbb{E}\left(Y^\dagger(s) Y(t)\right)$ of the process $Y$ generating the curve samples $y_1, \dots, y_n$ in the data, with point-wise mean $\mathbb{E}(Y(t)) = 0$ for all $t\in\mathcal{T}$ without loss of generality in the following. Under standard assumptions, this yields
	a (truncated) FPC basis $b_0^{(r)}: \mathcal{T} \rightarrow \mathbb{C}$, $r = 1,\dots,m_0$ with respective eigenvalues $\lambda_1 \geq \dots \geq \lambda_{m_0} \geq 0$.
	Observing only evaluation vectors $\mathbf{y}_i = (y_{i1}, \dots, y_{ik_i})^\top = (y_i(t_{i1}) + \epsilon_{i1}, \dots, y_i(t_{ik_i}) + \epsilon_{ik_i})^\top$ at time-points $t_{i1}, \dots, t_{ik_i} \in \mathcal{T}$ subject to some iid. white noise measurement errors $\epsilon_{i1},\dots, \epsilon_{ik_i} \sim N(0, \sigma^2)$, predicted FPC score vectors $\bycoef_{i}$, comprising predicted basis coefficients of $y_i$ expanded in the basis $\{b_0^{(r)}\}_r$, can be obtained via the conditional expectations
	\begin{equation}
		\label{FPCscore}
		\bycoef_{i} = \mathbb{E}\left( (\langle b_0^{(1)}, Y \rangle, \dots, \langle b_0^{(m_0)}, Y \rangle )^\top \mid \mathbf{Y}_i + \mathbf{\epsilon}_i = \mathbf{y}_i \right) = \mathbf{\Lambda} \mathbf{B}^\dagger_i \mathbf{\Sigma}_i^{-1} \mathbf{y}_i
	\end{equation}
	under a working normality assumption, with matrices $\mathbf{\Lambda} = \operatorname{diag}(\lambda_1, \dots, \lambda_{m_0})$, $\mathbf{B}_i$ with columns $(b_0^{(r)}(t_{i1}), \dots, b_0^{(r)}(t_{ik_i}))^\top$, $r=1,\dots, m_0$, and $\mathbf{\Sigma}_i$ the covariance matrix of $\mathbf{Y}_i + \mathbf{\epsilon}_i = (Y(t_{i1}) + \epsilon_{i1}, \dots, Y(t_{ik_i}) + \epsilon_{i1})^\top$ obtained from corresponding evaluations of $C$ plus $\sigma^2$ on the diagonal.
	
	Approach a) directly analyses shapes of the predicted score vectors $\bycoef_{1}, \dots, \bycoef_{n}$ as described in Section \ref{sec:coefficientlevel} of the supplementary material.
	
	Approach b) uses \eqref{FPCscore} to motivate integration weights $\bW_i = \mathbf{\Sigma}_i^{-1} \mathbf{B}_i \mathbf{\Lambda} \mathbf{B}^\dagger_i \mathbf{\Sigma}_i^{-1}$ for the empirical inner products $\langle y_i, y'_i \rangle_i = \by_i^\dagger \bW_i \mathbf{y}'_i$, $i=1, \dots, n$, introduced in Section \ref{chap_diffgeo} of the main manuscript for $\by_i, \by'_i \in \mathbb{C}^{k_i}$, such that with this choice
	\begin{equation*}
		\langle y_i, y'_i \rangle_i = \by_i^\dagger \bW_i \mathbf{y}'_i = \mathbb{E}\left( \langle Y, Y' \rangle \mid \mathbf{Y}_i + \mathbf{\epsilon}_i = \mathbf{y}_i, \mathbf{Y}'_i + \mathbf{\epsilon}'_i = \mathbf{y}'_i \right)
	\end{equation*}
	for an independent copy $Y'$ of $Y$, with $\mathbf{Y}'_i$ and $\mathbf{\epsilon}'_i$ defined as $\mathbf{Y}_i$ and $\mathbf{\epsilon}_i$. 
	Approach b) might be refined by approximating $\|y_i\|^2$ with $\mathbb{E}\left( \langle Y, Y \rangle \mid \mathbf{Y}_i + \mathbf{\epsilon}_i = \mathbf{y}_i\right)$ and $\langle b^*, y_i \rangle$ with $\mathbb{E}\left( \langle b^*, Y \rangle \mid \mathbf{Y}_i + \mathbf{\epsilon}_i = \mathbf{y}_i\right)$ for a known function $b^*:\mathcal{T} \rightarrow \mathbf{C}$ as described by \citetsup{stoecker2022efp}, which is slightly different from $\langle y_i, y_i \rangle_i$ and $\langle b^*, y_i \rangle_i$, respectively. However, basing all computations on $\langle y_i, y'_i \rangle_i$ as described in the main manuscript, holds the advantage of a unified definition of the shape geometry on evaluation vectors and curves.
	
	Both a) and b) rely, however, on the covariance $C(s,t)$ of the process $Y$ underlying the realizations $y_i$, 
	while we ultimately analyze shapes/forms $[y_i]$, $i=1,\dots, n$, presenting equivalence classes. In practice, this might in many cases not be a problem, when the $y_i$ are in fact roughly aligned and not as arbitrarily recorded as they might be in theory. 
	However in general, it renders FPC based approaches for such settings more complicated and beyond the scope of this work \citepsup[compare][for related work in a different non-regression setting]{stoecker2022efp}. 
	Carrying out the FPC alternatively on tangent space level (i.e. in a linear space) would require computation of $\Log_{[p]}([y_i])$ at some shape/form $[p]$ involving already computation/prediction of inner products.
	
	We leave such considerations to future research, and focus instead on simpler weight matrices $\mathbf{W}_i$ which are known to also work reasonably well in regression scenarios with sparsely/irregularly sampled functional response \citepsup{Scheipl2015, Scheipletal2016, BrockhausGreven2015, brockhaus2017boosting, rugamer2018boosting, Stoecker2019FResponseLSS}.
}

\new{
\subsection{Tensor-product structure in non-parametric regression}
\label{sec:kernelTPfactorization}

\renewcommand{\k}{\mathcal{K}}
\newcommand{\p}{\hat{P}}

We illustrate the broad applicability of the proposed TP factorization (Section 3.2) for the example of \emph{Additive Regression with Hilbertian Responses} proposed by \citet{jeon2020additiveHilbertian} showing that also approaches avoiding (finite-dimensional) basis representations may lead to the desired form of effect estimates $\hat{h}_j(\mathbf{x})$.
Hence, although they do not consider manifold valued responses, TP factorization can be directly applied to visualize and investigate their effect estimates.
We adapt relevant equations to fit our notation and refer for details to their work.

\citet{jeon2020additiveHilbertian} consider regression with an additive predictor $h(\mathbf{x}) = \sum_{j=1}^{J} h_j(x_j)$ with $h_j(x_j)$ depending on the $j$th scalar covariate in $\mathbf{x} = (x_1, \dots, x_J)^\top$.
In Section 2.5 p.~2679, they point out that the estimator $\hat{h}_j(x_j)$ of $h_j(x_j)$ is a linear smoother if the initial estimate of their back-fitting algorithm is (as, e.g., in all their numerical studies). Assuming this in the following, the expression becomes
\begin{equation*}
	\hat{h}_j(x_j) = \frac{1}{n} \sum_{i=1}^n w_{ij}^{[g]}(x_j)\, y_i
\end{equation*}
with weight functions $w_{ij}^{[g]}(x_j)$, $i=1,\dots, n$, $j=1,\dots, J$ after $g$ fitting iterations.
In fact, this immediately has the desired TP form given in Section \ref{sec_tensorproduct_effects}, setting $m = m_j = n$, $\theta_j^{(r,l)} = \frac{1}{n}\ind(r=l)$, $b_j^{(l)} = w_{lj}^{[g]}$ and  $\partial_r = y_r$ for all $l,r = 1,\dots,n$ and $j$. Here, tangent vectors are naturally identified with elements of the Hilbert space, as we are in the linear case.

It might seem odd to have the effect basis functions $b_j^{(i)}$ only implicitly defined depending on the fitting iteration. 
Yet in fact, the $w_{ij}^{[g]}$ are all in the span of 
\begin{equation*}
	b_j^{(i)}(x_j) = \frac{\k_j(x_j, x_{ij})}{\sum_{i=1}^n \k_j(x_j, x_{ij})}, \quad i = 1,\dots,n
\end{equation*}
with some kernels $\k_j$ evaluated around covariate realizations $\mathbf{x}_i = (x_{i1}, \dots, x_{iJ})^\top$, $i=1,\dots,n$.
This can be seen by re-writing the definition of $w_{ij}^{[g]}$ \citep[Sec. 2.5, p.~2679]{jeon2020additiveHilbertian}:
\begin{align*}
	w_{ij}^{[g]}(x_j) &= \frac{\k_j(x_j, x_{ij})}{\p_j(x_j)} - 1 - \sum_{\jmath \neq j} \int_{0}^{1} w_{i\jmath}^{[g-\ind(\jmath \geq j)]}(x_\jmath) \frac{\p_{j\jmath}(x_j, x_\jmath)}{\p_{j}(x_j)} \, dx_\jmath \\
	&= \frac{\k_j(x_j, x_{ij})}{\p_j(x_j)} - 1 - \sum_{l =1}^{n} \sum_{\jmath \neq j} \int_{0}^{1} w_{i\jmath}^{[g-\ind(\jmath \geq j)]}(x_\jmath) \frac{\k_j(x_j, x_{lj}) \k_\jmath(x_\jmath, x_{l\jmath})}{\p_{j}(x_j)} \, dx_\jmath \\
	&= \underbrace{\frac{\k_j(x_j, x_{ij})}{\p_j(x_j)}}_{=\, n\, b^{(i)}_j(x_j)} - 1 - \sum_{l =1}^{n} \frac{\k_j(x_j, x_{lj})}{\p_{j}(x_j)} \underbrace{\sum_{\jmath \neq j} \int_{0}^{1} w_{i\jmath}^{[g-\ind(\jmath \geq j)]}(x_\jmath) \k_\jmath(x_\jmath, x_{l\jmath}) \, dx_\jmath}_{=:\, a_{lj}^{[g]} \in \mathbb{R} \text{ or $\mathbb{C}$, respectively}}, \\
	&= n \sum_{l=1}^{n} (\ind(l=i) - \frac{1}{n} - a_{lj}^{[g]})\, b^{(l)}_j(x_j),
\end{align*}
where by definition
\begin{equation*}
	\p_j(x_j) = \frac{1}{n} \sum_{i=1}^n \k_j(x_j, x_{ij}), \quad \p_{j\jmath}(x_j, x_\jmath) = \frac{1}{n} \sum_{i=1}^n \k_j(x_j, x_{ij}) \k_\jmath(x_\jmath, x_{i\jmath}).
\end{equation*}
and 
by construction $1 \equiv \sum_{i=1}^{n} b_j^{(i)}(x_j)$.
(Starting values for the back-fitting algorithm presented in the paper
are given simply by $w_{ij}^{[0]} = 0$ or the Nadaraya-Watson-type estimator $w_{ij}^{[0]} = \frac{1}{n} \sum_{i=1}^{n} (\frac{\k_j(x_j, x_{ij})}{\p_j(x_j)} -1) \, y_i$.)\\
Consequently, also this non-parametric approach leads to the TP effect structure
\begin{equation*}
	\hat{h}_j(x_j) = \sum_{r=1}^n\sum_{l=1}^n \hat{\theta}_j^{(r,l)} \underbrace{\frac{\k_j(x_j, x_{lj})}{\sum_{i=1}^n \k(x_j, x_{ij})}}_{=b_j^{(l)}(x_j)} \, \underbrace{y_r}_{\partial_r}
\end{equation*}
with $\hat{\theta}_j^{(r,l)} = \ind(l=r) -\frac{1}{n} - a_{lj}^{[g]}$.
}
\bibliographystylesup{chicago} 
\bibliographysup{literatur}

\end{document}